\documentclass[pra, reprint, superscriptaddress, amsmath,amssymb, aps]{revtex4-2}

\usepackage{amsthm}
\usepackage{thmtools}
\declaretheorem{theorem}
\declaretheorem{corollary}
\declaretheorem{lemma}
\declaretheorem{definition}
\declaretheorem[style=definition]{remark}
\usepackage{tcolorbox}
\usepackage{esint}
\usepackage{bbold}
\usepackage{amssymb}
\usepackage{graphicx}
\usepackage{textcomp}
\usepackage[margin=0.8in]{geometry}
\usepackage{relsize}
\usepackage{url}
\usepackage[colorlinks=true, allcolors=camblue]{hyperref}
\usepackage{blkarray}
\usepackage{geometry}
\usepackage{cancel}
\usepackage{tikz}
\usepackage{natbib}
\usepackage{stmaryrd}
\usepackage[utf8]{inputenc}
\usepackage[T1]{fontenc}
\usepackage[caption=false]{subfig}
\usetikzlibrary{quantikz2}
\usepackage{multirow}
\usepackage{mathtools}
\mathtoolsset{showonlyrefs}

\usepackage{ytableau}
\ytableausetup{aligntableaux=center}

\definecolor{camblue}{cmyk}{0.2443, 0.0000, 0.1250, 0.3098}

\usepackage{xcolor}
\usepackage{lipsum}

\newcommand{\vtx}[4]{%
\begin{tikzpicture}[scale=1.0,baseline={(0,0)}]
  \draw (-0.9,0)--(0.9,0);
  \draw (0,-0.9)--(0,0.9);
  \fill (0,0) circle (2pt);

  \ifstrequal{#1}{out}{\path[tips, ->, line width=0.8pt] (0,0)--(0,0.65);}
                     {\path[tips, ->, line width=0.8pt]  (0,0.9)--(0,0.5);}
  \ifstrequal{#2}{out}{\path[tips, ->, line width=0.8pt]  (0,0)--(0,-0.65);}
                     {\path[tips, ->, line width=0.8pt]  (0,-0.9)--(0,-0.5);}
  \ifstrequal{#3}{out}{\path[tips, ->, line width=0.8pt]  (0,0)--(-0.65,0);}
                     {\path[tips, ->, line width=0.8pt]  (-0.9,0)--(-0.5,0);}
  \ifstrequal{#4}{out}{\path[tips, ->, line width=0.8pt]  (0,0)--(0.65,0);}
                     {\path[tips, ->, line width=0.8pt]  (0.9,0)--(0.5,0);}
\end{tikzpicture}%
}

\tikzset{
  cup_gate/.style={
    draw=none, 
    inner sep=0.5pt,
    minimum width=0.5em,
    minimum height=3em,
    path picture={
      \coordinate (NW) at (path picture bounding box.north west);
      \coordinate (NE) at (path picture bounding box.north east);
      \coordinate (SW) at (path picture bounding box.south west);
      \coordinate (SE) at (path picture bounding box.south east);
      \path let \p1 = (NW), \p2 = (NE), \p3 = (SW), \p4 = (SE) in
        coordinate (nwl) at ($(NW)$)
        coordinate (nel) at ($(NE)$)
        coordinate (swl) at ($(SW)$)
        coordinate (sel) at ($(SE)$);
      \draw[line width=1.5pt] (nwl) -- (swl) -- (sel) -- (nel);
    }
  }
}

\tikzset{
  cap_gate/.style={
    draw=none, 
    inner sep=0.5pt,
    minimum width=0.5em,
    minimum height=3em,
    path picture={
      \coordinate (NW) at (path picture bounding box.north west);
      \coordinate (NE) at (path picture bounding box.north east);
      \coordinate (SW) at (path picture bounding box.south west);
      \coordinate (SE) at (path picture bounding box.south east);
      \path let \p1 = (NW), \p2 = (NE), \p3 = (SW), \p4 = (SE) in
        coordinate (nwl) at ($(NW)$)
        coordinate (nel) at ($(NE)$)
        coordinate (swl) at ($(SW)$)
        coordinate (sel) at ($(SE)$);
      \draw[line width=1.5pt] (swl) -- (nwl) -- (nel) -- (sel);
    }
  }
}

\begin{document}

\title{The Power of Power-of-SWAP: Postselected Quantum Computation with the Exchange Interaction}
\author{Jędrzej Burkat}
\email{jbb55@cam.ac.uk}
\affiliation{Cavendish Laboratory, Department of Physics, University of Cambridge, CB3 0HE, UK}
\author{Sergii Strelchuk}
\affiliation{Department of Computer Science, University of Oxford, OX1 3QD, UK}
\author{Michał Studziński}
\affiliation{International Centre for Theory of Quantum Technologies, University of Gdańsk, 80-309, Poland}

\begin{abstract}
We introduce Exchange Quantum Polynomial Time (\textsf{XQP}) circuits, which comprise quantum computation using only computational basis SPAM and the isotropic Heisenberg exchange interaction. Structurally, this restricted model captures decoherence-free subspace computation without access to singlet states. We prove that \textsf{XQP}, as well as its constrained family consisting solely of $\sqrt{\mathrm{SWAP}}$ gates, is universal for quantum computation under polynomial overheads in circuit depth and size. This establishes the universality of the $\sqrt{\mathrm{SWAP}}$ gate on its own under a logical encoding of qubits. We further prove that circuits generated by $\sqrt{\mathrm{SWAP}}$ gates are semi-universal, and thus generate $t$-designs for the uniform distribution over $SU(2)$-invariant unitaries. Finally, we study additional properties of the \textsf{XQP} model, including its relation to the six-vertex and Potts model in statistical physics. Our findings suggest that \textsf{XQP} circuits are naturally suited to near-term hardware and provide a simple platform for experimental demonstrations of quantum advantage.
\end{abstract}

\maketitle

\section{Introduction}

Identifying the source of power in quantum computation is an elusive task, made difficult by the fact that `magic' ingredients lifting sub-universal circuits to full universality are not unique, typically depending on the structure of the underlying gate set. Whilst properties such as superposition, interference and entanglement are necessary, they are not always sufficient for quantum advantage. The additional resources required to achieve it can vary from something as simple as a $\mathrm{SWAP}$ operation in matchgate circuits \citep{Jozsa_2008} to the inclusion of magic states (or non-Clifford gates) in the stabilizer formalism \citep{Gottesman_1998}. Removing these resources from a universal family of circuits yields \textit{restricted models} of computation, which sometimes reduces their power to on par with classical machines, and in other cases yields new, intermediate complexity classes which are not universal yet still hard to simulate classically \citep{Bremner_2010, Aaronson_2010}. Restricted models of computation thus form a fascinating area of study in quantum complexity theory, as they can offer new insights and intuition into what makes quantum computers powerful. In addition, they form great candidates for experimental demonstrations of quantum computational supremacy, as they are often structurally simpler and more amenable to near-term implementation than their universal counterparts \citep{Bremner_2017, Hangleiter_2023, Tillmann_2013}. Since almost any set of quantum gates is universal \citep{Deutsch_1995, LLoyd_1995}, discovering interesting restricted circuit families is not an easy task, requiring additional insights from mathematics, computer science, or physics. In this work, we do exactly this by considering a universal model of quantum computation based on the tunable isotropic Heisenberg exchange interaction \citep{DiVincenzo_2000}, and subtracting the ingredient which makes it powerful -- access to state preparation and measurement of singlet states. This yields a new restricted class of quantum circuits, which we call `Exchange Quantum Polynomial Time' ($\mathsf{XQP}$), and subsequently study. 

We begin with a few words on the Heisenberg exchange interaction. In physics language, we take qubits to be spin-$\frac{1}{2}$ particles, with local operators:
\begin{equation}
\vec{S}_i = (S^x_i, S^y_i, S^z_i) = \frac{1}{2}(X_i, Y_i, Z_i).
\end{equation} 

\noindent The computational basis states $|0 \rangle$ and $|1 \rangle$ are identified with the eigenstates of $S^z$ as $|0 \rangle \cong |+1/2 \rangle$ and $|1 \rangle \cong |-1/2 \rangle$. Similarly, the $n$-qubit computational basis is taken to be the joint eigenbasis of the commuting set of operators $\{ S^z_1, \dots, S^z_n \}$. State initialisation and readout therefore correspond to local preparation of $S^z$ eigenstates and local measurement of spin in the $z$-direction. The initialised state is set to evolve under the isotropic Heisenberg Hamiltonian:
\begin{equation} \label{eq:heisenberg}
    H = \sum_{i < j} J_{ij} \vec{S}_i \cdot \vec{S}_j,
\end{equation}

\noindent where $J_{ij}$ is the tunable exchange coupling between spins $i$ and $j$. Using the identity:
\begin{equation}
    \vec{S}_i \cdot \vec{S}_j = \frac{1}{4}(X_i X_j + Y_i Y_j + Z_i Z_j) = \frac{1}{2}E_{ij} - \frac{1}{4}\mathbf{1},
\end{equation}

\noindent where $E_{ij}$ is the $\mathrm{SWAP}_{ij}$ operation, this may be re-written as:
\begin{equation}
    H = \frac{1}{2} \sum_{i < j} J_{ij} E_{ij} - \frac{1}{4} \left( \sum_{i < j} J_{ij} \right) \mathbf{1}.
\end{equation}

\noindent As the second term contributes a global phase under time evolution, the dynamics are equivalently generated by the transpositions $E_{ij}$. We assume that the couplings are switched on for controlled durations $\delta t$, one qubit pair $(i,j)$ at a time. Then, the evolution of the system can be described as a sequence of `pulses' $U_{ij}(\theta) = e^{i \theta E_{ij}}$ acting on qubits, where $\theta = J_{ij} \delta t / 2$ are referred to as `pulse angles'. As a matrix:
\begin{align}
\label{eq:Utheta-def}
U(\theta)
&:= \cos \theta \cdot \mathbf{1} + i \sin \theta \cdot \mathrm{SWAP} \\
&= \begin{pmatrix}
e^{i\theta} & 0 & 0 & 0\\
0 & \cos\theta & i\sin\theta & 0\\
0 & i\sin\theta & \cos\theta & 0\\
0 & 0 & 0 & e^{i\theta}
\end{pmatrix}.
\end{align}

\noindent We refer to the unitaries $U(\theta)$ as `exchange gates'. Sometimes, they are also referred to as the `partial $\mathrm{SWAP}$', or `power-of-$\mathrm{SWAP}$' gates \cite{yordanov_2020}, owing to the fact that $ \theta \in [0 , \pi/2 ]$ continuously interpolates between the identity and the $\mathrm{SWAP}$ operation. Indeed, $U(0)=\mathbf{1}$, $U(\pi/4)=\sqrt{\mathrm{SWAP}}$, and $U(\pi/2)=i \cdot \mathrm{SWAP}$.

The tunable Heisenberg exchange interaction first came to prominence as a candidate for a physical implementation of quantum computation following the proposal of DiVincenzo et al. \citep{DiVincenzo_2000}, building on the seminal work of Kempe, Bacon, Lidar and Whaley \citep{Kempe_2001}. The authors showed that by encoding logical qubits into \textit{decoherence-free subspaces} (or \textit{subsystems}) of the physical Hilbert space, it is possible to implement a universal set of quantum gates using only evolution by the Heisenberg Hamiltonian. In representation theory terms, this amounts to treating certain basis states of the irreducible representation (irrep) spaces of the symmetric group as logical qubits, and using unitaries generated by the Lie algebra of transpositions $E_{ij}$ to approximate any unitary on the logical space. We formalise this model of computation as $\mathsf{DFS}_k$ (where $k$ refers to the number of physical qubits used to encode a single logical qubit) in Definition \ref{def:dfs}. In this sense, the Heisenberg exchange satisfies \textit{encoded} universality on an appropriately chosen state preparation and measurement basis \citep{Kempe_2001b}. In addition, $\mathsf{DFS}_k$ computation is robust to collective decoherence, forming a `passive' error correcting code against any noise of the form $U^{\otimes n}$ on the physical qubits, in which syndrome extraction is not required \citep{Bacon_2000, Lidar_2001a, Lidar_2001b}. This property is used in the context of quantum error correction, through hybrid QEC-DFS codes which combine the passive and active noise mitigation of both to protect against a wider range of noise \citep{Lidar_1999, Wu_2002, Dash_2024, Dasu_2025, Kasatkin_2026}.

Experimentally, $\mathsf{DFS}_3$ computation on quantum dots was demonstrated in \citep{Laird_2010, Weinstein_2023, Abraham_2026}. Other experimental implementations of a coherent, tunable Heisenberg exchange also include \citep{He_2019, Nowack_2011, Petta_2005, Medford_2013, Eng_2015, Fogarty_2018, van_Diepen_2021, Madzik_2025}. The ease of implementing the exchange interaction on existing quantum hardware is thus a major motivation for our work, making a strong case for the practical relevance and applicability of exchange-based computational models. The other is the observation that although universal quantum computation can be performed with the Heisenberg exchange, it will \textit{always} require an external source of singlet states (for logical qubit encoding), or additional gates (for example the single-qubit $S = \text{diag}(1, i)$ gate) to prepare and measure singlets. This raises the central question: \textit{given a restricted quantum computer which can only prepare and measure states in a local basis and implement the isotropic Heisenberg exchange, how non-classical is its computational power? In addition, if the gate set is further restricted to just a single, non-classical $\sqrt{\mathrm{SWAP}}$ operation, is that enough to achieve a provable quantum advantage?}

We answer these questions as follows. In Section \ref{sec:definitions}, we define the key concepts referred to throughout our work, including exchange quantum polynomial time circuits. In Section \ref{sec:gadget}, we present a simple postselected phase gadget $\mathcal{G}(\theta)$, which implements a single-qubit phase gate $e^{i \theta Z}$ on its target using only the exchange gates $U(\pm \theta)$ and computational basis SPAM on two ancillary qubits. With access to small pulse angles, this gadget admits a probability boosting technique, wherein the phase gate $e^{i \theta Z}$ can be implemented with an overall success probability $1 - \epsilon$ using only $\mathcal{O}(1/\epsilon)$ applications of exchange gates and ancillae. This could be leveraged experimentally, for example to generate $S = \mathrm{diag}(1, i)$ gates to prepare singlet states with a high success probability. Using our gadget, in Section \ref{sec:multiplicative} we show that $\mathsf{PostXQP} = \mathsf{PostBQP}$ (Theorem \ref{thm:postxqpbqp}), i.e. that postselected $\mathsf{XQP}$ circuits are universal for quantum computation, much like in the case of $\mathsf{IQP}$ or $\mathsf{BosonSampling}$ \citep{Bremner_2010,Aaronson_2010, Hamilton_2017}. This implies that efficient weak simulation of $\mathsf{XQP}$ circuits to multiplicative precision would collapse the polynomial hierarchy to its third level, a statement widely believed to be highly unlikely in complexity theory. Remarkably, this result holds for $\mathsf{XQP}(\pi/4)$, i.e. $\mathsf{XQP}$ circuits in which exchange gates are limited to integer powers of $\sqrt{\mathrm{SWAP}}$. We strengthen this result in Section \ref{sec:additive}; by exploiting the probability boosting procedure of the phase gadgets $\mathcal{G}(\theta)$, we show that arbitrary quantum circuits can be approximated by polynomially-sized \textsf{XQP} circuits with small pulse angles, up to a small probability of failure. By making the failure probability sufficiently small, this allows us to ignore postselection altogether, and thus solve arbitrary $\mathsf{BQP}$ problems using $\mathsf{XQP}$ circuits coupled with ancillae and probability-boosted gadgets (Corollary \ref{cor:xqpequalsbqp}). Notably, this result requires access to unphysical infinitesimally small pulse angles -- an assumption we relax in Section \ref{sec:semi_universality}, where we focus on the $\mathsf{XQP}(\pi/4)$ model (restricted to $\sqrt{\mathrm{SWAP}}$) more closely: Theorem \ref{thm:semiuniversal} proves that $\sqrt{\mathrm{SWAP}}$-generated circuits are semi-universal \citep{hulse2024} on any system size, previously only known for groups composed of arbitrary exchange gates $U(\theta)$. This implies that $\sqrt{\mathrm{SWAP}}$ gates alone can approximate any $SU(V_J)$ unitary on the irrep spaces $V_J$ of the symmetric group, i.e. extends the result of \cite{Kempe_2001} to $\sqrt{\mathrm{SWAP}}$ alone. Finally, using a construction of \cite{Marvian_2024} we show that with the addition of two ancillae, $\sqrt{\mathrm{SWAP}}$ circuits can approximate any $U(\theta)$ gate up to a global phase and arbitrary precision (Lemma \ref{lem:sqrtswap_approx}). Combined with the results of Section \ref{sec:additive}, this culminates in Theorem \ref{thm:sqrtswap_universal}: $\mathsf{XQP}(\pi/4) = \mathsf{XQP} = \mathsf{BQP}$, i.e. circuits composed solely of computational basis SPAM and $\sqrt{\mathrm{SWAP}}$ gates are universal for quantum computation. Section \ref{sec:discussion} discusses the implications of our results, as well as the experimental feasibility of implementing $\mathsf{XQP}$ circuits on near-term hardware, and additional material in the appendices. We end with Section \ref{sec:open}: a selection of open problems regarding the $\mathsf{XQP}$ model for future work. 

In addition to studying the complexity-theoretic properties of our new \textsf{XQP} model in the main text, we describe its further properties in the appendices. Appendix \ref{appendix:alternativeproof} contains an alternative proof of Theorem \ref{thm:postxqpbqp}. Appendix \ref{sec:schur_states} relates the computational basis inputs of $\mathsf{XQP}$ to the Schur basis inputs of $\mathsf{DFS}_k$, culminating in a proof that the ability to additively estimate $\mathsf{XQP}$ circuit amplitudes allows for the additive estimation of $\mathsf{DFS}_k$ amplitudes, and thus any encoded $\mathsf{BQP}$ amplitude (Theorem \ref{thm:amplitudes}). This forms an additional piece of evidence for the additive-approximation hardness of $\mathsf{XQP}$ which does not rely on the postselection gadgets used in our main results. In Appendix \ref{sec:entangling_power}, we study the entangling power and linear entropy of $U(\theta)$ gates, showing that the $\sqrt{\mathrm{SWAP}}$ gate maximises both. Finally, in Appendix \ref{sec:potts} we relate $\mathsf{XQP}$ circuits to the six-vertex model of statistical mechanics, identifying non-trivial instances of efficiently computable $\mathsf{XQP}$ probability amplitudes. Conversely, this connection lets us describe the worst-case hardness of computing partition functions of related six-vertex models by means of our results on \textsf{XQP}. By further relating the six-vertex model to the Potts model, we find that the addition of $S$ gates to $\mathsf{XQP}$ circuits corresponds to a modification of the underlying six-vertex model to encode the partition function of a complex $q=4$ Potts model \citep{Baxter_1982}. 

Throughout our work, we assume basic familiarity with properties of decoherence-free subspaces and subsystems \citep{Kempe_2001}, making use of the $k=3$ and $k=4$ encodings. For more details on the irrep structure of the former, we recommend the discussion in \citep{Kempe_2001b, van_Meter_2019}. For properties of the latter, see \citep{Woodworth_2006}. In addition to the references above, \citep{Fong_2011, Setiawan_2014, Zeuch_2014, Lidar_1998, Hsieh_2003, West_2010} contain helpful discussion of decoherence-free subspaces, as well as various methods of implementing logical two-qubit gates in practice. For a review-style exposition, see references \citep{Kempe_2001, Bacon_1999, Lidar_2003, Lidar_2014}, as well as Bacon's PhD thesis \citep{Bacon_2003}. For a thorough generalisation of the results of \citep{Kempe_2001} to the case of qudits, see \citep{Van_Meter_2021}.

As a local Hamiltonian problem, estimating the ground state energy of the isotropic Heisenberg Hamiltonian is $\mathsf{QMA}$-complete \citep{Piddock_2015, Cubitt_2016}. In fact, such Hamiltonians are universal in the sense that any other Hamiltonian evolution can be encoded with an isotropic Heisenberg model on a larger system \citep{Piddock_2018, Cubitt_2018}. This means that in principle, any $\mathsf{BQP}$ circuit can be simulated through a carefully engineered evolution of the Heisenberg Hamiltonian. However, the encoding presented in the appendix of \citep{Cubitt_2018} crucially relies on encoding logical qubits using singlet states; without $S$ gates or phase gadgets, this falls outside the scope of $\mathsf{XQP}$ circuits.

Finally, in \citep{Aaronson_2016} a related restricted model of computation called $\mathsf{QBall}$ was considered. In their work, the authors defined quantum circuits restricted to gates of the form $\cos \theta \cdot R(e) + i \sin \theta \cdot R((ij))$, with the computational basis consisting of states $|\sigma \rangle$ indexed by permutations $\sigma \in S_n$, and the action of gates given by left regular representation $R((ij)) |\sigma \rangle = |(i j) \circ \sigma \rangle$. The computational power of this model was found to depend on the input states: for example, on any input state $|\sigma \rangle = R(\sigma) |e \rangle$, $\mathsf{QBall}$ amplitudes can be represented as the trace of a unitary, and subsequently that any amplitude $\langle \sigma' | U_{\mathsf{QBall}} | \sigma \rangle$ can be efficiently approximated by a $\mathsf{DQC1}$ circuit. It was also shown (Appendix K of \citep{Aaronson_2016}) that the state:

\begin{equation}
    |\phi_0 \rangle = \frac{1}{\sqrt{k!(n-k)!}} \sum_{\sigma \in S_{k, n-k}} R(\sigma) |e \rangle
\end{equation}

\noindent encodes the computational basis state $|0^{ k} 1^{ n-k} \rangle$, and similarly $R((ij))| \phi_0 \rangle$ encodes $E_{ij} |0^{ k} 1^{ n-k} \rangle$. Finally, measurement in the computational basis was shown to be implementable by measurement in the $|\sigma \rangle$ basis. This implies that $\mathsf{QBall}$ circuits initialised to $|\phi_0 \rangle$ can simulate $\mathsf{XQP}$ circuits, and our results for $\mathsf{XQP}$ (in particular, Theorem \ref{thm:sqrtswap_universal}) imply that on inputs of the form $R(\sigma) |\phi_0 \rangle$, $\mathsf{QBall}$ circuits can solve problems in $\mathsf{BQP}$. This partially resolves the open problem of classifying the computational power of $\mathsf{QBall}$ circuits on various input states, which was raised at the end of \citep{Aaronson_2016}.

\section{Definitions} \label{sec:definitions}

We begin by formalising the notion of computation on decoherence-free subspaces introduced in \citep{Kempe_2001}:

\begin{definition}[Decoherence-Free Subspace Computation \citep{Kempe_2001} ] \label{def:dfs}
    A $\mathsf{DFS}_k$ circuit (where $k = 3, 4$) on $kn$ qubits consists of the following:
    \begin{enumerate}
        \item Preparation of a state $| 0_L \rangle^{\otimes n}$, where:
        \begin{equation}
        |0_L \rangle = \begin{cases}
        |S \rangle \otimes |0 \rangle, & k = 3, \\
        |S \rangle \otimes |S \rangle, & k = 4,

        \end{cases}
        \end{equation}
         $|S \rangle = \frac{1}{\sqrt{2}}(|01 \rangle - |10 \rangle)$ is the singlet state, and:
         \begin{equation}
         Z_L = -E_{12}, \quad X_L = \frac{1}{\sqrt{3}}(E_{13} - E_{23}).
         \end{equation}
        \item Application of $\text{poly}(n)$ Heisenberg exchange gates $e^{i H t}$, where $H$ is a linear combination of exchange interactions and their (possibly nested) commutators. In other words, $H$ is an element of the Lie algebra generated by transpositions $E_{ij}$:
        \begin{equation}
            H = \sum_{i, j} \alpha_{ij} E_{ij} - i \sum_{i, j, k, l} \beta_{ijkl} [E_{ij}, E_{kl}] + \ldots
        \end{equation}
        \item Projective measurements $\{ |S \rangle \langle S|, \mathbf{1} - |S \rangle \langle S| \}$ on pairs of qubits, or measurements in the computational ($Z^{\otimes kn}$) basis.
    \end{enumerate}
\end{definition}

\noindent Our definition of evolution in $\mathsf{DFS}_k$ circuits is lenient, as often only pulses $e^{i \theta E_{ij}}$ on single qubit pairs can be switched on at a time. Sometimes, gates of the form $e^{i \alpha_{ij} E_{ij} + i \alpha_{kl} E_{kl}}$ may also be directly available by switching on two interactions simultaneously \citep{Kempe_2001}. However, in practice Hamiltonians in the full Lie algebra can be generated by short sequences of exchange gates using the Lie--Trotter formulae \cite{Marvian_2022}:

\begin{align}
&e^{i A \alpha \delta t} e^{ i  B \beta \delta t} = e^{i( \alpha A + \beta B) \delta t} + \mathcal{O}(\delta t^2), \\
& e^{i A \delta t} e^{i B \delta t} e^{-i A \delta t} e^{-i B \delta t} = e^{-[A, B] \delta t^2} + \mathcal{O}(\delta t^3),
\end{align}

\noindent leading to an $\mathcal{O}(1 / \text{poly}(n))$ error in the approximation of the desired Hamiltonian evolution. We now define the class of exchange quantum polynomial time circuits, which is our main object of study. They can be seen as a restriction of $\mathsf{DFS}_k$ circuits to computational basis state preparation and measurement:

\begin{definition}[Exchange Quantum Polynomial Time Circuits] \label{def:xqp}
    An $\mathsf{XQP}$ circuit on $n$ qubits consists of the following:
    \begin{enumerate}
        \item Preparation of a state $| \mathbf{x} \rangle$ in the computational ($Z^{\otimes n}$) basis.
        \item Application of $\text{poly}(n)$ Heisenberg exchange gates $U_{ij}(\theta)$, where:
        \begin{equation}
            U_{ij}(\theta) = e^{i\theta E_{ij}} = \cos \theta \cdot \mathbf{1} + i \sin \theta \cdot E_{ij},
        \end{equation}
        and $E_{ij} = (X_i X_j + Y_i Y_j + Z_i Z_j + \mathbf{1})/2$ is a $\mathrm{SWAP}$ operation between qubits $i$ and $j$. 
        \item Measurement in the computational basis.
    \end{enumerate}
\end{definition}

\noindent Such circuits are efficient to describe by listing the input bitstring $\mathbf{x}$, and the parameters of each $U_{ij}(\theta)$ gate it is acted on. The choice of the bitstring states in the definition is arbitrary, and we are free to work with any orthonormal basis $| \alpha \rangle = V|0 \rangle$ and $|\beta \rangle = V |1 \rangle$, so long as the change of basis is global. This is because all transpositions $E_{ij}$ commute with $V^{\otimes n}$, therefore $\mathsf{XQP}$ circuits $C$ satisfy ${V^\dagger}^{{\otimes n}} C V^{\otimes n} = C$. In practice, this makes executing such circuits extremely simple on hardware; all that is required is the ability to initialise and measure each qubit in a common orthogonal basis, and apply the tunable Heisenberg exchange.

Many of our results regarding $\mathsf{XQP}$ (and $\mathsf{DFS}_k$) circuits will also apply when pulse angles are restricted to a finite set of values. This aims to model realistic scenarios, where the Heisenberg interaction can only be switched on for some minimum, finite time $\delta t_\mathrm{min}$. To denote this, we define the following restricted circuit families:

\begin{definition} An $\mathsf{XQP}(\theta)$ $(\mathsf{DFS}_k(\theta))$ circuit is an $\mathsf{XQP}$ $(\mathsf{DFS}_k)$ circuit for which all Heisenberg exchange gates are of the form $U_{ij}(\theta)$ for some fixed $\theta$. Equivalently, $\mathsf{XQP}(\theta)$ $(\mathsf{DFS}_k(\theta))$ consists of $\mathsf{XQP}$ $(\mathsf{DFS}_k)$ circuits with pulse angles constrained to $m \theta$, for fixed $\theta$ and $m \in \mathbb{N}$.
\end{definition}

\noindent If $\theta / \pi \notin \mathbb{Q}$, then $\mathsf{XQP}(\theta)$ is dense in $\mathsf{XQP}$ and can approximate any $\mathsf{XQP}$ circuit to arbitrary precision. On the other hand, if $\theta$ is a rational multiple of $\pi$ then $\mathsf{XQP}(\theta)$ can still form an infinitely large subset of $\mathsf{XQP}$. We will pay particular attention to $\mathsf{XQP}(\pi/4)$, which consists of circuits restricted to integer powers of $U(\pi/4) = \sqrt{\mathrm{SWAP}}$. We find that such circuits generate an infinitely large (Theorem \ref{thm:semiuniversal}), yet strict (Theorem \ref{thm:xqpsubset}) subset of $\mathsf{XQP}$, whose computational power is nonetheless universal (Theorem \ref{thm:sqrtswap_universal}). To prove many of our results regarding this class, in Sections \ref{sec:additive} and \ref{sec:semi_universality} we make use of the following Lemma regarding $\mathsf{DFS_3(\pi/4)}$ circuits:

\begin{restatable}{lemma}{dfsthree}\label{lemma:dfs3universal}
$\mathsf{DFS}_3(\pi/4) = \mathsf{BQP}$.
\end{restatable}

\begin{proof}
This is a standalone result proven in Section \ref{sec:semi_universality}. As the proof is self-contained and much of Section \ref{sec:semi_universality} builds on its techniques, we defer it until then. 
\end{proof}

In Section \ref{sec:multiplicative} we also make use of the complexity class $\mathsf{PostXQP}$, consisting of $\mathsf{XQP}$ computations conditional on particular measurements of a subset of the output register. Here, we adopt the definition of \citep{Bremner_2010}:

\begin{definition}[Postselected Exchange Quantum Polynomial Time]

A language $L$ is in the class $\mathsf{PostXQP}$ iff there is an error tolerance $0 < \varepsilon < \frac{1}{2}$ and a uniform family $\{C_w\}$ of postselected $\mathsf{XQP}$ circuits with a specified single qubit output register $\mathcal{O}_w$ (for the $L$-membership decision problem) and a specified postselection register $\mathcal{P}_w$ such that:
\begin{enumerate}
    \item If $w \in L$, then $p(\mathcal{O}_w = 1 | \mathcal{P}_w = \mathbf{p}) \geq 1 - \varepsilon$.
    \item If $w \notin L$, then $p(\mathcal{O}_w = 0 | \mathcal{P}_w = \mathbf{p}) \geq 1 - \varepsilon$.
\end{enumerate}

\end{definition}

\noindent In the definition, the registers $\mathcal{O}_w$ and $\mathcal{P}_w$ are specified out of the full output of the postselected circuits, meaning that not all qubits must be measured, or have their measurements recorded. In our case this means that we can use logically encoded qubits (using $k$ physical qubits per logical qubit) and represent their logical measurement by the output of a subset of the physical qubit measurements. The successful logical SPAM can also be conditioned on postselections, as the definition does not discriminate between the use of postselection for state preparation, evolution, or measurement. The case where state evolution is conditioned on postselection is discussed in Section \ref{sec:multiplicative}, whereas the case where logical SPAM is conditioned on postselection is discussed in Appendix \ref{appendix:alternativeproof}.

\begin{remark} \label{remark:xqpbqp}
For any qubit number $n$, there exist unitaries $V \in U(2^n)$ which cannot be approximated by sequences of exchange gates $U(\theta)$. 
\end{remark}
\begin{proof}
    The $n$-non-universality of the exchange interaction follows from Lemma 9 of \citep{Childs2010}, where families of non-universal Hamiltonians were identified. This includes the transpositions $E_{ij}$, which generate the exchange gates $U(\theta) = e^{i \theta E_{ij}}$. We emphasise that Remark \ref{remark:xqpbqp} is a statement about synthesis of unitaries, rather than a statement about classes of decision problems solvable by the respective computational models.
    
    We also note that another, related classification was given in \citep{bouland2016}, where the authors gave conditions under which two-qubit commuting Hamiltonians are universal. As the exchange Hamiltonian $E_{ij}$ is non-commuting (since $[E_{12}, E_{23}] \neq 0$), $\mathsf{XQP}$ falls outside of the scope of two-qubit commuting Hamiltonian circuits.
\end{proof}

 In Section \ref{sec:semi_universality} and Appendix \ref{sec:schur_states} we will make reference to the decomposition of the $n$-qubit Hilbert space into irreducible representation spaces of the symmetric group $S_n$. This decomposition is succinctly described by Schur--Weyl duality:

 \begin{definition}[Schur--Weyl Duality] \label{def:schur_weyl}
    Under the joint action of $S_n$ and $SU(2)$, the $n$-qubit Hilbert space decomposes as:

\begin{equation} \label{eq:schur_weyl}
    (\mathbb{C}^2)^{\otimes n} \stackrel{S_n \times SU(2)}{\cong} \bigoplus_{J \in \mathcal{J}_n} V^{S_n}_J \otimes W^{SU(2)}_J,
\end{equation}

\noindent where $V_J$ ($W_J$) are the irrep spaces of $S_n$ ($SU(2)$) labelled by the total angular momentum $J$. Under the action of the symmetric group, $(V_J \otimes W_J)$ form isotypic components labelled by $J$. In other words, $W_J$ forms the multiplicity space of the $S_n$ irrep $V_J$. The set $\mathcal{J}_n$ of allowed angular momenta is given by $\mathcal{J}_n = \{0, 1, \ldots, n/2\}$ for even $n$, and $\mathcal{J}_n = \{1/2, 3/2, \ldots, n/2\}$ for odd $n$. The dimensions of the $S_n$ irrep spaces are given by $\dim V_J = \binom{n}{n/2 - J} - \binom{n}{n/2 - J - 1}$, and the dimensions of the $SU(2)$ irrep spaces are given by $\dim W_J = 2J + 1$.
 \end{definition}

\section{The Postselected Phase Gadget} \label{sec:gadget}

The central tool used throughout this work is the phase gadget $\mathcal{G}(\theta)$, which acts on a single qubit state $|\psi \rangle = \alpha | 0 \rangle + \beta |1 \rangle$, producing a phase gate $\text{diag}(1, e^{-2i \theta})$ upon a successful postselection on two ancillary qubits. The gadget only requires the ability to prepare and measure qubits in the computational basis and apply exchange gates $U(\theta)$. As we will discuss in Section \ref{sec:multiplicative}, with access to just the $U(\pi/4) = \sqrt{\mathrm{SWAP}}$ gate this is sufficient to implement the $S$ gate, and consequently perform postselected universal quantum computation.

\begin{definition}[Postselected Phase Gadget] \label{def:phasegadget}
Let $|\psi\rangle$ be a single-qubit state on register $i$. The postselected phase gadget
$\mathcal{G}(\theta)$ consists of the following steps:
\begin{enumerate}
  \item Append an ancilla state $|01\rangle_{jk}$ on registers $j, k$.
  \item Apply the exchange gates $U_{ij}(\theta)$ and $U_{ik}(-\theta)$.
  \item Measure registers $j, k$ in the computational basis and postselect on $|01\rangle_{jk}$.
\end{enumerate}
This requires two ancillary qubits initialised to $|01\rangle_{jk}$ and access to the exchange gates $U(\pm \theta)$. $\mathcal{G}(\theta)$ succeeds on measurement outcome $|01\rangle_{jk}$, and fails otherwise.
\end{definition}

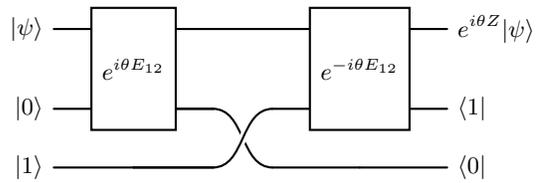
\begin{figure}
    \begin{quantikz} 
    \lstick[1]{\ket{\psi}} & \gate[wires=2]{e^{i \theta E_{12} }} &  & \gate[wires=2]{e^{-i \theta E_{12} }} & \rstick[1]{ $e^{i \theta  Z } | \psi \rangle$ }  \\
   \lstick{\ket{0}} & &\permute{2, 1} & &  \rstick{\bra{1}} \\
   \lstick{\ket{1}} & & & &  \rstick{\bra{0}} 
\end{quantikz}
    \caption{\textit{Quantum Circuit for the Phase Gadget}. Upon successful postselection on the ancilla measurements, the gadget implements a phase gate $e^{i \theta Z }$ on the input state. The success probability is given by $p(\text{success}) = \cos^2 \theta$.}
    \label{fig:phasegadget}
\end{figure}

We now describe the action of $\mathcal{G}(\theta)$. Starting with $| \psi \rangle$ on the $i^\text{th}$ register, appending an ancillary qubit $|0\rangle$ to the $j^\text{th}$ register and applying the exchange interaction $e^{i \theta E_{i j}} = \cos ( \theta) \cdot \mathbf{1} + i \sin (\theta) \cdot E_{i j}$ produces the state:

\begin{equation}
    \alpha e^{i \theta}|00 \rangle_{ij}  +  \ \beta \cos \theta |10 \rangle_{ij} + i \beta \sin \theta |01 \rangle_{ij}.
\end{equation}

\noindent By the principle of deferred measurement, we may postselect immediately on a measurement of $|0 \rangle_j$. The resulting (non-normalised) state is $\alpha e^{i \theta}|0 \rangle_{i} + \beta \cos \theta |1 \rangle_{i}$. We then append a second ancillary qubit $|1 \rangle_{k}$, and apply the exchange gate $e^{- i \theta E_{i k}} = \cos \theta \cdot \mathbf{1} - i \sin \theta \cdot E_{i k}$. This produces the (non-normalised) state:

\begin{equation}
\alpha e^{i \theta} \cos \theta |01 \rangle_{ik} - i \alpha e^{i \theta} \sin \theta |10 \rangle_{ik} + \beta e^{- i \theta} \cos \theta |11 \rangle_{ik}.
\end{equation}

\noindent Postselection on a measurement of $|1 \rangle_k$ then gives:

\begin{equation}
    \alpha e^{i \theta} \cos \theta |0 \rangle_{i}  + \beta e^{- i \theta } \cos \theta |1 \rangle_{i}. \label{eq:phasegadget_output}
\end{equation}

\noindent In normalising the output state, the factors of $\cos \theta$ cancel out and we are left with $ \alpha e^{i \theta } |0 \rangle_i + \beta  e^{-i \theta } |1 \rangle_i$. Overall, a successful action of $\mathcal{G}(\theta)$ implements the equivalent of a phase gate $e^{i \theta Z_i } $ on the single qubit input -- the only scenario where it fails is when $\theta = \pi / 2$, in which case $e^{\frac{ i\pi}{2} E_{ij}} = i E_{i j}$ is a $\mathrm{SWAP}$ gate. In that case one can simply apply the gadget twice, setting $\theta = \pi / 4$ to obtain $\mathcal{G}(\pi/2) = \mathcal{G}(\pi/4) \circ \mathcal{G}(\pi/4)$, equivalent to a $Z$ gate. Our gadget is illustrated in Figure \ref{fig:phasegadget}.

In addition to the theoretical utility of the phase gadget, it can also be shown to be practical -- in the sense that the postselection probability for $\mathcal{G}(\theta)$ can be brought arbitrarily close to $1$ with polynomially large overheads in space and time. This property is crucial for Section \ref{sec:additive}, where we use it to show that classically efficiently sampling from $\mathsf{XQP}$ circuits is unlikely. To boost the gadget probabilities, we assume the ability to implement small pulse angles and segment each $\mathcal{G}(\theta)$ gadget into a sequence of high-likelihood $\mathcal{G}(\theta/N)$ gadgets. The relevant properties of $\mathcal{G}(\theta)$ are summarised in the following Lemma:

\begin{lemma} \label{lemma:phasegadget}
Let $\mathcal{G}(\theta)$ be as in Definition \ref{def:phasegadget}. Then:
\begin{enumerate}
  \item Upon success, $\mathcal{G}(\theta)$ implements $e^{i\theta Z}$ on  $|\psi\rangle$, with probability $p(\mathrm{success}) = \cos^2\theta$.
  \item The successive composition of phase gadgets satisfies $\mathcal{G}(\theta_1) \circ \mathcal{G}(\theta_2) = \mathcal{G}(\theta_1 + \theta_2)$.
  \item With pulse angles $\theta/N$ available, $\mathcal{G}(\theta)$ can be
    implemented with $p(\mathrm{success}) \geq 1 - \epsilon$, where $\epsilon = \mathcal{O}(\theta^2/N)$. This requires $2N$ ancillary qubits without, or $2$ ancillary qubits with, mid-circuit measurements. If  $U(\theta/N)$ and $U(-\theta/N)$ are both directly implementable, this requires $N$ applications of each. If only $U(\theta/N)$ is directly implementable and  $\theta = \pi p/q$ for coprime $p, q$, this requires at most $N^2 q$ applications of $U(\theta/N)$.
\end{enumerate}
\end{lemma}
\begin{proof}

The fact that $\mathcal{G}(\theta)$ implements $e^{i \theta Z}$ can be seen directly from Equation \eqref{eq:phasegadget_output}. Since $e^{i \theta_1 Z} e^{i \theta_2 Z} = e^{i (\theta_1 + \theta_2) Z}$, the composition rule follows directly. From Definition \ref{def:phasegadget}, a successful application of $\mathcal{G}(\theta)$ requires two measurements with probabilities $p(x_j = 0) = 1 - |\beta|^2 \sin^2 \theta$ and $p(x_k = 1 | x_j = 0) = \cos^2 \theta / p(x_j = 0)$. Therefore, $p(\mathrm{success}) = \cos^2 \theta$. 

To prove the third statement we use the composition rule, which allows us to boost the overall success probability close to $1$ for any $\theta$. This is achieved by segmenting the angle into $\theta / N$, and applying the gadget successively to obtain $\mathcal{G}(\theta/N)^N = \mathcal{G}(\theta)$. This gives an overall success probability of $p(\text{success}) = \cos^{2N}(\theta / N)$. For each $\mathcal{G}(\theta / N)$ we require two ancillary qubits, giving $2N$ in total. If we allow mid-circuit measurements, we only need to prepare an initial $2$-qubit ancillary state $|01 \rangle_{jk}$ and postselect on a measurement of $|01 \rangle_{jk}$ at each round. Neither approach requires access to non-exchange gates beyond the initial state preparation and measurements. 

Next, if we have direct access to both $U(\pm \theta / N)$, every intermediate gadget $\mathcal{G}(\theta / N)$ can be implemented with one application of each. Otherwise, assuming that  $\theta = \pi p / q$, we have that $\theta/N = \pi p / Nq$. We then have that $U(\theta/N)^{Nq} = U(\pi p) = (-1)^p \mathbf{1}$, giving $U(\theta/N)^{Nq - 1} = (-1)^p U(-\theta / N)$. Therefore, we can implement $U(-\theta / N)$ up to a global phase with $Nq - 1$ applications of $U(\theta / N)$. It follows that each intermediate $\mathcal{G}(\theta / N)$ can be implemented with $Nq$ applications of $U(\theta / N)$, giving $N^2 q$ applications of $U(\theta / N)$ in total.

Finally, we wish to pick $N$ so that overall, $p(\text{success}) \geq 1 - \epsilon$. Since $\cos (\theta/N) \geq 1 - \theta^2 / 2N^2$, Bernoulli's inequality gives $p(\text{success}) \geq 1 - \theta^2 / N$. Thus, setting $N = \lceil \theta^2 / \epsilon \rceil$ is sufficient to ensure that $p(\text{success}) \geq 1 - \epsilon$, giving $\epsilon = \mathcal{O}(\theta^2/N)$ as required.

\end{proof}

We note that another phase gadget for synthesising phase gates using exchange interactions was proposed in \citep{Wu_2003}. The protocol therein is deterministic, in the sense that with mid-circuit measurements and adaptive gate applications it is possible to repeat steps until success, without risk of destroying the input state and with a success probability which approaches $1$. However, unlike our protocol this requires the ability to form projective measurements in the singlet basis (as well as parity measurements of two qubits), which are not available in the $\mathsf{XQP}$ model. By Lemma \ref{lemma:phasegadget}, our gadget $\mathcal{G}(\theta)$ can also approach a success probability of $1$, albeit with a small risk of destroying the target state on failure.

\section{Postselected \texorpdfstring{$\mathbf{XQP}$}{XQP} Circuits are Universal for Quantum Computation} \label{sec:multiplicative}

We now prove our first main result: postselected $\mathsf{XQP(\pi/4)}$ circuits (and thus postselected $\mathsf{XQP}$ circuits) are universal for quantum computation.

\begin{restatable}{theorem}{postxqpbqp} \label{thm:postxqpbqp}
$\mathsf{PostXQP(\pi/4)} = \mathsf{PostXQP} = \mathsf{PostBQP}$.
\end{restatable}

\noindent It is sufficient to prove the $\mathsf{PostXQP(\pi/4)}$ case only. Obviously, $\mathsf{PostXQP(\pi/4)} \subseteq \mathsf{PostBQP}$. To show that $\mathsf{PostBQP} \subseteq \mathsf{PostXQP(\pi/4)}$, we will show that any postselected $\mathsf{BQP}$ computation can be implemented on a postselected $\mathsf{XQP(\pi/4)}$ circuit with two physical qubits encoding each logical qubit, plus ancillae. In contrast to the $\mathsf{DFS}_k$ encoding using singlet states, we will encode the logical $|0_L \rangle$ and $|1_L \rangle$ states as $|01\rangle$ and $|10 \rangle$ respectively, meaning that logical SPAM can be implemented with the physical SPAM available in $\mathsf{XQP(\pi/4)}$. 

\begin{proof} 

We will make use of just two physical quantum gates: $U(n\pi/4) = (\sqrt{\mathrm{SWAP}})^n$, and the $S = \text{diag}(1, i)$ gate which can be obtained by postselection via $\mathcal{G}(3\pi/4)$ (we can also obtain $S^\dagger = S^3$ directly via $\mathcal{G}(\pi/4)$). The former is given by:

\begin{equation}
    U(\pi/4) = \sqrt{\mathrm{SWAP}} =  \begin{pmatrix} e^{\frac{i \pi}{4}} & 0 & 0 & 0 \\0 & \frac{1}{\sqrt{2}} & \frac{i}{\sqrt{2}} & 0 \\ 0 & \frac{i}{\sqrt{2}} & \frac{1}{\sqrt{2}} & 0 \\ 0 & 0 & 0 & e^{\frac{i \pi}{4}}\end{pmatrix}  .
\end{equation}

\noindent We encode our logical qubits as $|0_L \rangle = |01\rangle, \ |1_L \rangle = |10\rangle$. To obtain our universal set of gates we will first construct $H_L$, an operation which acts as the Hadamard gate on the logical space:
\begin{equation}
    H_L = S^\dagger_i U_{ij}(\pi/4) S_j  = i \times \left( \begin{array}{c | c | c} e^{\frac{-i \pi}{4}} & &  \\ \hline &  \mathlarger{H}  & \\ \hline
    & & e^{\frac{-i \pi}{4}}\end{array} \right).
\end{equation}
Clearly, $U(\pi/2) = iX_L$, and therefore, $H_L^\dagger U(\pi/2) H_L = iZ_L$. We  can also construct the logical phase gate:
\begin{equation}
    S_L = H_L^\dagger U(3\pi/4) H_L  = e^{3i \pi/4} \text{diag}(1, 1, i, 1).
\end{equation}

\noindent Now, construct the two qubit gate:
\begin{align}
    (CS_L)_{ij} &= \mathbf{1}_i \otimes (U(\pi/2)^\dagger S_L U(\pi/2))_{i+1, j} \otimes \mathbf{1}_{j+1} \\
    &= e^{3 i \pi / 4} \times \mathbf{1}_i \otimes \text{diag}(1, i, 1, 1)_{i+1, j} \otimes \mathbf{1}_{j+1}.
\end{align} 

\noindent Up to a global phase, this is a controlled phase gate \citep{Glaudell_2021} on the logical space. Suppose we have two logical qubits encoded in the registers $(i, i+1)$ and $(j, j+1)$. Applying $CS_L$ to qubits $(i+1, j)$ results in the following transformation:
\begin{align}
    |\psi \rangle &= \alpha |00_L \rangle + \beta |01_L \rangle + \gamma |10_L \rangle + \delta |11_L \rangle \\
    &= \alpha |01 01 \rangle + \beta |01 10 \rangle + \gamma |10 01 \rangle + \delta |10 10 \rangle \\
    & \rightarrow  \alpha |01 01 \rangle + \beta |01 10 \rangle + \gamma |10 01 \rangle + i \delta |10 10 \rangle \\
    & = \alpha |00_L \rangle + \beta |01_L \rangle + \gamma |10_L \rangle + i \delta |11_L \rangle.
\end{align}

\noindent We therefore obtain the universal set of logical gates $\{X_L, Z_L, H_L, S_L, CZ_L, CS_L\}$ using the physical gate $U(\pi/4)$, and the ${S, S^\dagger}$ gates obtained via postselection. 

Now, consider any $\mathsf{PostBQP}$ computation: without loss of generality, it will consist of a state and postselection register initialised to $|0\rangle^{\otimes n} |0\rangle^{\otimes m}$, a sequence of quantum gates, a computational basis measurement on the first qubit, and a postselection on $|0\rangle^{\otimes m}$. We can re-express this computation as a postselected $\mathsf{XQP(\pi/4)}$ circuit by initialising the state $|01\rangle^{\otimes n + m + k}$, where we've appended $2k = \mathcal{O}(\text{poly}(n, m))$ phase gadget ancillae. For each gate in the original circuit, synthesise the corresponding logical gate in the $\mathsf{XQP(\pi/4)}$ circuit using $U(\pi/4)$ gates and the phase gadget. For the single qubit measurement in the original circuit, perform the two-qubit measurement of $\{ |01 \rangle, | 10 \rangle \}$ on physical qubits $1, 2$ in the $\mathsf{XQP(\pi/4)}$ circuit. Finally, postselect on the measurement of $ | 01\rangle^{\otimes m + k}$ on the postselection registers to ensure the correct postselection of the original circuit, and the success of the phase gadgets. This shows that $\mathsf{PostBQP} \subseteq \mathsf{PostXQP(\pi/4)}$, which proves that $\mathsf{PostXQP(\pi/4)} = \mathsf{PostBQP}$.
\end{proof}

In Appendix \ref{appendix:alternativeproof} we present an alternative proof of Theorem \ref{thm:postxqpbqp}, wherein a $\mathsf{PostBQP}$ computation on $n$ qubits with $m$ measurements is re-expressed as a $\mathsf{PostXQP(\pi/4)}$ computation on $3n$ qubits with just $n + m$ phase gadget applications. This is achieved by using the sequence $(S \otimes \mathbf{1}) U(\pi/4) |01 \rangle$ to prepare the singlet states $\frac{1}{\sqrt{2}}(|01 \rangle - |10 \rangle)$, which may then be used to re-express the original $\mathsf{PostBQP}$ circuit as a $\mathsf{PostDFS}_3(\pi/4)$ circuit, using the fact that $\mathsf{DFS}_3(\pi/4) = \mathsf{BQP}$ (Lemma \ref{lemma:dfs3universal}). These techniques are of particular relevance for our subsequent results in Sections \ref{sec:additive} and \ref{sec:semi_universality}.

\begin{restatable}{corollary}{xqpsimulation} \label{cor:xqpsimulation}
If the output probability distributions generated by uniform families of $\mathsf{XQP(\pi/4)}$ could be weakly classically simulated to within multiplicative error $1 \leq c < \sqrt{2}$, then $\mathsf{PostBPP = PP}$. This implies that the polynomial hierarchy collapses to its third level, i.e. $\mathsf{PH} = \Delta_3$. 
\end{restatable}

\noindent For the proof of Corollary \ref{cor:xqpsimulation}, see Appendix \ref{appendix:alternativeproof}. Although proven for $\mathsf{XQP(\pi/4)}$, the statement also holds for $\mathsf{XQP}$. More importantly, it tells us that circuits consisting solely of computational basis SPAM and $\sqrt{\mathrm{SWAP}}$ are hard to simulate classically. Given the ease of implementing $\sqrt{\mathrm{SWAP}}$ in practice (see Section \ref{sec:discussion}), implementing $\mathsf{XQP(\pi/4)}$ circuits could form one of the lowest bars for experimental demonstrations of quantum advantage.

\section{Computational Universality of XQP without Postselection} \label{sec:additive}

In Section \ref{sec:gadget} we have shown that the phase gadget $\mathcal{G}(\theta)$ admits a boosting procedure, which brings its success probability close to $1$ at the cost of polynomially many additional ancillae and exchange gate applications. In other restricted models, such constructions are typically not available, with success probabilities held constant \cite{Bremner_2010}. Consequently, the overall success probability of a circuit composed of such gadgets decays exponentially with their count, making their use experimentally impractical. For proving multiplicative simulation hardness, this is not an issue -- the relevant proof only requires postselection to be possible, rather than likely. In our case, this is true for models where a fixed pulse angle $\theta$ is large, for example $\mathsf{PostXQP}(\pi/4)$, where $\mathcal{G}(\pi/4)$ succeeds with a constant probability of $1/2$. On the other hand, if small pulse angles $\pi/4N$ are available, by Lemma \ref{lemma:phasegadget} one may boost the success probability of $\mathcal{G}(\pi/4)$ to $1 - \mathcal{O}(1 / N)$ with overheads polynomial in $N$. Here we exploit this fact, designing $\mathsf{PostXQP}(\pi/4N)$ circuits that sample from $\mathsf{BQP}$ computations on $n$ qubits with a success probability $1 - \varepsilon$ under time and ancilla overheads of just $\mathcal{O}(n^2 / \varepsilon)$. As a consequence, we show that the full unconstrained \textsf{XQP} model can solve any problem in $\mathsf{BQP}$. This proceeds by encoding a $\mathsf{BQP}$ computation into a $\mathsf{DFS}_3(\pi/4)$ computation, which is then encoded into an $\mathsf{XQP}(\pi/4N)$ circuit without postselection, using a number of gadgets that is linear in the qubit count. We boost the postselection probabilities to be sufficiently high so that their register can effectively be ignored, meaning we never leave the $\mathsf{XQP}(\pi/4N)$ model and still obtain an output distribution that is close to the desired $\mathsf{BQP}$ output distribution in total variation distance. 

\begin{lemma} \label{lem:xqpbqp} With $N = \mathcal{O}(n / \varepsilon)$, a $\mathsf{PostXQP}(\pi / 4N)$ circuit can successfully implement any $\mathsf{BQP}$ computation on $n$ logical qubits with probability at least $1 - \varepsilon$ using $3n$ physical qubits, $\mathcal{O}(n^2 / \varepsilon)$ ancillary qubits, and a total gadget time overhead of $\mathcal{O}(n^2 / \varepsilon)$.
\end{lemma}

\begin{proof}

Given a $\mathsf{BQP}$ circuit on $n$ qubits, our goal is to construct a $\mathsf{PostXQP}(\pi/4N)$ circuit which reproduces the original output distribution with a total postselection success probability of $1 - \varepsilon$. To do so, we first prepare the $\mathsf{DFS}_3(\pi/4)$ state $|0^n_L \rangle = (|S \rangle |0 \rangle)^{\otimes n}$ on $3n$ physical qubits using $n$ many $S$ gates, via $(S \otimes \mathbf{1}) U(\pi/4) |01 \rangle = |S \rangle$. This forms the logical all-zero state of a $\mathsf{DFS}_3(\pi/4)$ computation. By Lemma \ref{lemma:dfs3universal}, we can then implement the original $\mathsf{BQP}$ circuit on the logical space using $U(\pi/4)$ exchange gates only. To perform a measurement in the logical basis via physical measurement in the computational basis, we can use the technique presented in Appendix \ref{appendix:alternativeproof}. This only requires measuring the middle physical qubit of each triple encoding a logical qubit, and negating the measurement outcome to obtain $|\lnot x_2 \dots \lnot x_{3n - 1} \rangle = |\mathbf{x}_L \rangle$, using one $S$ gate per logical qubit. Altogether, this is sufficient to encode the $\textsf{BQP}$ computation.

Overall, the described procedure requires $2n$ applications of the phase gadget $\mathcal{G}(3\pi/4)$ and a $\text{poly}(n)$ number of exchange gates for the $\mathsf{DFS}_3(\pi/4)$ computation. By Lemma \ref{lemma:phasegadget}, each individual phase gadget can be boosted to a success probability of $1 - \epsilon$, where $\epsilon = \mathcal{O}(1/N)$. For the total success probability, we require $(1 - \epsilon)^{2n} \geq 1 - \varepsilon$. Thus, we can set $\epsilon = \varepsilon / 2n$, giving $N = \mathcal{O}(n / \varepsilon)$. 

Assuming native access to each $U(k \pi / 4N)$ for $k = 1, \dots, N$, by Lemma \ref{lemma:phasegadget} we require $2N$ ancillary qubits and exchange gate applications per gadget, giving the stated $4Nn = \mathcal{O}(n^2 / \varepsilon)$ ancilla and exchange gate overheads. If only $U(\pi / 4N)$ is directly implementable, in the worst case the total time overhead rises to $4N \times 4Nn = \mathcal{O}(n^3 / \varepsilon^2)$, which is still polynomial in $n$ and $1/\varepsilon$. If mid-circuit measurements are allowed, then by Lemma \ref{lemma:phasegadget} we only require $2$ ancillary qubits. 
\end{proof}

As a direct consequence of Lemma \ref{lem:xqpbqp}, we can show that classically simulating $\mathsf{XQP}(\pi/4N)$ circuits to a small total variation distance is unlikely.

\begin{theorem} \label{thm:tvxqpbqp}
Suppose we can efficiently weakly simulate $\mathsf{XQP}(\pi/4N)$ circuits on $n$ qubits to within total variation distance $\delta$. Then, we can efficiently weakly simulate $\mathsf{BQP}$ circuits on $\mathcal{O}(n / N)$ qubits to within total variation distance $\delta + \varepsilon$, where $\varepsilon = \mathcal{O}(n/N^2)$.
\end{theorem}

\begin{proof}
Our goal is to construct an efficient classical sampler $\tilde{\mathcal{D}}(\mathbf{z})$ for a $\mathsf{BQP}$ output distribution $\tilde{D}(\mathbf{z}) = |\langle \mathbf{z} | \tilde{U} | \mathbf{0} \rangle |^2$ on $m$ qubits, given an efficient classical sampler $\mathcal{D}(\mathbf{x,y})$ for an $\mathsf{XQP}(\pi/4N)$ output distribution on $n$ qubits. We will again proceed by encoding the $\mathsf{BQP}$ computation into a $\mathsf{DFS}_3(\pi/4)$ computation encoded in an $\mathsf{XQP}(\pi/4N)$ circuit. 

Let $n = 3m + 2k$, where $m$ is the number of $\mathsf{DFS}_3(\pi/4)$ logical qubits, and $2k$ is the number of ancillae for the phase gadgets. Define the distribution $D(\mathbf{x, y}) = | \langle \mathbf{x, y} | U | 010^m 01^k \rangle |^2$, where $U$ is an $\mathsf{XQP}(\pi/4N)$ circuit such that on measurement $\mathbf{y} = 01^k$, the $|010^m \rangle$ state is (i) encoded into the $\mathsf{DFS}_3(\pi/4)$ state $| 0^m_L \rangle$, (ii) an encoded $\mathsf{BQP}$ circuit $\tilde{U}$ is performed on the logical space using $\sqrt{\mathrm{SWAP}}$ gates (Lemma \ref{lemma:dfs3universal}), and (iii) the output state is decoded back into the computational basis, such that for a measurement $x_1 x_2 x_3$ of consecutive triples on the $3m$-qubit register, $|z_L \rangle = |f(x_1, x_2, x_3)\rangle$ where $f$ is some Boolean function (see Appendix \ref{appendix:alternativeproof}). Conditioning on the measurement $\mathbf{y} = 01^k$ on the ancilla register,

\begin{equation}
D(\mathbf{x} | \mathbf{y} = 01^k) = \frac{D(\mathbf{x}, \mathbf{y} = 01^k)}{ 1 - \varepsilon } = Q(\mathbf{x}),
\end{equation} 

\noindent where $Q(\mathbf{x})$ is the desired $\mathsf{BQP}$ output distribution up to a deterministic post-processing map $g: \{0, 1 \}^{3m} \rightarrow \{0, 1\}^m$, where:

\begin{equation}
g(\mathbf{x}) = f(x_1, x_2, x_3), \dots, f(x_{3m-2}, x_{3m-1}, x_{3m}) = \mathbf{z},
\end{equation}

\noindent so that $g(Q(\mathbf{x})) = \tilde{D}(\mathbf{z})$. Also define the marginal distribution $D'(\mathbf{x}) = \sum_{\mathbf{y}} D(\mathbf{x, y})$. The total variation distance (TVD) between $D'$ and $Q$ is given by:

\begin{align}
    & \|D' - Q \| 
    =  \frac{1}{2} \sum_{\mathbf{x}} \left| \sum_{\mathbf{y}} D(\mathbf{x, y}) - Q(\mathbf{x}) \right| \\
     &= \frac{1}{2} \sum_{\mathbf{x}} \left| D(\mathbf{x}, \mathbf{y} = 01^k) + D(\mathbf{x}, \mathbf{y} \neq 01^k) - Q(\mathbf{x}) \right| \\
     &= \frac{1}{2} \sum_{\mathbf{x}} \left|  (1 - \varepsilon) D(\mathbf{x} | \mathbf{y} = 01^k) + \varepsilon D(\mathbf{x} | \mathbf{y} \neq 01^k) - Q(\mathbf{x}) \right| \\
     &= \varepsilon \cdot \frac{1}{2} \sum_{\mathbf{x}} \left|  D(\mathbf{x} | \mathbf{y} \neq 01^k) - D(\mathbf{x} | \mathbf{y} = 01^k) \right| \leq \varepsilon .
\end{align}

\noindent Now, suppose we can efficiently sample from a distribution $\mathcal{D}(\mathbf{x,y})$, such that in total variation distance $\| \mathcal{D} - D \| \leq \delta$. Define the marginal $\mathcal{D}'(\mathbf{x}) = \sum_{\mathbf{y}} \mathcal{D}(\mathbf{x, y})$. The TVD between $\mathcal{D}'(\mathbf{x})$ and $D'(\mathbf{x})$ is given by:
\begin{align}
    & \| \mathcal{D}' - D' \| 
    =  \frac{1}{2} \sum_{\mathbf{x}} \left| \sum_{\mathbf{y}} ( \mathcal{D}(\mathbf{x, y}) -  D(\mathbf{x, y}) ) \right| \\
    & \leq \frac{1}{2} \sum_{\mathbf{x, y}} \left| \mathcal{D}(\mathbf{x, y}) -  D(\mathbf{x, y}) \right| \\
    & \leq \delta.
\end{align}

\noindent By the triangle inequality, $\| \mathcal{D}' - Q \| \leq \| \mathcal{D}' - D' \| + \| D' - Q \| \leq \delta + \varepsilon$. With our classical sampler, we can ignore the ancilla register altogether and read off $\mathbf{x}$ from $\mathcal{D}(\mathbf{x, y})$ to sample from $\mathcal{D}'(\mathbf{x})$. We can then apply the post-processing map $g$ to the output to sample from a distribution $\tilde{\mathcal{D}}(\mathbf{z}) = g(\mathcal{D}'(\mathbf{x}))$. By the data processing inequality \citep{Polyanskiy_2025}, $\| \tilde{\mathcal{D}} - \tilde{D} \| = \| g(\mathcal{D}') - g(Q) \| \leq \| \mathcal{D}' - Q \| \leq \delta + \varepsilon$. Therefore, we can efficiently sample from a distribution $\tilde{\mathcal{D}}(\mathbf{z})$ which is close to the output of a $\mathsf{BQP}$ computation in TVD $\delta + \varepsilon$. By Lemma \ref{lem:xqpbqp}, $k = \mathcal{O}(m^2 / \varepsilon)$, $\varepsilon = \mathcal{O}(m/N)$, so $k = \mathcal{O}(mN)$. Therefore, $m = \mathcal{O}(n / N)$ and the additional error in TVD is $\varepsilon = \mathcal{O}(n/N^2)$.
\end{proof}

Our techniques further imply that despite no access to singlet states, the unconstrained family \textsf{XQP} is able to solve any problem in \textsf{BQP}. 

\begin{corollary} \label{cor:xqpequalsbqp}
$\mathsf{XQP} = \mathsf{BQP}$.
\end{corollary}

\begin{proof}
Since any $U(\theta)$ gate is a two-qubit unitary with efficiently computable entries, $\mathsf{XQP} \subseteq \mathsf{BQP}$ is immediate. For the other direction, let $L$ be a language in $\mathsf{BQP}$. Then, following amplification by a majority vote, there exists a polynomial-time uniform family $\{C_w\}$ of quantum circuits of size $m = \mathrm{poly}(|w|)$, with a designated single-qubit output register $\mathcal{O}_w$ such that $p(\mathcal{O}_w = 1) \geq 5/6$ if $w \in L$ and $p(\mathcal{O}_w = 0) \geq 5/6$ if $w \notin L$. Applying the encoding presented in Lemma \ref{lem:xqpbqp} and Theorem \ref{thm:tvxqpbqp} yields an $\mathsf{XQP}(\pi/4N)$ circuit $\tilde{C}_w$ on $\mathcal{O}(mN)$ qubits, with an efficient classical post-processing map $g$ giving a decoded bit $\tilde{\mathcal{O}}_w$. As the output distribution of $\tilde{C}_w$ is within total variation distance $\varepsilon = \mathcal{O}(m/N)$ of $C_w$, we have $p(\tilde{\mathcal{O}}_w = 1) \geq 5/6 - \varepsilon$ for $w \in L$ and $p(\tilde{\mathcal{O}}_w = 0) \geq 5/6 - \varepsilon$ for $w \notin L$. Setting $N = cm$ for a sufficiently large integer $c$ ensures that $\varepsilon \leq 1/6$. Since the resulting circuit depth, size, and time to implement the encoding is polynomial in $m$, we then have that $L \in \mathsf{XQP}$.
\end{proof}

We emphasise that the above construction specifically requires $N$ to vary with the instance size of the problem, i.e. access to pulse angles $\theta \sim \mathcal{O}(1/m)$ which shrink as the problem size increases. For any \textit{restricted} family $\mathsf{XQP}(\pi/4N)$ with $N$ fixed, the error $\mathcal{O}(m/N)$ introduced by the above method will exceed the constant gap for sufficiently large $m$, violating the bounded error constraints. 

Roughly speaking, Theorem \ref{thm:tvxqpbqp} only proves the `universality' of $\mathsf{XQP}(\pi/4N)$ up to a finite size (determined by $N$), showing that circuits in this family can simulate any quantum algorithm on $m$ qubits so long as $m$ is not much larger than $N$. Corollary \ref{cor:xqpequalsbqp} complements this, showing that as the minimum pulse angles available become very small, the computational power of $\mathsf{XQP}(\pi/4N)$ asymptotically approaches that of \textsf{BQP}. Finally, Theorem \ref{thm:postxqpbqp} tells us that along the way, each intermediate $\mathsf{XQP}(\pi/4N)$ remains multiplicatively hard to simulate. Whilst the full unrestricted family \textsf{XQP} is theoretically interesting, in practice any physical implementation of the Heisenberg exchange interaction will be limited to a fixed, minimum pulse angle -- it is precisely our restricted intermediates that aim to model these real-world scenarios. In Section \ref{sec:semi_universality} we strengthen these results and show that Corollary \ref{cor:xqpequalsbqp} also holds for $\mathsf{XQP}(\pi/4N)$ for any positive integer $N$. This is achieved by indirectly approximating $U(\pi/4M)$ gates (for arbitrarily large $M$) by $U(\pi/4)$ gates, thereby allowing $\mathsf{XQP}(\pi/4)$ circuits to solve any problem in $\mathsf{BQP}$ without constraints on the problem size or pulse angles. Overall, these results provide compelling evidence that additive-error simulation of $U(\pi/4N)$ networks is unlikely to be classically tractable, even for small $N$. Whether this is true in the `average-case' remains an open question, which we discuss in Section \ref{sec:open}.

\section{The \texorpdfstring{$\sqrt{\mathbf{SWAP}}$}{sqrt(SWAP)} Gate is Universal for Quantum Computation} \label{sec:semi_universality}

Thus far we have shown that postselected $\mathsf{XQP(\pi/4)}$ circuits are universal for quantum computation, and the full \textsf{XQP} model is also universal without postselection. By introducing new tools we now strengthen these results, showing that the $\sqrt{\mathrm{SWAP}}$ exchange interaction alone is sufficient for universal quantum computation. We begin by proving the computational universality of $\sqrt{\mathrm{SWAP}}$ in the $\mathsf{DFS}_3$ encoding, i.e. given access to singlet states. We will then show that this computational universality extends to computational basis SPAM, i.e. that $\mathsf{XQP}(\pi/4) = \mathsf{BQP}$.

\dfsthree*

\begin{proof}
We first write the matrix elements of $E_{12}, E_{23}$ and $E_{13}$ in the $\{ |0_L \rangle, |1_L \rangle \}$ basis of $\mathsf{DFS}_3$ using Young's orthogonal form:

\begin{align}
E_{12} &= \begin{pmatrix} -1 & 0 \\ 0 & 1 \end{pmatrix}, \\
E_{23} &= \frac{1}{2}\begin{pmatrix} 1 & -\sqrt{3} \\ -\sqrt{3} & -1 \end{pmatrix}, \\
E_{13} &= \frac{1}{2} \begin{pmatrix} 1 & \sqrt{3} \\ \sqrt{3} & -1 \end{pmatrix}.
\end{align}

\noindent The exchange gates $U(\pi/4)_{ij}$ are thus single-qubit rotations $\cos (\pi/4) \mathbf{1} + i \sin (\pi/4) (\mathbf{n} \cdot \boldsymbol{\sigma}_L)$ on the logical space, where $\boldsymbol{\sigma}_L = (X_L, Y_L, Z_L)$, \ $\mathbf{n}_{12} = (0, 0, -1), \ \mathbf{n}_{23} = (-\sqrt{3}/2, 0, 1/2)$, and $\mathbf{n}_{13} = (\sqrt{3}/2, 0, 1/2)$. Now consider the following choice of generators $G_i$ constructed from sequences of $\sqrt{\mathrm{SWAP}}_{ij}$:

\begin{align}
    G_1 &= U(\pi/4)_{12} U(\pi/4)_{23} U(3\pi/4)_{12} U(3\pi/4)_{23}, \label{eq:generator} \\
    G_2 &= U(\pi/4)_{23} U(\pi/4)_{13} U(3\pi/4)_{23} U(3\pi/4)_{13}.
\end{align}

\noindent For this choice, $G_i = \cos \phi \mathbf{1} + i \sin \phi (\mathbf{m}_i \cdot \boldsymbol{\sigma}_L)$, where $\phi = \arccos(5/8)$, $\mathbf{m}_1 = (\sqrt{3/39}, -\sqrt{27/39}, -\sqrt{9/39})$, and $\mathbf{m}_2 = (-\sqrt{12/39}, -\sqrt{27/39}, 0)$. Because $\phi$ is an irrational multiple of $\pi$ and $\mathbf{m}_1 \not\parallel \mathbf{m}_2$, the set $\langle G_1, G_2 \rangle$ generated by $G_i$ is dense in $SU(2)$, and therefore can approximate any logical single-qubit gate. To obtain a universal set of gates, we also need a logical two-qubit entangling gate constructed from $U(\pi/4)_{ij}$ gates only. Such constructions are provided in \citep{Setiawan_2014, Shi_2012, Weinstein_2023}.

\end{proof}

Whilst the above implies that the $\sqrt{\mathrm{SWAP}}$ interaction is sufficient for universal (encoded) quantum computation with singlets, it does not tell us whether it is `universal' in the same sense as the full set of exchange interactions with arbitrary pulse angles, i.e. whether it can simultaneously (and independently) implement $SU(\dim V_J)$ action on all $S_n$ irreps $V_J$, for any $n$ \citep{Kempe_2001}. Using the notion of `semi-universality' (defined in \citep{hulse2024}), we can resolve this question in the affirmative.

\begin{definition}[Semi-Universality \citep{hulse2024}]
Recall the decomposition of the $n$-qubit Hilbert space under the joint action of $S_n \times SU(2)$ in Definition \ref{def:schur_weyl}:

\begin{equation} 
    (\mathbb{C}^2)^{\otimes n} \cong \bigoplus_{J \in \mathcal{J}_n} V^{S_n}_J \otimes W^{SU(2)}_J.
\end{equation}

\noindent Let $\mathcal{V}^{(n)} = \{V : [V, U^{\otimes n}] = 0 \ \forall U \in SU(2) \}$ be the group of all $SU(2)$-invariant unitaries on $n$ qubits. By Schur's lemma, any unitary $V \in \mathcal{V}^{(n)}$ may be written in the form:

\begin{equation} \label{eq:semiuniversality_decomp}
    V = \bigoplus_{J \in \mathcal{J}_n} v_J \otimes \mathbf{1}_{2J + 1},
\end{equation}

\noindent where $v_J \in U(\dim V_J)$ are unitaries acting only on the $S_n$ irrep spaces $V_J$. Similarly, let $\mathcal{SV}^{(n)} \subset \mathcal{V}^{(n)}$ be the group of all $SU(2)$-invariant unitaries on $n$ qubits without relative phases between the isotypic components $(V_J \otimes W_J)$. For $V \in \mathcal{SV}^{(n)}$ Equation \eqref{eq:semiuniversality_decomp} holds with $v_J \in SU(\dim V_J)$ for all $J \in \mathcal{J}_n$. Up to isomorphism, we may therefore write:
\begin{align}
    &\mathcal{V}^{(n)} \cong \prod_{J \in \mathcal{J}_n} U(\dim V_J), & &\mathcal{SV}^{(n)} \cong \prod_{J \in \mathcal{J}_n} SU(\dim V_J).
\end{align}

\noindent We say that a subgroup $\mathcal{W}^{(n)} \subseteq  \mathcal{V}^{(n)}$ of $SU(2)$-invariant unitaries on $n$ qubits is semi-universal iff it holds true that: 
\begin{equation}
    \mathcal{SV}^{(n)} \subseteq \mathcal{W}^{(n)} \subseteq  \mathcal{V}^{(n)}. 
\end{equation}
    
\end{definition}

\begin{theorem} \label{thm:semiuniversal}
$\mathsf{XQP}(\pi/4)$ circuits are semi-universal for any number of qubits $n$. That is, $\mathsf{XQP}(\pi/4)$ circuits on $n$ qubits generate a group whose closure contains $\mathcal{SV}^{(n)}$.
\end{theorem}

\noindent To prove Theorem \ref{thm:semiuniversal}, we will adapt a useful result from Section V of \citep{hulse2024}:

\begin{lemma} \label{lemma:semiuniversalinduction}
For a semi-universal set of $m$-qudit unitaries $\mathcal{SV}^{(m)}$, let $\mathcal{W}^{(m+1)} = \langle \mathcal{SV}^{(m)} \otimes \mathbf{1}, \mathbf{1} \otimes \mathcal{SV}^{(m)} \rangle$ denote the closure of the set of $m+1$-qudit unitaries generated by applying $\mathcal{SV}^{(m)}$ to the first, or last $m$ qudits. For $m \geq 4$, $\mathcal{W}^{(m+1)} = \mathcal{SV}^{(m+1)}$.
\end{lemma}

\begin{proof}[Proof of Theorem \ref{thm:semiuniversal}]
    To apply Lemma \ref{lemma:semiuniversalinduction} and obtain Theorem \ref{thm:semiuniversal}, we just need to show that the set of $\sqrt{\mathrm{SWAP}}$-generated unitaries on $4$ qubits is semi-universal, i.e. $\langle U(\pi/4)_{ij} : i, j \in  \llbracket 4 \rrbracket \rangle \supseteq \mathcal{SV}^{(4)}$ (where $\langle \cdot \rangle$ denotes the closure of a generated set). We have already shown this is the case for $n=3$ in Lemma \ref{lemma:dfs3universal}. Notably, because $U(3\pi/4) = - U(\pi/4)^\dagger$, the matrices $G_i$ satisfy $\det G_i = 1$ on the $J=1/2$ irrep (spanned by $|0_L \rangle, |1_L \rangle$), and $G_i = e^{2\pi i} = 1$ on the (one-dimensional) $J=3/2$ irrep. Therefore, up to isomorphism:
    \begin{equation}
    \langle U(\pi/4)_{ij} : i, j \in  \llbracket 3 \rrbracket \rangle \supseteq SU(2) \times SU(1).
    \end{equation}
    
    For $n=4$, $\mathcal{J}_4 = \{0, 1, 2\}$, with corresponding irrep dimensions $\dim V_J \in \{2, 3, 1\}$. We need to be able to implement a unitary which acts as $SU(2) \times SU(3) \times SU(1)$ on the irrep spaces. In Young's orthogonal form, the transpositions $E_{ij}$ are given by:

\begin{align}
E_{12} &=  \left( \begin{array}{cc|ccc} 1 & 0 & 0 & 0 & 0 \\ 0 & -1 & 0 & 0 & 0 \\ \hline 0 & 0 & 1 & 0 & 0 \\ 0 & 0 & 0 & 1 & 0 \\ 0 & 0 & 0 & 0 & -1 \end{array} \right), \label{eq:yorthoform} \\ \\
E_{23} &=  \left( \begin{array}{cc|ccc} -\frac{1}{2} & \frac{\sqrt{3}}{2} & 0 & 0 & 0 \\ \frac{\sqrt{3}}{2} & \frac{1}{2} & 0 & 0 & 0 \\ \hline 0 & 0 & 1 & 0 & 0 \\ 0 & 0 & 0 & -\frac{1}{2} & \frac{\sqrt{3}}{2} \\ 0 & 0 & 0 & \frac{\sqrt{3}}{2} & \frac{1}{2}\end{array} \right), \\ \\
E_{34} &=  \left( \begin{array}{cc|ccc} 1 & 0 & 0 & 0 & 0 \\ 0 & -1 & 0 & 0 & 0 \\ \hline 0 & 0 & -\frac{1}{3}  & \frac{\sqrt{8}}{3} & 0 \\ 0 & 0 & \frac{\sqrt{8}}{3} & \frac{1}{3} & 0 \\ 0 & 0 & 0 & 0 & 1\end{array} \right),
\end{align}

\noindent where the $2 \times 2$ block acts on $V_{0}$, the $3 \times 3$ block acts on $V_1$, and we have omitted action on the one-dimensional $V_2$ where $E_{ij} = (1)$. Notice that our previously constructed generators $G_i$ allow us to implement any $SU(2)$ unitary $A$ on $V_0$. However, this implements the same unitary in the form $1 \oplus A$ on $V_1$, which does not allow independent action on the two irrep spaces. To bypass this, we can use the fact that on $V_0$, $E_{12} = E_{34}$, as well as that $E_{13} = E_{24}$ and $E_{23} = E_{14}$. We can therefore construct the following generators, which act as the identity on $V_0 \oplus V_2$, but act non-trivially on $V_1$:

\begin{align}
G_{ij, kl} &= [ U(\pi/4)_{ij} U(\pi/4)_{kl} U(3\pi/4)_{ij} U(3\pi/4)_{kl} ] \label{eq:generator2} \\
& \times [ U(\pi/4)_{\widetilde{kl}} U(\pi/4)_{\widetilde{ij}} U(3\pi/4)_{\widetilde{kl}} U(3\pi/4)_{\widetilde{ij}} ], \\
\end{align}

\noindent where $\widetilde{ij} = \{1, 2, 3, 4\} \setminus \{i, j\}$ is the complementary pair of $ij$, i.e. $\widetilde{12} = 34$. On $V_0$, the generators $G_{ij, kl}$ act as a unitary of the form $UU^\dagger = \mathbf{1}$ (for example, $G_{12, 23}$ acts as $G_1 G_1^\dagger$). Consider the set $H = \langle G_{12, 13}, G_{13, 23} \rangle$, where again for simplicity $\langle \cdot \rangle$ denotes the closure. We will show that it is dense in $SU(3)$ by elimination. A simple calculation shows that on $V_1$, each generator has determinant $1$ and spectrum $\{ 1, e^{i \phi}, e^{-i \phi} \}$, where $\phi = \arccos(1/8)$ is an irrational multiple of $\pi$. Therefore, on $V_1$ the closure of both generators individually is given by $ \langle G_{12, 13} \rangle = \{ M_{12, 13} \mathrm{diag}(1, e^{i t}, e^{-i t}) M_{12, 13}^\dagger : t \in \mathbb{R} \}$ for some unitary $M_{12, 13}$, and similarly for $ \langle G_{13, 23} \rangle$. By the irrationality of $\phi/\pi$, these closures form connected circles in $SU(3)$ containing the identity. Therefore, $\langle G_{12, 13} \rangle$ and $\langle G_{13, 23} \rangle$ each lie in the identity component $H_0$ of $H$. Since $H_0$ is a subgroup of $SU(3)$, it is closed under multiplication, and so the product of any two elements in $H_0$ also lies in $H_0$. This implies that $H$ itself is contained in $H_0$, therefore $H = H_0$ and $H$ is connected. Next, suppose that $G_{13, 23}$ and $G_{12, 13}$ share a common eigenvector $|\lambda \rangle \in V_1$. Then, $[G_{12, 13}, G_{13, 23}] |\lambda \rangle = 0$. Another simple calculation shows that $\det [G_{12, 13}, G_{13, 23}] = 801 i / 512$, hence the generators cannot share a common eigenvector, and so $H$ acts irreducibly on $V_1$. We now use the fact that a connected, closed subgroup of $SU(3)$ acting irreducibly on $\mathbb{C}^3$ is either $SU(3)$ itself, or conjugate to $SO(3)$ \cite{antoneli_2014}. The latter can be ruled out by the fact that $\mathrm{Tr}(G_{12, 13} G_{13, 23}) = 9/64 - 9i/32$, whereas every element of $SO(3)$ has a real trace. Therefore, $H$ cannot be conjugate to $SO(3)$, and so it must be that $H = SU(3)$. Thus, the group generated by $G_{ij, kl}$ is dense in $SU(3)$, and so we can implement any $SU(3)$ unitary on $V_1$ which acts as the identity on $V_0 \oplus V_2$. Using the generator sets $\langle G_i \rangle$ and $\langle G_{ij, kl} \rangle$ we can then generate unitaries of the form:

\begin{align}
    U_A &= \left( \begin{array}{c|c|c} A & & \\ \hline &   \begin{matrix} 1 & 0 \\ 0 & A \end{matrix} & \\ \hline  & & 1 \end{array}  \right), & U_{A, B} &= \left( \begin{array}{c|c|c}  \mathbf{1}  & &  \\ \hline &   \begin{pmatrix} 1 & 0 \\ 0 & A^\dagger \end{pmatrix} B & \\ \hline & & 1 \end{array} \right),
\end{align}

\noindent where $A \in SU(2)$ and $B \in SU(3)$. We can therefore generate unitaries of the form $U_A U_{A, B}$, which act as $ A \oplus B \oplus 1$ on $V_0 \oplus V_1 \oplus V_2$. Each constituent generator is built out of $U(\pi/4)_{ij}$ exchange gates, and so:

\begin{equation}
    \langle U(\pi/4)_{ij} : i, j \in  \llbracket 4 \rrbracket \rangle \supseteq SU(2) \times SU(3) \times SU(1).
\end{equation}

\noindent This establishes that on $n=4$, $\mathsf{XQP}(\pi/4)$ circuits contain $\mathcal{SV}^{(4)}$, i.e. are semi-universal. By Lemma \ref{lemma:semiuniversalinduction}, this implies that $\mathsf{XQP}(\pi/4)$ circuits are semi-universal for any number of qubits $n$.

\end{proof}

Semi-universality of $\mathsf{XQP}$ and $\mathsf{XQP}(\pi /4)$ implies that both circuit families can generate $SU(\dim V_J)$ unitaries $v_J$ on each subspace $V_J$, as in Equation \eqref{eq:semiuniversality_decomp}. However, neither can implement the full centre $\mathcal{P}$ of the group $\mathcal{V}^{(n)}$, consisting of relative phases on the isotypic components:

\begin{equation} \label{eq:centre}
    \mathcal{P} = \{ \sum_{J \in \mathcal{J}_n} e^{i \theta_J} \Pi_J \ : \ \theta_J \in \mathbb{R} \},
\end{equation}

\noindent where $\Pi_J$ is a projector onto $(V_J \otimes W_J)$ and $\mathcal{V}^{(n)} = \mathcal{P} \cdot \mathcal{SV}^{(n)}$. Despite this, some relative phases can be implemented -- on this front, we can show that $\mathsf{XQP}(\pi/4)$ is heavily constrained. For an exchange gate $U_{ij}(\theta)$, the corresponding unitary $u_{ij}(\theta)_J$ acting on the irrep space $V_J$ satisfies:

\begin{equation}
    \det u_{ij}(\theta)_J = \det e^{i \theta \cdot [ij]_J} =  
    e^{i \theta \cdot \mathrm{Tr} [ij]_J} = e^{i \theta \cdot \chi_J(ij)},
\end{equation}

\noindent where $[ij]_J$ is the irreducible representation of $(ij)$ on $V_J$, and $\chi_J(ij)$ is the irreducible character of the transposition $(ij)$ in the irrep $J$. For the symmetric group, $\chi_J(ij)$ is always an integer. Therefore, $\det u_{ij}(\pi/4)_J = e^{i \pi m / 4}$ for $m \in \mathbb{Z}$. We can use this fact to prove that unitaries generated by $\mathsf{XQP}(\pi/4)$ circuits form a strict subset of $\mathsf{XQP}$. We note that similarly to Remark \ref{remark:xqpbqp}, this is a statement about the accessible unitaries only, and does not concern computational power.

\begin{theorem} \label{thm:xqpsubset}
    $\mathsf{XQP}$ circuits on $n \geq 2$ qubits generate unitaries with relative phases between isotypic components $(V_J \otimes W_J)$ that cannot be approximated by $\mathsf{XQP}(\pi/4)$. 
\end{theorem}

\begin{proof}
    For two irreps $J, K$, define the determinant ratio:
    
    \begin{equation}
        \Delta_{J, K}(U) = \frac{\det (u_J)^{\dim V_K}}{\det (u_K)^{\dim V_J}}.
    \end{equation}
    
\noindent This quantity is invariant under a change of global phase $U \rightarrow e^{i \phi}U$:

\begin{equation}
    \Delta_{J, K}(e^{i \phi}U) = \frac{(e^{i \phi \dim V_J} \det u_J)^{\dim V_K} }{ (e^{i \phi \dim V_K} \det u_K)^{\dim V_J} } = \Delta_{J, K}(U).
\end{equation}

\noindent The motivation for our determinant ratio is the following: if for two unitaries $U, C$ we have that $C = e^{i \phi} U$, then $\Delta_{J, K}(C) = \Delta_{J, K}(U)$ for any pair $J, K \in \mathcal{J}_n$. We will identify a unitary $U \in \mathsf{XQP}$ such that no $C \in \mathsf{XQP}(\pi/4)$ can satisfy $\Delta_{J, K}(C) = \Delta_{J, K}(U)$ for some pair of irreps $J, K$. Therefore, there cannot exist a $C \in \mathsf{XQP}(\pi/4)$ such that $C = e^{i \phi} U$, and so unitaries generated by $\mathsf{XQP}(\pi/4)$ circuits cannot approximate arbitrary unitaries generated by $\mathsf{XQP}$ circuits.

For any $n \geq 2$, $\mathcal{J}_n$ contains the trivial one-dimensional $S_n$ irrep $J = n/2$, and the $(n-1)$-dimensional `standard' irrep $J = n/2 - 1$. For these, $\chi_{n/2}(ij) = 1$, and $\chi_{n/2 - 1}(ij) = n - 3$ (recall that in the standard representation the character of a permutation is its number of fixed points, minus one). Now, take for example $U = e^{i \frac{\pi}{3} E_{ij}} \in \mathsf{XQP}$. This satisfies:

\begin{equation}
    \Delta_{n/2, n/2-1}(e^{i \frac{\pi}{3} E_{ij}}) = \frac{e^{i \pi (n-1)  / 3}}{e^{i \pi (n- 3) / 3}} = e^{\frac{2\pi i}{3}}.
\end{equation}

\noindent On the other hand, any circuit $C \in \mathsf{XQP}(\pi/4)$ consisting of $m$ many $U(\pi/4)_{ij}$ gates will satisfy:
    
\begin{equation}
    \Delta_{n/2, n/2-1}(C) = \frac{e^{i \pi m(n-1)  / 4}}{e^{i \pi m (n- 3) / 4}} = e^{\frac{m \pi i}{2}}, \quad m \in \mathbb{N}.
\end{equation}

\noindent Notably, $\Delta_{n/2, n/2-1}(C)$ is constrained to take values from the discrete set $\{ \pm 1, \pm i \}$. Hence, even if the $\mathsf{XQP}(\pi/4)$ circuit $C$ is chosen such that (modulo phase) the action of $C$ on irrep spaces $V_J$ approximates $U_{ij}(\pi/3)$ (i.e. $c_J \propto u_{ij}(\pi/3)_J$ for all $J \in \mathcal{J}_n$), it is not able to approximate the relative phases between the different irreps. Therefore, $\mathsf{XQP}(\pi/4)$ unitaries form a strict subset of $\mathsf{XQP}$ unitaries.
\end{proof}

In spite of the constraint on relative phases, we now show that when supplemented with (unmeasured) ancillary qubits, $\sqrt{\mathrm{SWAP}}$ gates can approximate any arbitrary $U(\theta)$ exchange gate, and therefore any $\mathsf{XQP}$ circuit, up to a global phase. We begin with the following Lemma:

\begin{lemma} \label{lem:sqrtswap_approx}
Up to global phase, any $U(\theta)_{ij}$ exchange gate can be approximated by an $\mathsf{XQP}(\pi/4)$ circuit with the addition of two unmeasured ancillary qubits in the state $|00 \rangle_{ab}$.
\end{lemma}
\begin{proof}
We take a special case of the construction in Theorem 4 of \cite{Marvian_2024}. Consider the four-qubit gate:
\begin{equation}
C_{ijab}(\theta) = e^{i \theta (E_{ij} - E_{ab})} = U(\theta)_{ij} \otimes U(-\theta)_{ab}.
\end{equation}
Also consider the 4-qubit irrep decomposition: 
\begin{equation}
\mathcal{H}_i \otimes \mathcal{H}_j \otimes \mathcal{H}_a \otimes \mathcal{H}_b \cong (V_0 \otimes W_0) \oplus (V_1 \otimes W_1) \oplus (V_2 \otimes W_2).
\end{equation}
By Equation \eqref{eq:yorthoform} and the arguments surrounding Theorem \ref{thm:semiuniversal}, we have that the argument $(E_{ij} - E_{ab})$ of $C_{ijab}(\theta)$ acts as identity on all the $W$ (multiplicity) spaces, and as $\mathbf{0} \oplus c \oplus \mathbf{0}$ on $V_0 \oplus V_1 \oplus V_2$, where $c$ is a $3 \times 3$ matrix acting on the $V_1$ irrep space. Furthermore, $\mathrm{Tr}(c) = 0$ and $c^\dagger = c$, therefore we have that $C = 1 \oplus e^{i \theta c} \oplus 1$ where $e^{i \theta c} \in SU(3)$. By Theorem \ref{thm:semiuniversal}, $\mathsf{XQP}(\pi/4)$ circuits can implement any $SU(3)$ unitary on the $V_1$ irrep by means of the generators $G_{ij, kl}$. This will automatically implement the desired identity on $V_0$ and $V_2$. Explicitly, one may apply the Solovay--Kitaev algorithm \cite{dawson_2005} to obtain a circuit $\tilde{C}_{ijab}(\theta)$ consisting of $\sqrt{\mathrm{SWAP}}$ gates such that $\| \tilde{C}_{ijab}(\theta) - C_{ijab}(\theta) \| \leq \epsilon$ using $\mathcal{O}(\log^k(1/\epsilon))$ many $\sqrt{\mathrm{SWAP}}$ gates, for some constant $k$. 

Now, consider the action of $C_{ijab}(\theta)$ on $|\psi \rangle_{ij} \otimes |00 \rangle_{ab}$, where $|\psi \rangle_{ij}$ is an arbitrary two-qubit state. This is given by:
\begin{equation}
C_{ijab}(\theta) |\psi \rangle_{ij} \otimes |00 \rangle_{ab} = e^{-i \theta} U(\theta)_{ij} |\psi \rangle_{ij} \otimes |00 \rangle_{ab}.
\end{equation}
\noindent Therefore, applying $\tilde{C}_{ijab}(\theta)$ and discarding the ancillae allows us to implement $U(\theta)_{ij}$ up to a global phase, with an error of at most $\epsilon$.
\end{proof}

Lemma \ref{lem:sqrtswap_approx} implies that using two ancillary qubits and a small overhead in circuit size, any $\mathsf{XQP}$ circuit can be approximated by $\sqrt{\mathrm{SWAP}}$-generated circuits. Together with Corollary \ref{cor:xqpequalsbqp}, this implies that $\mathsf{XQP}(\pi/4)$ circuits can be used to solve any problem in $\mathsf{BQP}$. Remarkably, this establishes that the $\sqrt{\mathrm{SWAP}}$ gate with computational basis SPAM is sufficient for universal quantum computation.

\begin{theorem} \label{thm:sqrtswap_universal}
$\mathsf{XQP}(\pi/4) = \mathsf{XQP} = \mathsf{BQP}$. 
\end{theorem}
\begin{proof}
Consider an arbitrary $\mathsf{XQP}$ circuit $U = \prod_{k=1}^L U(\theta_k)_{i_k j_k}$ consisting of $L$ many $U(\theta)_{ij}$ gates. By Lemma \ref{lem:sqrtswap_approx}, each $U(\theta)_{ij}$ can be approximated up to global phase by synthesising $\tilde{C}_{ijab}(\theta)$ using $\sqrt{\mathrm{SWAP}}$ gates with the addition of two ancillary qubits, which may be reused throughout and discarded without measurement at the end. To obtain an overall circuit $\tilde{U}$ approximating $U$ up to a global phase and error $\epsilon$, we can choose the error of each $\tilde{C}_{ijab}(\theta)$ to be $\epsilon / L$. This requires $\mathcal{O}(L \log^k(L/\epsilon))$ many $\sqrt{\mathrm{SWAP}}$ gates, and implies that $\mathsf{XQP}(\pi/4) = \mathsf{XQP}$. By Corollary \ref{cor:xqpequalsbqp}, any $\mathsf{BQP}$ problem may be solved by an appropriately chosen $\mathsf{XQP}(\pi/4N)$ circuit with a sufficiently high $N$, which can itself be approximated by an $\mathsf{XQP}(\pi/4)$ circuit with two additional ancillae. Therefore, $\mathsf{XQP}(\pi/4) = \mathsf{BQP}$. By the fact that $U(\pi/4N)^N = \sqrt{\mathrm{SWAP}}$, it further follows that $\mathsf{XQP}(\pi/4N) = \mathsf{XQP} = \mathsf{BQP}$ for any $N \in \mathbb{N}$.
\end{proof}

Finally, Theorem \ref{thm:semiuniversal} implies that $\sqrt{\mathrm{SWAP}}$-generated circuits can implement any $SU(2)$-invariant unitary on $n$ qubits up to relative phases, previously only known to be generated by exchange gates with arbitrary pulse angles. Recent work \citep{Liu_2024, Mitsuhashi_2025} has shown that random symmetric quantum circuits generate $t$-designs for the uniform distribution over all $SU(2)$-invariant unitaries. This leads to the following Corollary:

\begin{corollary} \label{cor:designs}
    Random $\mathsf{XQP}(\pi /4)$ circuits (i.e. circuits consisting of $\sqrt{\mathrm{SWAP}}_{ij}$ gates only) generate $t$-designs for the uniform distribution over $\mathcal{V}^{(n)} = \{V : [V, U^{\otimes n}] = 0 \ \forall U \in SU(2) \}$, the set of $SU(2)$-invariant unitaries on $n$ qubits. The maximum attainable degree $t$ grows at least linearly with the  number, as $t \geq n - 2$.
\end{corollary}

\noindent This result follows directly from Corollary 2 of \citep{Liu_2024}, where it was applied to $U(\theta)$ gates. Notably, it relies solely on the semi-universality property of the random circuit ensemble and multiplicities of $SU(2)$ irreps -- for the lower bound in Corollary \ref{cor:designs}, the restrictions on relative phases in $\mathsf{XQP}(\pi/4)$ are not relevant. It would be interesting to explore the structure of random $\sqrt{\mathrm{SWAP}}$ circuits required to generate the $t$-designs, as well as calculate the exact value of $t$ for such circuits (for which restrictions on attainable relative phases are pertinent). For $\mathsf{XQP}$, in \cite{Liu_2024, Mitsuhashi_2025} it was shown that $t = (n-1)(n-3) - 1$. Whether the scaling of $t$ in random $\mathsf{XQP}(\pi /4)$ circuits is also quadratic in $n$, or remains linear is currently unknown. We address these questions in future work.

\section{Discussion} \label{sec:discussion}

\textit{Universal Quantum Computation with the $\sqrt{\mathrm{SWAP}}$ Gate}. At the beginning of this article, we defined the \textsf{XQP} model of computation. Over the course of each section, we introduced tools to successively boost its computational power. This culminated in Theorem \ref{thm:sqrtswap_universal}, which asserts that the $\sqrt{\mathrm{SWAP}}$ interaction is enough to achieve universal quantum computation with computational basis SPAM. Given the various steps involved, it is worth illustrating how this works in practice.

We begin with a \textsf{BQP} problem (or more broadly, a generic quantum algorithm) with a related circuit $C$. We wish to express this as an $\mathsf{XQP}(\pi/4)$ circuit, i.e. a sequence of computational basis SPAM and $\sqrt{\mathrm{SWAP}}$ gates only. Without loss of generality, we assume $C$ takes in a computational basis state $|0\rangle^{\otimes n}$ and consists of gates from some universal set, followed by a computational basis measurement. We then do the following:
\begin{enumerate}
    \item Decompose $C$ into a sequence of single-qubit rotations and $\mathrm{CNOT}$ gates. By Lemma \ref{lemma:dfs3universal}, each can be implemented on the logical $\mathsf{DFS}_3(\pi/4)$ space using $\sqrt{\mathrm{SWAP}}$ gates; the former is synthesised using generators $G_i$ from Equation \eqref{eq:generator}, whereas the latter can be implemented using the construction in \cite{Setiawan_2014}. This gives us a circuit $C'$ consisting of $\sqrt{\mathrm{SWAP}}$ gates and singlet SPAM that implements $C$ on the logical space. 
    \item Replace each input singlet state in $C'$ with the $|010\rangle$ state followed by a $\sqrt{\mathrm{SWAP}}$ gate and $S$ gate, as was done in Lemma \ref{lem:xqpbqp}. Similarly, run the inverse of this construction at the end of the circuit to allow for logical measurement via physical computational basis measurement (see Appendix \ref{appendix:alternativeproof} for the details). The $S$ gates can be implemented using the postselected phase gadget construction from Section \ref{sec:gadget}. 
    \item Since we don't want to perform any postselections in practice, we implement the probability-boosted phase gadgets. This can be done by applying the methods developed in Section \ref{sec:additive}. For each $S$ gate in $C'$, introduce an additional (poly-sized) ancilla register, so that the success probability of all postselections is sufficiently high to discard the ancillae without measurement. This gives us a circuit $C''$ consisting of $\sqrt{\mathrm{SWAP}}$ gates, small-angle exchange gates $U(\pi/4N)$ (where $N$ is large) and computational basis SPAM. 
    \item Finally, we want to remove the small-angle gadget exchange gates, since they are not native to the $\mathsf{XQP}(\pi/4)$ model. This can be done by introducing an additional, two-qubit ancilla register and synthesising four-qubit gates $U(\pi/4N) \otimes U(-\pi/4N)$ (where $N$ is large) directly using the generators $G_{ij,kl}$ from Equation \eqref{eq:generator2}. By Theorem \ref{thm:semiuniversal} and Lemma \ref{lem:sqrtswap_approx}, this synthesises $U(\pi/4N) \otimes U(-\pi/4N)$ using sequences of $\sqrt{\mathrm{SWAP}}$ gates only. The first tensor factor acts on qubits used by the phase gadgets, and the second tensor factor acts on the new ancilla register contributing a global phase. This gives us a circuit $C'''$ consisting of $\sqrt{\mathrm{SWAP}}$ gates and computational basis SPAM only, whose decoded output distribution can be made arbitrarily close to that of the original circuit $C$.
\end{enumerate}

A couple of remarks are in order. The above construction is efficient, since by the Solovay--Kitaev theorem each logical single-qubit rotation and $U(\pi/4N) \otimes U(-\pi/4N)$ gate can be approximated to within $\epsilon$ using $\mathcal{O}(\text{polylog}(1/\epsilon))$ gates \cite{dawson_2005}. Furthermore, other non-\textsf{DFS} encoding schemes may also be used to emulate arbitrary quantum circuits with $\mathsf{XQP}(\pi/4)$: for example, Corollary \ref{cor:xqpequalsbqp} and Theorem \ref{thm:sqrtswap_universal} admit an alternative proof using the two-qubit encoding used in the proof of Theorem \ref{thm:postxqpbqp}, where $S$ gates are required for logical operations rather than for logical SPAM. In that case, each $S$ gate can again be synthesised with $\sqrt{\mathrm{SWAP}}$ using the boosted gadget as was done in steps 3--4 above. This observation shows that the universality of \textsf{XQP} is not intrinsically tied to the $\mathsf{DFS}_3$ model: it is a powerful class of circuits in its own right. Finally, although the `full' \textsf{XQP} model itself is unphysical to implement directly (since any experimental setting would carry a minimum physically-implementable pulse angle $\theta_{\mathrm{min}}$), Theorem \ref{thm:sqrtswap_universal} shows that it can be approximated to arbitrary precision using the experimentally-accessible, fixed-angle $U(\pi/4) = \sqrt{\mathrm{SWAP}}$ gate. As a consequence, these two-qubit gates suffice to achieve universal quantum computation. This comes with the cost of overheads, ancillae, computational basis SPAM and logical encoding of qubits. The computational power of $\mathsf{XQP}(\pi/4)$ nonetheless provides a worst-case hardness of simulating such circuits, which are remarkably easy to construct -- whether this extends into average-case hardness is an interesting open question.

\

\textit{Further Properties of $\mathsf{XQP}$ Circuits.} Theorem \ref{thm:sqrtswap_universal} inspires many open problems, which we discuss further in Section \ref{sec:open}. For example, the power of \textsf{XQP} circuits under the presence of noise remains to be identified (as was done for \textsf{IQP} in \cite{Bremner_2017}), and we do not know if such circuits with low depth, random structure, or lack of phase gadgets remain hard to simulate. To help make progress on answering these questions, we describe additional properties of $\mathsf{XQP}$ circuits in the Appendix. The irrep structure properties of inputs in this model are studied in Appendix \ref{sec:schur_states}, where we decompose bitstring states into linear combinations of Schur states. Our observations lead to Theorem \ref{thm:amplitudes}, which illustrates how a hypothetical estimation algorithm of arbitrary $\mathsf{XQP}$ amplitudes could be used to estimate the amplitudes of arbitrary $\mathsf{BQP}$ circuits. This approach does not rely on phase gadgets, showing that the hardness of simulating $\mathsf{XQP}$ is not intrinsically tied to the presence of (approximate) phase gates. Appendix \ref{sec:entangling_power} studies the entangling power of exchange gates $U(\theta)$, showing that it is maximised for $\theta = \pi/4$, i.e. the $\sqrt{\mathrm{SWAP}}$ gate. These calculations are particularly relevant as a starting point for tensor network-based approaches to classically simulating $\mathsf{XQP}$, in the presence of noise or otherwise.

\textit{$\mathsf{XQP}$ Circuits and Six-vertex Models.} In the self-contained Appendix \ref{sec:potts}, we relate \textsf{XQP} circuits to \textit{six-vertex} models in statistical physics. This forms a bridge between our complexity-theoretic results and statistical mechanics, which allows us to gain additional insight into \textsf{XQP} from known results for the six-vertex model, as well as classify the computational hardness of certain partition functions by relating them to our findings:
\begin{itemize}
    \item We show that every \textsf{XQP} circuit measurement amplitude $\langle \mathbf{x} | U | \mathbf{y} \rangle$ is proportional to the partition function $Z_{6V}$ of a six-vertex model with complex weights (corresponding to the matrix entries of exchange gates $U(\theta)$), under boundary conditions determined by the bitstrings $\mathbf{x, y}$.
    \item We identify \textit{classically efficiently computable} instances of \textsf{XQP} measurement probabilities, for circuits corresponding to six-vertex models under domain wall boundary conditions. This is possible thanks to an efficient determinant formula of Izergin, Coker and Korepin \cite{Izergin_1992}. 
    \item By applying Corollary \ref{cor:xqpequalsbqp}, we show that approximating the partition function $Z_{6V}$ of six-vertex models with such weights is \textsf{BQP}-complete. In \cite{De_las_Cuevas_2011}, \textsf{BQP}-completeness was shown for partition functions of six-vertex models with weights corresponding to two-qubit gates drawn from $ \langle U(\theta), S \otimes \mathbf{1}, \mathbf{1} \otimes S \rangle$. We strengthen this result to six-vertex models with weights corresponding to $U(\theta)$ only. Our work also complements the findings of \cite{Cai_2017}, where the computational complexity of partition functions was classified for six-vertex models with real vertex weights. 
\end{itemize}

Finally, we study the problem of encoding the partition function of the Potts model $Z_{\mathrm{Potts}}$ into \textsf{XQP} measurement amplitudes, via encoding $Z_{\mathrm{Potts}}$ into the partition function of an appropriate six-vertex model. We find that doing so exactly without access to $S$ gates or singlet states is not possible, however an approximate encoding is possible by means of probability-boosted phase gadgets developed in Section \ref{sec:gadget}. These findings may lead to a better understanding of the computational complexity of the Potts model, which is already known to be $\# \mathsf{P}$-hard to approximate under certain conditions \cite{Galanis_2021}.

\

\textit{Experimental Realisation of $\mathsf{XQP}$ Circuits.} Spin qubits, in which local spin eigenstates constitute the computational basis and the isotropic Heisenberg interaction can be applied, form a natural platform for implementing $\mathsf{XQP}$ circuits. In particular, experimental instantiations of $\mathsf{XQP}$ were demonstrated in \cite{He_2019, Nowack_2011, Petta_2005}, where single-qubit preparation and readout in a local $\{ |\uparrow \rangle, |\downarrow \rangle \}$ basis was combined with an exchange gate (including $\sqrt{\mathrm{SWAP}}$) on quantum dots. In a recent experiment on cold atom lattices \cite{Zhang_2023} a local spin basis was also used for SPAM, with singlet states prepared using $\sqrt{\mathrm{SWAP}} + S$ gates, followed by a network of $\sqrt{\mathrm{SWAP}}$ gates to generate entanglement. This effectively instantiates the $\mathsf{XQP}(\pi/4) + S$ model, and is easily amenable to an experimental demonstration of $\mathsf{XQP}(\pi/4)$ without external resources. In some implementations of exchange-based quantum computation, such as \citep{Laird_2010, Weinstein_2023, Medford_2013, Eng_2015, Fogarty_2018, van_Diepen_2021, Madzik_2025, Abraham_2026, Bojovic_2026}, direct singlet SPAM is also available, making the restrictions of \textsf{XQP} redundant. Despite this, our work makes it clear that when these additional resources are not available (or when product state SPAM is preferable), the tunable exchange interaction remains sufficient for demonstrating super-classical computational power. In addition, we have shown that this interaction can be limited to a single $\sqrt{\mathrm{SWAP}}$ gate and still yield a family of circuits likely to be hard to simulate classically, paving the way for new demonstrations of quantum supremacy (for example, by sampling bitstrings from random $\sqrt{\mathrm{SWAP}}$ networks should their average-case hardness be established). As an example of where this alternative may be useful, we note that in the experiment of \cite{Madzik_2025} the authors noted a lack of scalability in their approach (singlet state preparation via postselection) due to a $1/4^n$ scaling of the success probability for $n$ qubits. 

As a circuit family (rather than a native set), $\mathsf{XQP}$ may also be realised on superconducting qubits: a recent experiment by Google \cite{Rosenberg_2024} performs what is essentially an $\mathsf{XQP}(2\pi/5)$ circuit on a 1D chain with additional phase gates at its boundary. This is because the $\mathrm{fSim}(\theta, \phi)$ gate used in the experiment satisfies $\mathrm{fSim}(\theta, 2 \theta) = (e^{i \theta Z/2} \otimes e^{i \theta Z/2}) U(- \theta)$, where the tensor product $e^{i \theta Z/2} \otimes e^{i \theta Z/2}$ commutes with the $U(-\theta)$ term. A simple modification of the circuit geometry (imposing periodic boundary conditions by introducing an additional $\mathrm{fSim}(\theta, 2 \theta)$ between the first and last qubit of the chain at each cycle, c.f. Figure 1A of \cite{Rosenberg_2024}) would put such circuits into an \textsf{XQP} form. Another compilation technique to obtain exchange gates using $\mathrm{fSim}(\theta, \phi)$ and single-qubit $X$ gates is to write $U(-2\theta) = \mathrm{fSim}(\theta, 2 \theta) (X \otimes X) \mathrm{fSim}(\theta, 2\theta) (X \otimes X)$. Both methods allow the implementation of `brickwork' $\mathsf{XQP}$ instances on superconducting platforms.

\section{Open Problems} \label{sec:open}
We conclude by presenting a selection of open problems regarding the $\mathsf{XQP}$ model.

\

\noindent \textbf{Can `average-case' $\mathsf{XQP}$ circuits (and especially $\mathsf{XQP}(\pi/4)$ circuits) be efficiently classically sampled up to constant error in total variation distance?}

Provable average-case hardness of classical simulation up to additive error (usually represented in terms of total variation distance) is a significantly stronger result than that for multiplicative error, ruling out a more plausible model of how the output distribution of a quantum circuit could be sampled in practice, and thus forming a robust foundation for demonstrations of quantum computational supremacy \citep{Hangleiter_2023}. Investigating this question is an obvious next step in the study of $\mathsf{XQP}$. Given the practical ease of implementing the Heisenberg exchange interaction, a resolution in the negative would give a prescription for easy, theoretically-robust quantum advantage experiments available with near-term devices. 

We expect that a disproof of this question could be achieved with a similar method to \citep{Aaronson_2010} for $\mathsf{BosonSampling}$, \citep{Bremner_2016} for $\mathsf{IQP}$, \citep{Fefferman_2015} for Quantum Fourier Sampling, or \citep{Morimae_2017} for $\mathsf{DQC1}$. The strategy for showing hardness of additive-error simulation is typically as follows:
\begin{enumerate}
    \item A quantity with provable $\# \mathsf{P}$-hardness of approximation is identified, and it is shown that the restricted model of computation can approximate such a quantity without postselection. This gives a worst-case hardness result for estimating output probabilities of the model up to multiplicative error. This is then used to construct an \textit{average-case hardness conjecture}, stating that for a sufficiently large fraction of random instances of circuits in the model, approximating the output probabilities up to multiplicative error is $\# \mathsf{P}$-hard. In the case of $\mathsf{IQP}$ circuits this is the partition function of a complex Ising model \citep{Fujii_2017,Bremner_2016} or the gap of degree-three polynomials over $\mathbb{F}_2$. 
    \item The \textit{anti-concentration} of the output distribution of the model is conjectured (or proven \cite{Koehler_2026}), meaning that for a sufficiently large fraction of random instances of circuits in the restricted model, the output probabilities do not concentrate on a small number of outcomes. This can be given in terms of an upper or lower bound on the probabilities of the output distribution of $\mathcal{O}(2^{-n})$, which implies that no single outcome has a probability that is significantly larger than the uniform distribution. 
    \item The \textit{Stockmeyer counting algorithm} is applied, in order to show that given the ability to sample from the output distribution of the model up to additive error, there exists an $\mathsf{FBPP^{NP}}$ algorithm to approximate the $\# \mathsf{P}$-hard quantity up to multiplicative error for a large fraction of problem instances, implying the collapse of the polynomial hierarchy. Therefore, if the polynomial hierarchy does not collapse, then weak classical simulation of the model up to additive error is impossible.
\end{enumerate}

\noindent For a detailed survey of techniques used in proving hardness of classical simulations of restricted quantum circuits, we recommend \citep{Hangleiter_2023}.

The first two points are interesting open problems in their own right. In Appendix \ref{sec:potts} we have unsuccessfully attempted to identify a $\# \mathsf{P}$-hard quantity (the partition function of a complex Potts model \citep{Galanis_2021}) which could be encoded by $\mathsf{XQP}$ circuits without phase gadgets, by means of a six-vertex model. However, this strategy could still plausibly work -- the computational complexity of exactly computing the partition function of a complex six-vertex model was considered in \citep{Cai_2017}, where it was shown to be either polynomial-time computable or $\mathsf{\#P}$-hard depending on the choice of vertex weights. The question of approximation complexity was considered in \citep{Cai_2017b}, however this was only done for the case of real vertex weights and does not apply to the $\mathsf{XQP}$ model. Therefore, the first point could be resolved by independently studying the approximation complexity of six-vertex model partition functions with weights $a(\theta) = e^{i \theta}$, $b(\theta) = i \sin \theta$, $c(\theta) = \cos \theta$, and fixed boundary conditions.

Another possibility of resolving point (1) is to consider the probability amplitudes of $\mathsf{XQP}(\pi/4)$ circuits. Writing each gate as $U(\pi/4)_{ij} = 2^{-1/2} \sum_{p =0}^1 i^p E^p_{ij}$, we may express the measurement amplitudes of $\sqrt{\mathrm{SWAP}}$ circuits with $N$ gates as:

\begin{equation} \label{eq:weightenumerator}
\langle \mathbf{y} |U | \mathbf{x} \rangle = 2^{-N/2} \sum_{\mathbf{p} \in \{0, 1\}^N} i^{|\mathbf{p}|} f(\mathbf{p}),
\end{equation}

\noindent where $f(\mathbf{p}) = \langle \mathbf{y} | \prod_{k=1}^N E_{i_k j_k}^{p_k} | \mathbf{x} \rangle$ is an (efficiently computable) function which returns $1$ if the product of transpositions $E_{i_k j_k}^{p_k}$ maps the bitstring $\mathbf{x}$ to $\mathbf{y}$, and $0$ otherwise. This resembles a weight enumerator polynomial, the computation of which is known to be $\mathsf{NP}$-hard, as well as hard to approximate \citep{Vyalyi_2003}. However, as the set of bitstrings $\{ \mathbf{p} : f(\mathbf{p}) = 1 \}$ does not form a linear code the known hardness results do not apply. Evaluating the classical complexity of approximating quantities of the form of Equation \eqref{eq:weightenumerator} would resolve point (1). 

\

\noindent \textbf{What is the circuit depth and structure required for random $\mathsf{XQP}(\pi/4)$ circuits to form $t$-designs?}

The possibility of forming $t$-designs with random $\mathsf{XQP}(\pi/4)$ circuits was established in Corollary \ref{cor:designs}. However, this does not give any indication of what such circuits would look like. For example, can $t$-designs be formed with $\sqrt{\mathrm{SWAP}}_{ij}$ gates applied to random (local, or non-local) pairs of qubits? Or, would one need to form more complicated circuits with randomly applied generators $G_{ij, kl}$ of Theorem \ref{thm:semiuniversal}? Finally, what is the exact scaling of $t$ in terms of the circuit size? We pursue these questions in future work. 

\

\noindent \textbf{Do our results hold for $\mathsf{XQP}(\pi/k)$, where $k \neq 4N$?}

As an example, it would be interesting to consider circuits with gates restricted to integer powers of $U(\pi / 6) = \sqrt[3]{\mathrm{SWAP}}$, which can implement phase gates of the form $\text{diag}(1, e^{i \pi / 3})$ via the phase gadget, but cannot produce singlet states in an obvious way. More generally, an interesting question to pursue is the conditions on $\theta$ required for the results of this work to apply to $\mathsf{XQP}(\theta)$.

\

\noindent \textbf{Does $\mathsf{XQP}$ naturally contain any other restricted model of quantum computation?}

This question is motivated by containment results for other restricted models, and concerns \textsf{XQP} circuits without phase gadgets (which we conjecture are not computationally universal). For example, $\mathsf{DQC1}$ circuits are able to calculate $\mathsf{IQP}$ probabilities \citep{Morimae_2017} and $\mathsf{QBall}$ amplitudes \citep{Aaronson_2016}. A generalisation of $\mathsf{DQC1}$ called $\mathsf{\frac{1}{2} BQP}$ \citep{Aaronson_2016, Jacobs_2024} is also able to efficiently sample from $\mathsf{IQP}$ circuits \citep{Jacobs_2024}. We have found it difficult to find any such containments for $\mathsf{XQP}$, primarily due to the fact that arbitrary $\mathsf{XQP}$ amplitudes cannot be easily expressed as a trace of a unitary. Also, it would be interesting to consider the simulation hardness for `hybrid' models involving $\mathsf{XQP}$, as was done in \citep{Fujii_2018} for $\mathsf{IQP}$--$\mathsf{DQC1}$. 

\

\noindent \textbf{Are exchange gates universal on any input basis?}

Previous work \cite{Kempe_2001} established universality of the exchange interaction $U(\theta)$ on singlet states, or more broadly any input supported on a single irrep $V_J$ of the symmetric group. Theorem \ref{thm:sqrtswap_universal} can be seen as an extension of this to any bitstring input. An open question is whether there exist non-trivial inputs (i.e. not bitstring states with low Hamming weight) for which the exchange interaction ceases to be computationally universal, as well as what other bases of input states could be used to achieve computational universality.

\acknowledgements
We thank Richard Jozsa for insightful discussions. We also thank two anonymous referees, whose feedback led us to discover the proof of Theorem \ref{thm:sqrtswap_universal}. J.B. gratefully acknowledges the hospitality of the ICTQT and the University of Gdańsk, where this work was conceived. S.S. acknowledges support from the Royal Society University Research
Fellowship. S.S. also acknowledges support from the Polish National Science Centre (NCN) under the SONATA BIS grant no. UMO-2024/54/E/ST2/00316. M.S. acknowledges support by the IRA Programme, project no. FENG.02.01-IP.05-0006/23, financed by the
FENG program 2021-2027, Priority FENG.02, Measure FENG.02.01, with the support of the FNP. 

\appendix

\section{Alternative Proof of Theorem 1} \label{appendix:alternativeproof}

Here, we restate and prove Theorem \ref{thm:postxqpbqp} using an alternative method based on $\mathsf{DFS}_3$ circuits. 

\postxqpbqp*

\begin{proof}[Alternative Proof]
Consider an $\mathsf{XQP}(\pi/4)$ circuit initialised to the state $|010\rangle^{\otimes n} |01\rangle^{\otimes 2n}$, where the first $3n$ qubits are the computational register, and the remaining $4n$ qubits are ancillae for the phase gadgets. Apply $n$ many exchange interactions $U(\pi/4)_{12} U(\pi/4)_{45} \hdots U(\pi/4)_{3n - 2, 3n - 1}$ to the computational register. This produces the state $ \frac{1}{\sqrt{2}} \left( |010 \rangle + i |100 \rangle \right)^{\otimes n}$. Applying the phase gadgets sequentially as $\mathcal{G}(3\pi/4)_{1} \circ \mathcal{G}(3\pi/4)_{4} \circ \hdots \circ \mathcal{G}(3\pi/4)_{3n - 2}$ produces the state:

\begin{equation}
    \frac{1}{\sqrt{2}} \left( |01 0 \rangle - |100 \rangle \right)^{\otimes n} = |0_L \rangle^{\otimes n}
\end{equation}

\noindent Therefore, following postselection of $|01\rangle^{\otimes n}$ on the first half of the ancilla register, we have a logical $|0_L \rangle^{\otimes n}$. By Lemma \ref{lemma:dfs3universal}, we can apply an appropriate sequence of $\sqrt{\mathrm{SWAP}}$ gates to this state to perform any $\mathsf{BQP}$ algorithm, requiring no postselections.

To show that $\mathsf{PostBQP} \subseteq \mathsf{PostXQP}(\pi/4)$, we need to be able to re-express any $\mathsf{PostBQP}$ algorithm with a $\mathsf{PostXQP}(\pi/4)$ circuit. This requires the ability to measure in the $\{ |0_L \rangle, |1_L \rangle \}$ basis, where:
\begin{align}
    |0_L \rangle &= \frac{1}{\sqrt{2}}(|010 \rangle - |100 \rangle )
 \\
 |1_L \rangle &= -\frac{2}{\sqrt{6}}|001 \rangle + \frac{1}{\sqrt{3}} \frac{1}{\sqrt{2}}(|010\rangle + |100 \rangle).
\end{align}

\noindent Applying the phase gadget $\mathcal{G}(3\pi/4)_{1}$, this gets us:
\begin{align}
    |0_L \rangle &\rightarrow \frac{1}{\sqrt{2}}(|010 \rangle - i |100 \rangle )
 \\
 |1_L \rangle &\rightarrow -\frac{2}{\sqrt{6}}|001 \rangle + \frac{1}{\sqrt{3}} \frac{1}{\sqrt{2}}(|010\rangle + i |100 \rangle)
\end{align}

\noindent Following with a $U(\pi/4)_{12} = \frac{1}{\sqrt{2}} \mathbf{1} + \frac{i}{\sqrt{2}} E_{12}$, we get: 
\begin{align}
    |0_L \rangle &\rightarrow |010 \rangle
 \\
 |1_L \rangle &\rightarrow -\frac{2e^{i \pi/4}}{\sqrt{6}}|001 \rangle + \frac{i}{\sqrt{3}} |100 \rangle
\end{align}

\noindent Therefore, after applying $U(\pi/4)_{12} \cdot \mathcal{G}(3\pi/4)_{1}$ on each logical qubit, the two logical basis states can be perfectly distinguished with a single physical measurement on the middle qubit. A measurement of $|1\rangle$ denotes $|0_L \rangle$, whereas a measurement of $|0\rangle$ denotes $|1_L \rangle$. In other words, for a physical measurement of a triple $x_1 x_2 x_3$, $|x_L \rangle = |f(x_1, x_2, x_3)\rangle$ where $f$ is some Boolean function (for example $f(x_1, x_2, x_3) = \lnot x_2$, or $f(x_1, x_2, x_3) = x_1 \lor \lnot x_2 \lor x_3$). Similarly, any postselections in the original $\mathsf{PostBQP}$ circuit can be implemented by postselecting on the corresponding logical qubits, conditioned on additional postselections on the phase gadget ancillae. Therefore, $\mathsf{PostXQP}(\pi/4) = \mathsf{PostBQP}$.
\end{proof}

For completeness, we also give a proof of Corollary \ref{cor:xqpsimulation} using the $\mathsf{DFS}_3$ encoding:

\xqpsimulation*

\begin{proof}
    Let $L \in \mathsf{PostXQP}(\pi/4)$ be any language decided with bounded error by a uniform family of postselected $\mathsf{XQP}(\pi/4)$ circuits $C_w$ with a single-line output register $\mathcal{O}_w$ (i.e., the middle qubit of the $3$-qubit computational register) and postselection registers $\mathcal{P}_w$ (which activate the correct postselections and phase gadgets whenever $\mathcal{P}_w = \mathbf{p}$ ). Now introduce the conditional probability:
    \begin{equation}
        S_w(x) = \frac{p(\mathcal{O}_w = x \ \& \ \mathcal{P}_w = \mathbf{p})}{p( \mathcal{P}_w = \mathbf{p})}.
    \end{equation}

\noindent By the definition of $\mathsf{PostXQP}(\pi/4)$, for $0 < \delta < \frac{1}{2}$:
\begin{enumerate}
    \item[(i).] if $w \in L$, then $S_w(1) \geq \frac{1}{2} + \delta$,
    \item[(ii).] if $w \notin L$, then $S_w(1) \leq \frac{1}{2} - \delta$,
\end{enumerate}

\noindent Now, let $\mathcal{Y}_w$ denote the \textit{full} output register of $C_w$: this consists of qubits we have included in $\mathcal{O}_w$, $\mathcal{P}_w$, and any qubits we ignore at the end. If the output measurement of this register can be weakly simulated to multiplicative error $c$, then there is a uniform family of classical randomised circuits $\tilde{C}_w$ with output register $\tilde{\mathcal{Y}}_w$. Multiplicative weak simulation implies that:
\begin{equation}
    \frac{1}{c} p(\mathcal{Y}_w = \mathbf{y} ) \leq p(\tilde{\mathcal{Y}}_w = \mathbf{y}) \leq c p(\mathcal{Y}_w = \mathbf{y} )
\end{equation}

\noindent The inequality also holds for all marginal distributions. We can therefore sample the postselection output first, and proceed only if $\tilde{\mathcal{P}}_w = \mathbf{p}$ (discarding the outcome otherwise). Then, we sample the middle qubit of the $3$-qubit computational output. This allows us to sample from $\tilde{S}_w(x)$, where:

\begin{equation}
    \tilde{S}_w(x) = \frac{p(\tilde{\mathcal{O}}_w = x \ \& \ \tilde{\mathcal{P}}_w = \mathbf{p})}{p( \tilde{\mathcal{P}}_w = \mathbf{p})}.
\end{equation}

\noindent By the inequalities above, this satisfies:

\begin{equation}
    \frac{1}{c^2} S_w(x)  \leq \tilde{S}_w(x) \leq c^2 S_w(x)
\end{equation}

\noindent Thus our classical algorithm decides whether $w$ belongs in $L$ with bounded error if $c^2 < 1 + 2 \delta$. If we take $c < \sqrt{2}$, this means that $\tilde{S}_w(x)$ satisfies the definitions of $\mathsf{PostBPP}$, meaning that $L \in \mathsf{PostBPP}$. But $L \in \mathsf{PostXQP}(\pi/4) = \mathsf{PostBQP} = \mathsf{PP}$ \cite{Aaronson_2004}, so $\mathsf{PostBPP = PP}$. Along with Toda's theorem \cite{Toda_1991}, this implies that $\mathsf{PH} = \Delta_3$. \end{proof}

\section{Irrep Structure of XQP Input States} \label{sec:schur_states}

Here, we analyse the action of $\mathsf{XQP}$ circuits by considering the irrep subspace decomposition of the computational basis $|\mathbf{x} \rangle$ used in their state preparation and measurement. All $\mathsf{XQP}$ unitaries are generated by elements of the symmetric group algebra $\mathbb{C} S_n$, and we therefore write them as:

\begin{equation}
    U = \prod_{k}^N (\cos \theta_k \cdot \mathbf{1} + i \sin \theta_k \cdot E_{i_k j_k}) = \sum_{\sigma \in S_n} \alpha_\sigma \sigma,
\end{equation}

\noindent where $\alpha_\sigma$ are complex coefficients. Clearly, $U$ forms a reducible representation of $S_n$ on the vector space $(\mathbb{C}^2)^{\otimes n}$, which we decompose into irreducible representations (c.f. Definition \ref{def:schur_weyl}) as:

\begin{equation}
    (\mathbb{C}^2)^{\otimes n} \cong \bigoplus_{J \in \mathcal{J}_n} V_J \otimes W_J.
\end{equation}

\noindent The basis which adheres to this decomposition is the Schur (or Gelfand--Tsetlin) basis, in which states are written as:
\begin{equation}
| \mathbf{J} = (j_{[1]}, j_{[2]}, \ldots, j_{[n]}); M \rangle,
\end{equation} 

\noindent where $\mathbf{J}$ is a spin coupling sequence, i.e. a sequence of $S_k$ irreducible representations which the Schur state belongs to under the group subduction $S_n \supset S_{n-1} \supset \ldots \supset S_1$, with $j_{[n]} = J$. $M$ is the $z$-component of the total angular momentum of the Schur state, or equivalently the multiplicity label of the $S_n$ irrep space. Sometimes we will also write the Schur states as $| \lambda; \mu, \nu \rangle$, where $\lambda \vdash n$ is a partition of $n$ (a Young diagram) corresponding to the $S_n$ irrep, $\mu$ is a standard Young tableau of shape $\lambda$, and $\nu$ is a semistandard Young tableau of shape $\lambda$. The label $\mu$ may be represented as a Yamanouchi symbol, and the label $\nu$ may be also represented as a Gelfand--Tsetlin pattern. Here we only consider the former, which are defined as follows:

\begin{definition}[Yamanouchi Symbol] \label{def:yamanouchi}
Let $\lambda \vdash n$ be a partition of $n$ into at most $m$ parts, and let $\mu$ be a standard Young tableau of shape $\lambda$. The Yamanouchi symbol of $\mu$, written $\mathbf{Y} = (Y_1, Y_2, \ldots, Y_n)$, is a sequence of integers where $Y_k \in \{ 0, \ldots, m-1 \}$ is the index of the row in which the number $k$ appears in $\mu$, and $Y_k=0$ denotes the top row. We say that a Yamanouchi symbol $\mathbf{Y}$ is valid if it corresponds to a standard Young tableau. For $m=2$, it satisfies the following:
\begin{enumerate}
    \item For all valid $\mathbf{Y}$, $Y_k = 0$ if $k$ appears in the top row of $\mu$, and $Y_k = 1$ if $k$ appears in the bottom row.
    \item For all valid $\mathbf{Y}$, the total number of $0$'s appearing in the first $1 \leq i \leq n$ entries of $\mathbf{Y}$ is greater than or equal to the total number of $1$'s.
    \item In terms of spin coupling sequences $\mathbf{J}$, we have $j_{[i]} = i/2 - \sum_{l=1}^i Y_l$. Consequently, for two sequences $\mathbf{J}_1$ and $\mathbf{J}_2$ with corresponding Yamanouchi symbols $\mathbf{Y}_1$ and $\mathbf{Y}_2$, we have $J_1  = J_2$ if and only if $|\mathbf{Y}_1| = |\mathbf{Y}_2|$, where $J = j_{[n]}$ and $|\mathbf{Y}|$ is the Hamming weight of $\mathbf{Y}$.
\end{enumerate}

\noindent As an example, consider the standard Young tableaux of shapes $\lambda_1 = (3)$, $\lambda_2 = \lambda_3 = (2, 1)$ and $\lambda_4 = (2, 2)$:

\begin{align}
&\ytableaushort{123} \,, & &\ytableaushort{12,3} \, , & &\ytableaushort{13,2} \, , & &\ytableaushort{13,24} \, . 
\end{align}

\noindent Their respective Yamanouchi symbols are $\mathbf{Y}_1 = (0 0 0)$, $\mathbf{Y}_2 = (0 0 1)$, $\mathbf{Y}_3 = (0 1 0)$ and $\mathbf{Y}_4 = (0 1 0 1)$, where for clarity we have omitted commas between entries. The corresponding spin coupling sequences are $\mathbf{J}_1 = (1/2, 1, 3/2)$, $\mathbf{J}_2 = (1/2, 1, 1/2)$, $\mathbf{J}_3 = (1/2, 0, 1/2)$ and $\mathbf{J}_4 = (1/2, 0, 1/2, 0)$.

\end{definition}

We begin by decomposing Schur states into the computational basis as in \citep{Havlicek_2018, Havlicek_2019}:
\begin{align}
|\mathbf{J}, M \rangle &= \sum_{m_1, m_2} \sum_{m_{[2]}, m_3} \cdots \sum_{m_{[n-1]}, m_n} C^{j_{[2]}, m_{[2]}}_{\frac{1}{2} m_1 ;\frac{1}{2}  m_2} C^{j_{[3]}, m_{[3]}}_{j_{[2]}, m_{[2]} ;\frac{1}{2}  m_3} \\
&\cdots C^{j_{[n-1]}, m_{[n-1]}}_{j_{[n-2]}, m_{[n-2]} ;\frac{1}{2}  m_{n-1}} C^{J M}_{j_{[n-1]}, m_{[n-1]} ;\frac{1}{2}  m_n} |m_1 \dots m_n \rangle,
\end{align}

\noindent where $|m_1 \dots m_n \rangle = | \mathbf{x} \rangle$ is a computational basis state with $m_i = 1/2 - x_i$, and $m_{[j]} = \sum_{k=1}^j m_k$. The scalars $C^{j_{[k]}, m_{[k]}}_{j_{[k-1]}, m_{[k-1]} ;\frac{1}{2}  m_k}$ are Clebsch--Gordan coefficients. The inner products $\langle \mathbf{x} | \mathbf{J}, M \rangle$ are simple products of Clebsch--Gordan coefficients:
\begin{equation}
    \langle \mathbf{x} |\mathbf{J}, M \rangle = C^{j_{[2]}, m_{[2]}}_{\frac{1}{2} m_1 ;\frac{1}{2}  m_2} C^{j_{[3]}, m_{[3]}}_{j_{[2]}, m_{[2]} ;\frac{1}{2}  m_3} \cdots C^{J M}_{j_{[n-1]}, m_{[n-1]} ;\frac{1}{2}  m_n}.
\end{equation}

\noindent We can therefore write $|\mathbf{x} \rangle$ in the Schur basis as:
\begin{align}
|\mathbf{x}\rangle &= \sum_{\mathbf{J}} \sum_{M = -J}^J \langle \mathbf{J}, M | \mathbf{x}  \rangle |\mathbf{J}, M \rangle \\&= 
\sum_{\mathbf{J}} \sum_{M = -J}^J C^{j_{[2]}, m_{[2]}}_{\frac{1}{2} m_1 ;\frac{1}{2}  m_2} C^{j_{[3]}, m_{[3]}}_{j_{[2]}, m_{[2]} ;\frac{1}{2}  m_3} \cdots \\
& \cdots C^{j_{[n-1]}, m_{[n-1]}}_{j_{[n-2]}, m_{[n-2]} ;\frac{1}{2}  m_{n-1}} C^{J M}_{j_{[n-1]}, m_{[n-1]} ;\frac{1}{2}  m_n} |\mathbf{J}, M \rangle.
\end{align}

\noindent The sum over $\mathbf{J}$ can be simplified using conservation of angular momentum: at each step, $j_{[k]} \leq k/2$, and $|m_{[k]}| \leq j_{[k]}$. Writing $\mathbf{M} = (m_{[1]}, m_{[2]}, \ldots, m_{[n]} = M)$, we write $\mathbf{J} \geq |\mathbf{M}|$ to denote sequences $\mathbf{J}$ which satisfy the constraints:
\begin{equation}
    \frac{i}{2} \geq j_{[i]} \geq |m_{[i]}| \quad \forall i \in \{ 1, 2, \ldots, n \}. \label{eq:jmconstraints}
\end{equation}

\noindent With this notation, we can write:
\begin{align}
|\mathbf{x}\rangle &= 
\sum_{\mathbf{J} \geq |\mathbf{M}|} C^{j_{[2]}, m_{[2]}}_{\frac{1}{2} m_1 ;\frac{1}{2}  m_2} C^{j_{[3]}, m_{[3]}}_{j_{[2]}, m_{[2]} ;\frac{1}{2}  m_3} \cdots \\
& \cdots C^{j_{[n-1]}, m_{[n-1]}}_{j_{[n-2]}, m_{[n-2]} ;\frac{1}{2}  m_{n-1}} C^{J M}_{j_{[n-1]}, m_{[n-1]} ;\frac{1}{2}  m_n} |\mathbf{J}, M \rangle.
\end{align}

\noindent Clearly, the constraints of Equation \eqref{eq:jmconstraints} are always satisfied when $\mathbf{J} = \mathbf{M}$, leading to the following Lemma:
\begin{lemma} \label{lem:decomposition}
    A computational basis state $|\mathbf{x} \rangle$ satisfies $\langle \mathbf{x} | \Pi_J | \mathbf{x} \rangle \neq 0$ for all $J \in \{ |M|, |M|+1, \ldots, n/2 \}$, where $\Pi_J$ is the projector onto the isotypic component $(V_J \otimes W_J)$, and $M = n/2 - |\mathbf{x}|$. For any two computational basis states $|\mathbf{x} \rangle$, $|\mathbf{y} \rangle$ such that $|\mathbf{x}| = |\mathbf{y}|$, we have $k^2_{J, M} = \langle \mathbf{x} | \Pi_J | \mathbf{x} \rangle = \langle \mathbf{y} | \Pi_J | \mathbf{y} \rangle$.
\end{lemma}
\begin{proof}
    Without loss of generality, assume $M \geq 0$. For a Hamming weight $|\mathbf{x}| = n/2 - M$, pick the computational basis state $| \mathbf{x} \rangle = |0^{n/2 + M} 1^{n/2 - M} \rangle$. The corresponding sequence of $m_{[i]}$ is $\mathbf{M} = (1/2, 1, \ldots, n/4 + M/2, n/4 + M/2 - 1/2, \ldots, M)$. Up to the $(n/2 + M)^{th}$ step of this sequence we have $\mathbf{J} = \mathbf{M}$, after which $j_{[i]}$ can keep increasing up to $n/2$, or decrease down with $\mathbf{M}$ to $j_{[n]} = M$. In fact, the sequence $\mathbf{J} = \mathbf{M}$ is the only sequence ending with $j_{[n]} = M$ which satisfies the constraints of Equation \eqref{eq:jmconstraints}, and therefore for our chosen $|\mathbf{x} \rangle$, we have $\Pi_J | \mathbf{x} \rangle = k_{M, M} |\mathbf{J} = \mathbf{M}, M \rangle$. For a different bitstring $|\mathbf{y} \rangle$ with the same Hamming weight, notice that $\langle \mathbf{y} | \Pi_J | \mathbf{y} \rangle = \langle \mathbf{x} | \sigma^\dagger \Pi_J \sigma | \mathbf{x} \rangle$ for some $\sigma \in S_n$. Since $\Pi_J$ is a projector onto an isotypic component, in the Schur basis it forms an identity on the $(V_J \otimes W_J)$ subspace, and therefore commutes with all $\sigma \in S_n$.
\end{proof}

We note that the values of the coefficients $k_{J, M}^2$ can be easily calculated:
\begin{align}
k_{J, M}^2 &= \langle \mathbf{x} | \Pi_J | \mathbf{x} \rangle \\
&= \frac{1}{\binom{n}{\frac{n}{2}- M}} \sum_{|\mathbf{y} | = \frac{n}{2} - M} \langle \mathbf{y} | \Pi_J | \mathbf{y} \rangle \\
&= \frac{1}{\binom{n}{\frac{n}{2}- M}} \dim V_J,
\end{align} 

\noindent where $\dim V_J$ has a closed-form expression given in Definition \ref{def:schur_weyl}. Notably, if $n$ is even and $|\mathbf{x}| = n/2$, then $k_{0, 0}^2 = (n/2 + 1)^{-1}$.

With the introduced irrep decomposition, we can re-write $\mathsf{XQP}$ input states $|\mathbf{x} \rangle$ as:
\begin{equation}
|\mathbf{x} \rangle = \sum_{J} k_{J, M} | \phi_J \rangle,
\end{equation}
where $|\phi_J \rangle = \frac{1}{k_{J, M}} \Pi_J | \mathbf{x} \rangle$ is a normalised state in the $V_J$ irrep subspace. Any $\mathsf{XQP}$ unitary $U$ acting on this state will preserve the relative amplitudes $k_{J, M}$, therefore no $\mathsf{XQP}$ circuit can produce a Schur state, or a state with support exclusively on a single irrep (except $|\mathbf{0} \rangle$ and $|\mathbf{1} \rangle$, which reside in the trivial irrep). Hence, exchange-based quantum computation must always require external resources (gates, postselection or a supply of singlets) to prepare $\mathsf{DFS}_k$ logical basis states.

In spite of these restrictions, we now show that the ability to efficiently calculate $\mathsf{XQP}$ circuit \textit{amplitudes} $\langle \mathbf{y} | U | \mathbf{x} \rangle$ is sufficient to efficiently calculate the amplitudes of any quantum circuit, forming an additional piece of evidence (alongside Theorem \ref{thm:tvxqpbqp}) that $\mathsf{XQP}$ circuits are hard to classically simulate to additive precision. We first prove two technical lemmas, before combining them with Lemma \ref{lem:decomposition} to obtain Theorem \ref{thm:amplitudes}.

\begin{lemma} \label{lem:amplitudes}
Suppose we can classically efficiently calculate quantities $a_{\mathbf{y}, \mathbf{x}}^{U}$ such that $|a_{\mathbf{y}, \mathbf{x}}^{U} - \langle \mathbf{x} | U | \mathbf{y} \rangle| \leq \epsilon$ for an $\mathsf{XQP}$ circuit $U$ with $|\mathbf{x}| = n/2$ and $n$ even. Then, we can classically efficiently calculate quantities $\tilde{a}_{\mathbf{y}, \mathbf{x}}^{U}$, such that $|\tilde{a}_{\mathbf{y}, \mathbf{x}}^{U} - \langle \phi_{0}^{(\mathbf{x})} | U  | \phi_0^{(\mathbf{y})}\rangle| \leq  (n/2 + 1) \epsilon$.
\end{lemma}
\begin{proof}

We will make use of the Hamiltonian $T = \sum_{i < j} E_{ij}$, which is a sum of all exchange interactions. It is proportional to the total angular momentum operator $\hat{J}^2$, where:
\begin{equation}
    \hat{J}^2 = \sum_{i < j} E_{ij} + \left(n - \frac{n^2}{4} \right) \mathbf{1}.
\end{equation}

\noindent Take the smallest integer $Q$ such that $Q > \frac{n}{2}(\frac{n}{2} + 1)$, and define the gates $T(r) = e^{2 \pi i r T / Q}$ for $0 \leq r < Q$. We have that:
\begin{equation}
    T(r)| \phi_J \rangle = e^{2 \pi i r c_n / Q} e^{2 \pi i r J(J+1) / Q} |\phi_J \rangle,
\end{equation}

\noindent where $c_n = n^2/4 - n$. This unitary can be synthesised with an $\mathsf{XQP}$ circuit to high approximation using the Lie--Trotter formula. The weighted sum of $T(r)$ gates over $r$ acts as a projector onto the $J=0$ irrep subspace:
\begin{equation}
\Pi_0 = \frac{1}{Q} \sum_{r = 0}^{Q-1} e^{-2 \pi i r c_n / Q} T(r).
\end{equation}
Indeed, for a state $|\phi_J \rangle$ in the $V_J$ irrep subspace:
\begin{equation}
\Pi_0 |\phi_J \rangle = \frac{1}{Q} \sum_{r = 0}^{Q-1} e^{2 \pi i r J(J+1) / Q} |\phi_J \rangle = \delta_{J, 0} |\phi_J \rangle.
\end{equation}

\noindent Given some $\mathsf{XQP}$ unitary $U$, we may compute $\langle \mathbf{x} | U T(r) | \mathbf{x} \rangle$. With the irrep decomposition of $|\mathbf{x} \rangle$:
\begin{equation}
    \langle \mathbf{x} | U T(r) | \mathbf{x} \rangle = e^{2 \pi i r c_n / Q} \sum_J k_{J, 0}^2 e^{2 \pi i r \frac{J(J+1)}{Q}} \langle \phi_{J}^{(\mathbf{x})} | U | \phi_{J}^{(\mathbf{x})}\rangle. 
\end{equation}

\noindent Writing $b_{J, 0} = k_{J, 0}^2 \langle \phi_{J}^{(\mathbf{x})} | U | \phi_{J}^{(\mathbf{x})} \rangle$, we now take the weighted sum of $\langle \mathbf{x} | U T(r) | \mathbf{x} \rangle$ over $r$:
\begin{align}
    \frac{1}{Q} \sum_{r = 0}^{Q-1} & e^{-2 \pi i r c_n / Q} \langle \mathbf{x} | U T(r) | \mathbf{x} \rangle  \\
    &= \frac{1}{Q } \sum_J b_{J, 0} \left( \sum_{r = 0}^{Q-1} e^{2 \pi i r\frac{J(J+1)}{Q}}\right) \\
    & = k_{0, 0}^2 \langle \phi_{0}^{(\mathbf{x})} | U | \phi_{0}^{(\mathbf{x})} \rangle.
\end{align}

\noindent We then divide by $k_{0, 0}^2 = (n/2 + 1)^{-1}$ to recover the inner product. Similarly, we can find $\langle \phi_{0}^{(\mathbf{x})} | U  | \phi_0^{(\mathbf{y})}\rangle$ by setting $U \rightarrow U \sigma$, where $\sigma$ is a sequence of transpositions such that $\sigma |\mathbf{x} \rangle = |\mathbf{y} \rangle$. 

If $|a_{\mathbf{y}, \mathbf{x}}^{U} - \langle \mathbf{x} | U | \mathbf{y} \rangle| < \epsilon$, then our classical estimator $\tilde{a}_{\mathbf{y}, \mathbf{x}}^{U} = \frac{(n/2 + 1)}{Q} \sum_{r = 0}^{Q-1} e^{-2 \pi i r c_n / Q} a_{\mathbf{y}, \mathbf{x}}^{U T(r)}$ satisfies:

\begin{align}
    |\tilde{a}_{\mathbf{y}, \mathbf{x}}^{U} &- \langle \phi_{0}^{(\mathbf{x})} | U  | \phi_0^{(\mathbf{y})}\rangle| \\
    &= \frac{n/2 + 1}{Q} \left| \sum_{r = 0}^{Q-1} a_{\mathbf{y}, \mathbf{x}}^{U  T(r)} - \sum_{r = 0}^{Q-1} \langle \mathbf{x} | U  T(r) | \mathbf{y} \rangle \right| \\
    &\leq   \frac{n/2 + 1}{Q} \sum_{r = 0}^{Q-1} |a_{\mathbf{y}, \mathbf{x}}^{U  T(r)} - \langle \mathbf{x} | U  T(r) | \mathbf{y} \rangle| \\
        &\leq (n/2 + 1) \epsilon.
\end{align}

\noindent Choosing $\epsilon = \delta / (n/2 + 1)$, we can then efficiently calculate $\langle \phi_{0}^{(\mathbf{x})} | U  | \phi_0^{(\mathbf{y})}\rangle$ to within additive error $\delta$.

\end{proof}

\begin{lemma} \label{lem:yamanouchi}
For any two $n$-qubit Schur states $|J_1, M_1, \mathbf{Y}_1 \rangle$, $|J_2, M_2, \mathbf{Y}_2 \rangle$ such that $J_1 = J_2$ and $M_1 = M_2$ (where $\mathbf{Y}$ is a Yamanouchi symbol), it is possible to construct a unitary $U_{\mathbf{Y}_1 \rightarrow \mathbf{Y}_2}$, such that:
\begin{equation}
U_{\mathbf{Y}_1 \rightarrow \mathbf{Y}_2} |J_1, M_1, \mathbf{Y}_1 \rangle = e^{i \phi} |J_2, M_2, \mathbf{Y}_2 \rangle,
\end{equation}
\noindent using $\mathcal{O}(\text{poly}(n))$ exchange interactions $U(\theta)$.

\end{lemma}

\begin{proof}
We begin by finding a sequence of transpositions $(i, i+1)$ permuting the entries of $\mathbf{Y}_1$ to $\mathbf{Y}_2$, such that all intermediate Yamanouchi symbols are valid. For example, the transposition $(2, 3)$ maps the Yamanouchi symbol $(010)$ into the symbol $(001)$. In general, this forms a sequence $\mathbf{Y}_{1, k} = \prod_{l=1}^k (i_l, i_l+1) \mathbf{Y}_1$, where $\mathbf{Y}_{1, N} = \mathbf{Y}_2$ for some $N$. We will now show that $N \leq \mathcal{O}(n^2)$. 

To do this, note that if $\mathbf{Y}$ is a valid Yamanouchi symbol with $Y_i = 1$ and $Y_{i+1} = 0$, then applying a transposition $(i, i+1)$ to $\mathbf{Y}$ produces another valid Yamanouchi symbol, as $(i, i+1)\mathbf{Y}$ satisfies point (2) in Definition \ref{def:yamanouchi}. Hence, we can always `shift' a $0$ to the left as: 
\begin{equation}
(Y_1 \cdots 10 \cdots Y_n) \xrightarrow{(i, i+1)} (Y_1 \cdots 01 \cdots Y_n).
\end{equation}

\noindent We will assume $n$ is even and consider the worst-case scenario to upper bound $N$. In general, we can always transform $\mathbf{Y}_1$ into $\mathbf{Y}_2$ by first shifting all $0$'s in $\mathbf{Y}_1$ to the left to produce $\tilde{\mathbf{Y}} = (0^{n/2} 1^{n/2})$, then shifting the $0$'s rightward to their correct positions to produce $\mathbf{Y}_2$. As the number of $1$'s in the first $i$ entries of $\mathbf{Y}$ cannot exceed the number of $0$'s, the first step will require the maximum number of transpositions $N_{\mathrm{max}}$ if $\mathbf{Y}_1 = (0101 \cdots 01)$. This gives $N_{\mathrm{max}} = \sum_{i=1}^{n/2} (i-1) = n(n-2)/8$. Similarly, the second step of shifting the $0$'s to the right will require the most individual transpositions if $\mathbf{Y}_2 = (0101 \cdots 01)$, which again requires at most $N_{\mathrm{max}}$ transpositions. Hence, the total number of transpositions required is upper bounded as:

\begin{equation}
N \leq 2 N_{\mathrm{max}} = \frac{1}{4} n(n-2) \leq \mathcal{O}(n^2).
\end{equation}

\noindent For every transposition $(i, i+1)$ in the sequence, we can implement the corresponding unitary $\sigma_{i, i+1}$, such that $\sigma_{i, i+1} |\mathbf{Y} \rangle = e^{i \phi_{i, i+1}} |(i, i+1) \mathbf{Y} \rangle$. This unitary just needs to act as the Pauli $X$ gate on the two-dimensional subspace spanned by $|\mathbf{Y} \rangle$ and $|(i, i+1) \mathbf{Y} \rangle$ (where we have omitted the $J, M$ indices). 

Now, consider the standard Young tableaux $\mu(\mathbf{Y})$ and $\mu((i, i+1) \mathbf{Y})$ corresponding to the two Yamanouchi symbols. These tableaux differ only in the positions of the numbers $i$ and $i+1$, which appear in distinct rows and columns. Define the content $\rho_i(\mathbf{Y}) = x_i - y_i$ where $(x_i, y_i)$ are the coordinates of the box in which $i$ appears in $\mu(\mathbf{Y})$. Building on the examples in Definition \ref{def:yamanouchi}, the coordinates of each box are:

\begin{align}
&\begin{ytableau} {\scriptstyle 0,0} & \scriptstyle 1, 0 & \scriptstyle 2, 0 \end{ytableau} \, , & & \begin{ytableau}
{\scriptstyle 0,0} & \scriptstyle 1, 0 \\ \scriptstyle 0, 1 \end{ytableau} \, , & &\begin{ytableau} {\scriptstyle 0,0} & \scriptstyle 1, 0 \\ \scriptstyle 0, 1 & \scriptstyle 1, 1   \end{ytableau} \, .
\end{align}

\noindent This gives the following contents:

\begin{align}
&\begin{ytableau} 0 & 1 & 2 \end{ytableau} \, , & & \begin{ytableau}
0 & 1 \\ -1 \end{ytableau} \, , & &\begin{ytableau} 0 & 1 \\ -1 & 0  \end{ytableau} \, ,
\end{align}

\noindent so that, for example, $\rho_1(\mathbf{Y}_2) = \rho_1(\mathbf{Y}_3) = 0$, $\rho_2(\mathbf{Y}_2) = \rho_3(\mathbf{Y}_3) = 1$, and $\rho_3(\mathbf{Y}_2) = \rho_2(\mathbf{Y}_3) = -1$ for the two Yamanouchi symbols $\mathbf{Y}_2 = (001)$ and $\mathbf{Y}_3 = (010)$ in Definition \ref{def:yamanouchi}. 

Next, define the axial distance $\Delta_{i, i+1}(\mathbf{Y}) = \rho_i(\mathbf{Y}) - \rho_{i+1}(\mathbf{Y})$. With regards to the above, $\Delta_{23}(\mathbf{Y}_2) = - \Delta_{23}(\mathbf{Y}_3) = 2$. More generally, if $\mathbf{Y}$ and $(i, i+1) \mathbf{Y}$ are valid Yamanouchi symbols, then $\mu(\mathbf{Y})$ and $\mu((i, i+1) \mathbf{Y})$ differ only in the positions of the integers $i$ and $i+1$, and we have $\Delta_{i, i+1}((i, i + 1)\mathbf{Y}) = -\Delta_{i, i+1}(\mathbf{Y})$. In Young's orthogonal form, the qubit transposition $E_{i, i+1}$ on the $\{ |\mathbf{Y} \rangle, |(i, i+1) \mathbf{Y} \rangle \}$ subspace is expressed as:

\begin{equation}
E_{i, i+1} = \begin{pmatrix} -\frac{1}{ \Delta } & \frac{\sqrt{\Delta^2 - 1}}{| \Delta |} \\ \frac{\sqrt{\Delta^2 - 1}}{| \Delta |} &  \frac{1}{ \Delta  } \end{pmatrix},
\end{equation}

\noindent where we have written $\Delta = \Delta_{i, i+1}(\mathbf{Y})$. Now also define the Young--Jucys--Murphy (YJM) elements \citep{Grinberg_2025, Vershik_2005} $X_i = \sum_{j < i} E_{ji}$, which are diagonal in Young's orthogonal form and satisfy $X_i |\mathbf{Y} \rangle = \rho_i(\mathbf{Y}) |\mathbf{Y} \rangle$. On the $\{ |\mathbf{Y} \rangle, |(i, i+1) \mathbf{Y} \rangle \}$ subspace, we then have:

\begin{align}
    X_i &= \begin{pmatrix} \rho_i(\mathbf{Y}) & 0 \\ 0 & \rho_i((i, i+1) \mathbf{Y}) \end{pmatrix}, \\ X_{i+1} &= \begin{pmatrix} \rho_{i+1}(\mathbf{Y}) & 0 \\ 0 & \rho_{i+1}((i, i+1) \mathbf{Y}) \end{pmatrix}.
\end{align}

\noindent Taking their difference, we get:

\begin{equation}
X_{i+1} - X_i = \begin{pmatrix} -\Delta & 0 \\ 0 & \Delta \end{pmatrix}.
\end{equation}

\noindent Hence, we can generate $\sigma_{i, i+1}$ with the Hamiltonian:
\begin{equation}
    H_{i, i+1} = \frac{|\Delta|}{\sqrt{\Delta^2 - 1}} \left( E_{i, i+1} - \frac{1}{\Delta^2} (X_{i+1} - X_i) \right).
\end{equation}

\noindent And overall, $\sigma_{i, i+1} = e^{\frac{i \pi}{2} H_{i, i+1}}$. Each $H_{i, i+1}$ is a simple linear combination of transpositions, so given access to exchange interactions $U(\theta)_{ij}$, each $\sigma_{i, i+1}$ in the sequence can be implemented to inverse-polynomial precision in $\mathcal{O}(\text{poly}(n))$ time using the Lie--Trotter formula. Thus, the overall unitary $U_{\mathbf{Y}_1 \rightarrow \mathbf{Y}_2} = \prod_{k=1}^N \sigma_{i_k, i_k+1}$ can be implemented with $\mathcal{O}(\text{poly}(n))$ exchange interactions.
\end{proof}

We now combine Lemmas \ref{lem:decomposition}, \ref{lem:amplitudes} and \ref{lem:yamanouchi} to show that the ability to efficiently calculate $\mathsf{XQP}$ circuit amplitudes is sufficient to efficiently calculate the amplitudes of any quantum circuit. In contrast to Theorem \ref{thm:tvxqpbqp}, this approach does not rely on phase gadgets or postselection.

\begin{theorem} \label{thm:amplitudes}
    Suppose we can classically efficiently calculate quantities $a^{U}$, such that $|a^{U} - \langle 0^{2n} 1^{2n} | U | 0^{2n} 1^{2n} \rangle| \leq \epsilon$ for an $\mathsf{XQP}$ circuit $U$ on $4n$ qubits. Then, we can classically efficiently calculate quantities $\tilde{a}^{U'}$, such that $|\tilde{a}^{U'} - \langle 0^n | U' | 0^n \rangle| \leq \epsilon \cdot \text{poly}(n)$, where $\langle 0^n | U' | 0^n \rangle$ is the amplitude of any $\mathsf{BQP}$ circuit $U'$ on $n$ qubits.
\end{theorem}

\begin{proof} From Lemma \ref{lem:decomposition}, it follows that for the $4n$-qubit state $|0^{2n} 1^{2n} \rangle$, $\Pi_0 |0^{2n} 1^{2n} \rangle = k_{0, 0} | \phi_0 \rangle$, where $|\phi_0 \rangle = |\mathbf{J} = (1/2, \cdots, n, \cdots, 0), M = 0 \rangle$ is a single Schur state. The corresponding Yamanouchi symbol for this state is $|\phi_0 \rangle = |\mathbf{Y} = (0^{2n} 1^{2n}) \rangle$. From Lemma \ref{lem:yamanouchi}, we can construct a sequence $U_{\mathbf{Y} \rightarrow \bar{\mathbf{Y}}}$, such that:

\begin{align}
U_{\mathbf{Y} \rightarrow \bar{\mathbf{Y}}} |\phi_0 \rangle &= U_{\mathbf{Y} \rightarrow \bar{\mathbf{Y}}} |\mathbf{Y} = (0^{2n} 1^{2n}) \rangle \\
&= e^{i \phi} |\bar{\mathbf{Y}} = (0101 \cdots 01) \rangle \\
&= e^{i \phi} |S \rangle^{\otimes 2n},
\end{align}

\noindent where $|S \rangle = (|01 \rangle - |10 \rangle)/\sqrt{2}$ is the singlet state. Hence, we can map $|\phi_0 \rangle$ to the $\mathsf{DFS}_4$ state $|0^n_L \rangle = |S \rangle^{\otimes 2n}$, and apply a sequence of exchange interactions $U_{\mathrm{DFS}}$ to implement any $\mathsf{BQP}$ circuit $U'$ on the logical qubits. We therefore have:

\begin{equation}
    \langle \phi_0 | U_{\mathbf{Y} \rightarrow \bar{\mathbf{Y}}}^\dagger U_{\mathrm{DFS}} U_{\mathbf{Y} \rightarrow \bar{\mathbf{Y}}} | \phi_0 \rangle = \langle 0^n | U' | 0^n \rangle.
\end{equation}

\noindent The corresponding $\mathsf{XQP}$ amplitude is then given by:

\begin{equation}
    \langle 0^{2n} 1^{2n} | U_{\mathbf{Y} \rightarrow \bar{\mathbf{Y}}}^\dagger U_{\mathrm{DFS}} U_{\mathbf{Y} \rightarrow \bar{\mathbf{Y}}} | 0^{2n} 1^{2n} \rangle.
\end{equation}

\noindent Using Lemma \ref{lem:amplitudes}, we can extract the desired $\mathsf{BQP}$ amplitude $\langle \phi_0 | U_{\mathbf{Y} \rightarrow \bar{\mathbf{Y}}}^\dagger U_{\mathrm{DFS}} U_{\mathbf{Y} \rightarrow \bar{\mathbf{Y}}} | \phi_0 \rangle$ out to within additive error, by setting:

\begin{equation}
U = T(r) U_{\mathbf{Y} \rightarrow \bar{\mathbf{Y}}}^\dagger U_{\mathrm{DFS}} U_{\mathbf{Y} \rightarrow \bar{\mathbf{Y}}}, \quad r \in \{ 0, 1, \ldots, Q-1\},
\end{equation} 

\noindent and taking the weighted sum over $r$ (where $Q > 2n(2n+1)$). Hence, we can efficiently calculate $\langle 0^n | U' | 0^n \rangle$ to within additive error $\epsilon \cdot \text{poly}(n)$. By Lemmas \ref{lem:amplitudes} and \ref{lem:yamanouchi}, the additional overhead to implement $U_{\mathbf{Y} \rightarrow \bar{\mathbf{Y}}}$ and $T(r)$ using exchange gates is $\mathcal{O}(\text{poly}(n))$.

\end{proof}

\section{The Entangling Power of Exchange Gates} \label{sec:entangling_power}

We now quantify the entangling capability of the exchange gates $U(\theta)$, as well as the special case of the $U(\pi/4) = \sqrt{\mathrm{SWAP}}$ gate. This analysis clarifies the strength of the exchange interaction as a primitive for entanglement generation, and forms a first step towards characterising the power of low-depth $\mathsf{XQP}$ circuits. We investigate this further in future work.  

To quantify the entangling power, we follow the Zanardi--Zalka--Faoro approach~\cite{ZanardiZalkaFaoro2000}.
The entangling power is defined as the Haar average over product inputs:
\begin{equation} \label{eq:ep0-def}
e_{p_0}(U)
:=
\mathbb{E}_{|\psi\rangle,|\phi\rangle}
\!\left[
E_{\mathrm{lin}}\big(U(|\psi\rangle\otimes|\phi\rangle)\big)
\right].
\end{equation}

\noindent Here $|\psi\rangle \in \mathcal H_A$ and $|\phi\rangle \in \mathcal H_B$ are independent
Haar-random one-qubit pure states. The linear entropy is given by:
\begin{align} \label{eq:lin-entropy-def}
E_{\mathrm{lin}}(|\Psi\rangle)
&=
1 - \operatorname{Tr}(\rho_A^2), &
\qquad
\rho_A &= \operatorname{Tr}_B(|\Psi\rangle\langle\Psi|),
 \end{align}

\noindent where $|\Psi\rangle$ is a pure two-qubit state. We will calculate $e_{p_0}$ as well as $E_{\mathrm{lin}}$ for $U(\theta)$ as a function of the pulse angle $\theta$. For $\mathcal{H}=\mathcal{H}_A\otimes \mathcal{H}_B$ with $\dim \mathcal{H}_A=\dim \mathcal{H}_B=2$, and $\theta\in\mathbb{R}$ we take the exchange interaction $U(\theta)$ in the basis $\{|00\rangle,|01\rangle,|10\rangle,|11\rangle\}$ to be the matrix given in Equation \eqref{eq:Utheta-def}. At this point we make a short remark about the phase convention:

\begin{remark}
Two matrix representations of $U(\theta)$ that differ only by multiplication by a scalar $e^{i \phi}$ represent the same physical operation, as physical two-qubit gates are represented by elements of the projective unitary group $PU(4)=U(4)/U(1)$. In particular, distinct but algebraically equivalent coefficient forms represent the same gate whenever the full matrices differ by only a single common phase factor. Such representatives have identical entangling and operational properties.
\end{remark}

\begin{lemma}
\label{lem:zanardi-ep-partial-swap}
Let $U(\theta)$ be the exchange gate defined as in~\eqref{eq:Utheta-def}. Its entangling power~\eqref{eq:ep0-def} equals:
\begin{equation}
\label{eq:ep0-Utheta}
e_{p_0}\big(U(\theta)\big)=\frac{1-\cos(4\theta)}{12}=\frac{\sin^2(2\theta)}{6}.
\end{equation}

\noindent In particular,
\begin{equation}
e_{p_0}\!\big(\sqrt{\mathrm{SWAP}}\big)=e_{p_0}\!\big(U(\pi/4)\big)=\frac{1}{6}.
\end{equation}

\noindent For the normalised linear entropy $\widetilde E_{\mathrm{lin}}:=2(1-\mathrm{Tr}(\rho_A^2))$, the corresponding average is $2e_{p_0}(U(\theta))=(1-\cos4\theta)/6$.
\end{lemma}

\begin{proof}
We follow Proposition~1 of \cite{ZanardiZalkaFaoro2000}. Consider two copies of the bipartite system,
\begin{equation}
\mathcal{H}^{\otimes 2}\cong (\mathcal{H}_A\otimes \mathcal{H}_B)\otimes(\mathcal{H}_A\otimes \mathcal{H}_B),
\end{equation}

\noindent and label the tensor factors as $(1,2,3,4)=(A_1,B_1,A_2,B_2)$.
Let $T_{13}$ denote the $\mathrm{SWAP}$ operator exchanging factors $1$ and $3$ (i.e.\ $A_1\leftrightarrow A_2$), and let $T_{24}$ denote the $\mathrm{SWAP}$ exchanging factors $2$ and $4$ (i.e.\ $B_1\leftrightarrow B_2$). For $d_A=d_B=2$, Proposition~1 in \cite{ZanardiZalkaFaoro2000} yields:
\begin{equation}
\begin{split}
\label{eq:zanardi-formula}
e_{p_0}(U)
&=
1-\frac{1}{36}\Big(I_0(U)+I_1(U)\Big),\\
I_0(U)&=8+\mathrm{Tr}\!\big(U^{\otimes 2}T_{13}U^{\dagger\otimes 2}T_{13}\big),\\
I_1(U)&=8+\mathrm{Tr}\!\big(U^{\otimes 2}T_{24}U^{\dagger\otimes 2}T_{13}\big).
\end{split}
\end{equation}

\noindent We now evaluate the two traces for $U(\theta)$ given in \eqref{eq:Utheta-def}.
A direct computation gives:
\begin{align}
\mathrm{Tr}\!\big(U(\theta)^{\otimes 2} &T_{13}U(\theta)^{\dagger\otimes 2}T_{13}\big)
\\
&=4\sin^4\theta+12\cos^2\theta+4\cos^2\theta\cos(2\theta),\\
\mathrm{Tr}\!\big(U(\theta)^{\otimes 2}&T_{24}U(\theta)^{\dagger\otimes 2}T_{13}\big)\\
&=4\sin^4\theta+4\sin^2\theta+4-4\sin^2\theta\cos(2\theta).
\end{align}

\noindent Adding the two expressions and simplifying using trigonometric identities yields:
\begin{equation}
\begin{split}
\mathrm{Tr}\!\big(&U(\theta)^{\otimes 2}T_{13}U(\theta)^{\dagger\otimes 2}T_{13}\big)
+
\mathrm{Tr}\!\big(U(\theta)^{\otimes 2}T_{24}U(\theta)^{\dagger\otimes 2}T_{13}\big)\\
&=
8+12\big(\cos^4\theta+\sin^4\theta\big) \\ 
&=17+3\cos(4\theta).
\end{split}
\end{equation}

\noindent Hence, $I_0(U(\theta))+I_1(U(\theta))=16+(17+3\cos4\theta)=33+3\cos4\theta$, and substituting into \eqref{eq:zanardi-formula} gives:
\begin{equation*}
\begin{split}
e_{p_0}(U(\theta))=
\frac{1-\cos(4\theta)}{12}=
\frac{\sin^2(2\theta)}{6},
\end{split}
\end{equation*}

\noindent which proves \eqref{eq:ep0-Utheta}. Setting $\theta=\pi/4$ gives $e_{p_0}=1/6$.
\end{proof}

In the following Lemma, we establish the maximum achievable entangling power of the exchange gate $U(\theta)$.

\begin{lemma}
\label{lem:max-ep-lin-Utheta}

Let $U(\theta)$ be the exchange gate defined as in~\eqref{eq:Utheta-def}.
Define the maximum linear-entropy entangling power by:
\begin{equation}
\mathrm{ep}_{\mathrm{lin}}(U):=\max_{|\alpha\rangle,|\beta\rangle}
E_{\mathrm{lin}}\big(U(|\alpha\rangle\otimes|\beta\rangle)\big),
\end{equation}

\noindent where the maximum is over all product inputs. Then,
\begin{equation}
\label{eq:max-ep-Utheta}
\mathrm{ep}_{\mathrm{lin}}\big(U(\theta)\big)=\frac{1}{2}\sin^2(2\theta).
\end{equation}

\noindent In particular,
\begin{equation}
\mathrm{ep}_{\mathrm{lin}}\!\big(\sqrt{\mathrm{SWAP}}\big)
=\mathrm{ep}_{\mathrm{lin}}\!\big(U(\pi/4)\big)
=\frac{1}{2}.
\end{equation}

\noindent For the normalised linear entropy $\widetilde E_{\mathrm{lin}}:=2(1-\mathrm{Tr}(\rho_A^2))$, one has $\max \widetilde E_{\mathrm{lin}}=\sin^2(2\theta)$.
\end{lemma}

\begin{proof}
Write an arbitrary product input as:
\begin{align}
|\alpha\rangle&=a|0\rangle+b|1\rangle,
&
|\beta\rangle&=c|0\rangle+d|1\rangle,
\end{align}

\noindent with $|a|^2+|b|^2=|c|^2+|d|^2=1$. Then,
\begin{equation}
|\alpha\beta\rangle
=ac|00\rangle+ad|01\rangle+bc|10\rangle+bd|11\rangle.
\end{equation}

\noindent Applying $U(\theta)$ gives:
\begin{align}
U(\theta)|\alpha\beta\rangle
&=
e^{i\theta}ac|00\rangle+\big(\cos\theta\,ad+i\sin\theta\,bc\big)|01\rangle \\
&+\big(i\sin\theta\,ad+\cos\theta\,bc\big)|10\rangle
+e^{i\theta}bd|11\rangle.
\end{align}

\noindent Denote the output amplitudes in the computational basis by:
\begin{equation*}
\begin{split}
A:=e^{i\theta}ac,\quad
B:=\cos\theta\,ad+i\sin\theta\,bc,\\
C:=i\sin\theta\,ad+\cos\theta\,bc,\quad
D:=e^{i\theta}bd.
\end{split}
\end{equation*}

\noindent For a two-qubit pure state $A|00\rangle+B|01\rangle+C|10\rangle+D|11\rangle$, one has:
\begin{equation}
E_{\mathrm{lin}}=2|AD-BC|^2.
\end{equation}

\noindent Hence,
\begin{equation}
\label{eq:Elin-ADBC}
E_{\mathrm{lin}}\big(U(\theta)|\alpha\beta\rangle\big)=2\big|AD-BC\big|^2.
\end{equation}

\noindent A direct simplification using the above expressions yields:
\begin{equation}
AD-BC
=
-i\sin\theta\cos\theta\,(ad-bc)^2.
\end{equation}

\noindent Now choose the product input parameters such that $ac=bd=0$, i.e.\ take for instance:
\begin{equation}
|\alpha\rangle=|0\rangle,\qquad |\beta\rangle=|1\rangle,
\end{equation}

\noindent so that $a=1,b=0,c=0,d=1$. Then $A=D=0$ and $B=\cos\theta$, $C=i\sin\theta$, so:
\begin{equation}
AD-BC = -\,\cos\theta\,(i\sin\theta)= -\,i\sin\theta\cos\theta,
\end{equation}

\noindent and therefore:
\begin{equation*}
\begin{split}
E_{\mathrm{lin}}\big(U(\theta)|01\rangle\big)
&=
2|\sin\theta\cos\theta|^2
=
2\sin^2\theta\cos^2\theta \\
&=
\frac{1}{2}\sin^2(2\theta).
\end{split}
\end{equation*}

\noindent This gives the lower bound:
\begin{equation}
\mathrm{ep}_{\mathrm{lin}}\big(U(\theta)\big)\ge \frac{1}{2}\sin^2(2\theta).
\end{equation}

\noindent For the matching upper bound, the Cauchy--Schwarz inequality and normalisation of the input states give:
\begin{equation*}
|ad-bc|^2\le (|a|^2+|b|^2)(|d|^2+|c|^2)=1.
\end{equation*}

\noindent Together with the expression above for $AD-BC$, this implies:
\begin{equation*}
\begin{split}
E_{\mathrm{lin}}\big(U(\theta)|\alpha\beta\rangle\big)
&=2\sin^2\theta\cos^2\theta\,|ad-bc|^4\\
&\le\frac{1}{2}\sin^2(2\theta).
\end{split}
\end{equation*}

\noindent
Thus, $\mathrm{ep}_{\mathrm{lin}}(U(\theta))=\frac{1}{2}\sin^2(2\theta)$.
\end{proof}

\section{XQP Circuits as Six-Vertex and Potts Models} \label{sec:potts}

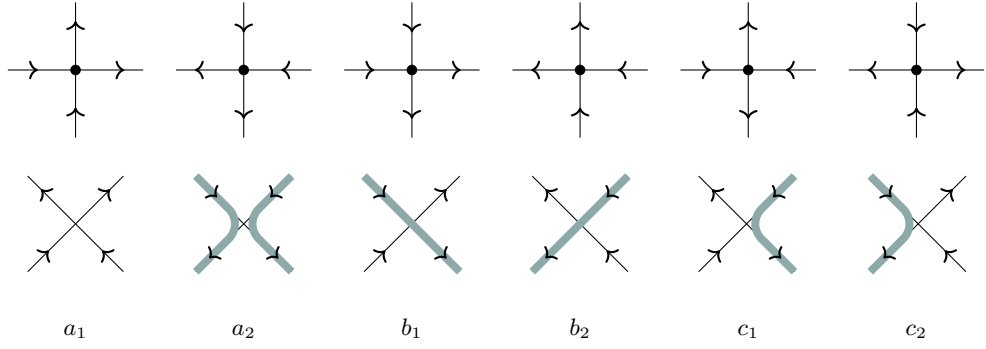
\begin{figure*}
\begin{equation}
\begin{array}{ccccccccccc}
\vtx{out}{in}{in}{out} &  & \vtx{in}{out}{out}{in} & & \vtx{in}{out}{in}{out} && \vtx{out}{in}{out}{in}&& \vtx{out}{out}{in}{in} && \vtx{in}{in}{out}{out} \\
\\ \begin{tikzpicture}[scale=1.0,baseline={(0,0)}, rotate=45]
  \draw (-0.9,0)--(0.9,0);
  \draw (0,-0.9)--(0,0.9);
  \path[tips, ->, line width=0.8pt] (0,0)--(0,0.65);
  \path[tips, ->, line width=0.8pt]  (0,-0.9)--(0,-0.5);
  \path[tips, ->, line width=0.8pt]  (-0.9,0)--(-0.5,0);
  \path[tips, ->, line width=0.8pt]  (0,0)--(0.65,0);
\end{tikzpicture}&& 
\begin{tikzpicture}[scale=1.0,baseline={(0,0)}, rotate=45]
  \draw (-0.9,0)--(0.9,0);
  \draw (0,-0.9)--(0,0.9);
  \draw[camblue, line width=3pt, rounded corners=8pt] (-0.9,0) -- (0,0) -- (0,0.9);
  \draw[camblue, line width=3pt, rounded corners=8pt] (0.9,0) -- (0,0) -- (0,-0.9);
  \path[tips, ->, line width=0.8pt]  (0,0.9)--(0,0.5);
  \path[tips, ->, line width=0.8pt]  (0,0)--(0,-0.65);
  \path[tips, ->, line width=0.8pt]  (0,0)--(-0.65,0);
  \path[tips, ->, line width=0.8pt]  (0.9,0)--(0.5,0);
\end{tikzpicture} && 
\begin{tikzpicture}[scale=1.0,baseline={(0,0)}, rotate=45]
  \draw (-0.9,0)--(0.9,0);
  \draw (0,-0.9)--(0,0.9);
  \draw[camblue, line width=3pt, rounded corners=8pt] (0,-0.9) -- (0,0) -- (0,0.9);
  \path[tips, ->, line width=0.8pt]  (0,0.9)--(0,0.5);
  \path[tips, ->, line width=0.8pt]  (0,0)--(0,-0.65);
  \path[tips, ->, line width=0.8pt]  (-0.9,0)--(-0.5,0);
  \path[tips, ->, line width=0.8pt]  (0,0)--(0.65,0);
\end{tikzpicture}&&  
\begin{tikzpicture}[scale=1.0,baseline={(0,0)}, rotate=45]
  \draw (-0.9,0)--(0.9,0);
  \draw (0,-0.9)--(0,0.9);
  \draw[camblue, line width=3pt, rounded corners=8pt] (-0.9,0) -- (0,0) -- (0.9,0);
  \path[tips, ->, line width=0.8pt] (0,0)--(0,0.65);
  \path[tips, ->, line width=0.8pt]  (0,-0.9)--(0,-0.5);
  \path[tips, ->, line width=0.8pt]  (0,0)--(-0.65,0);
  \path[tips, ->, line width=0.8pt]  (0.9,0)--(0.5,0);
\end{tikzpicture}&& 
\begin{tikzpicture}[scale=1.0,baseline={(0,0)}, rotate=45]
  \draw (-0.9,0)--(0.9,0);
  \draw (0,-0.9)--(0,0.9);
  \draw[camblue, line width=3pt, rounded corners=8pt] (0.9,0) -- (0,0) -- (0,-0.9);
  \path[tips, ->, line width=0.8pt] (0,0)--(0,0.65);
  \path[tips, ->, line width=0.8pt]  (0,0)--(0,-0.65);
  \path[tips, ->, line width=0.8pt]  (-0.9,0)--(-0.5,0);
  \path[tips, ->, line width=0.8pt]  (0.9,0)--(0.5,0);
\end{tikzpicture} && 
\begin{tikzpicture}[scale=1.0,baseline={(0,0)}, rotate=45]
  \draw (-0.9,0)--(0.9,0);
  \draw (0,-0.9)--(0,0.9);
  \draw[camblue, line width=3pt, rounded corners=8pt] (-0.9,0) -- (0,0) -- (0,0.9);
  \path[tips, ->, line width=0.8pt]  (0,0.9)--(0,0.5);
  \path[tips, ->, line width=0.8pt]  (0,-0.9)--(0,-0.5);
  \path[tips, ->, line width=0.8pt]  (0,0)--(-0.65,0);
  \path[tips, ->, line width=0.8pt]  (0,0)--(0.65,0);
\end{tikzpicture} \\
\\
a_1 &  & a_2 && b_1 && b_2 && c_1 && c_2 \\
\end{array}
\end{equation}
\caption{\textit{Six-vertex Model Weights.} Each vertex contains the same number of incoming and outgoing arrows (the ice rule). We can treat the six possible vertex arrangements in the top row as the trajectories of $|0 \rangle$ (thin lines corresponding to upwards arrows) and $|1 \rangle$ (thick lines corresponding to downwards arrows) by rotating $45^{\circ}$ in the anti-clockwise direction. For the $U(\theta) = \cos \theta \cdot \mathbf{1} + i \sin \theta \cdot E_{ij}$ exchange gates, the corresponding weights are given by $a_1 = a_2 = e^{i \theta}$, $b_1 = b_2 = i \sin \theta$, and $c_1 = c_2 = \cos \theta$.}
\label{fig:6V}
\end{figure*}

An alternative way of understanding $\mathsf{XQP}$ circuits is through the lens of the six-vertex (6V) model in statistical physics. This connection has already been established in \citep{Van_den_Nest_2009,De_las_Cuevas_2011} in the context of $\mathsf{DFS}_k$ circuits, where the authors showed that approximating partition functions of six-vertex models with complex vertex weights -- corresponding to the $U(\theta)$ and $S$ gates -- is $\mathsf{BQP}$-complete. In light of Theorem \ref{thm:sqrtswap_universal}, this section extends their result to the case of six-vertex models with vertices corresponding to $U(\theta)$ gates alone. Furthermore, viewing $\mathsf{XQP}$ circuit amplitudes as partition functions of a complex six-vertex model with external boundary conditions allows us to identify instances of $\mathsf{XQP}$ circuits that are classically efficiently simulable. 

An important feature of six-vertex models is their ability to encode the partition function of the Potts model. Leveraging this property, we will show that given access to the isotropic Heisenberg exchange interaction, it is possible to encode the partition functions of a critical $q=4$ Potts model at imaginary temperature into measurement amplitudes. However, the encoding process requires a modification of the underlying six-vertex model to include additional `external' vertex weights, which we show to correspond to $S$ gates on an quantum circuit level -- meaning that the \textit{exact} encoding of the Potts models falls outside the scope of $\mathsf{XQP}$. Therefore, the boundary between the $\mathsf{XQP}$ model with, and without $S$ gates or phase gadgets coincides with the boundary between six-vertex models which can, and cannot simulate the Potts model. It is also worth noting that the situation is unlike the case of $\mathsf{IQP}$ circuits, where the partition function of the Ising model (i.e. the $q=2$ Potts model) at imaginary temperature can be encoded into probability amplitudes without any modifications \cite{Fujii_2017}.

\begin{definition}[Six-Vertex Model]
    A six-vertex model on a finite graph $\mathcal{L}$ consists of:
    \begin{itemize}
        \item A set $V_{\mathrm{int}}(\mathcal L)$ of internal vertices, each of degree 4, and boundary edges incident to exactly one internal vertex,
        \item For each internal vertex $j\in V_{\mathrm{int}}(\mathcal L)$, six local weights $a_1^j,a_2^j,b_1^j,b_2^j,c_1^j,c_2^j$, corresponding to the six allowed local arrow configurations satisfying the ice rule (Figure \ref{fig:6V}),
        \item An arrow configuration $\gamma$ on the edges of $\mathcal L$ such that at each internal vertex there are exactly two incoming and two outgoing arrows.
    \end{itemize}

\noindent The weight of a configuration $\gamma$ is $w(\gamma)=\prod_{j\in V_{\mathrm{int}}(\mathcal L)} w_j(\gamma)$, where $w_j(\gamma)\in\{a_1^j,a_2^j,b_1^j,b_2^j,c_1^j,c_2^j\}$ is the weight associated to the local arrow configuration at vertex $j$. The partition function is given by:
    \begin{equation}
        Z_{6V}=\sum_{\gamma} w(\gamma) = \sum_{\gamma} \prod_{j\in V_{\mathrm{int}}(\mathcal L)} w_j(\gamma),
    \end{equation}
\noindent where the sum runs over all arrow configurations satisfying the ice rule, and (if boundary conditions are imposed) with fixed arrow orientations on the boundary edges.
\end{definition}

\begin{figure*}
    \subfloat{
\begin{tikzpicture}
\node (circ) {
    \scalebox{0.8}{%
    \rotatebox{90}{%
      \begin{quantikz}[row sep=0.35cm, column sep=0.55cm]
        \lstick{\rotatebox{-90}{$|0\rangle$}} & & & \gate[2]{\rotatebox{-90}{$\lambda_3,\nu_1$}} & & & \rstick{\rotatebox{-90}{$\langle 0|$}}\\
        \lstick{\rotatebox{-90}{$|0\rangle$}} & & \gate[2]{\rotatebox{-90}{$\lambda_2,\nu_1$}} & & \gate[2]{\rotatebox{-90}{$\lambda_3,\nu_2$}} & & \rstick{\rotatebox{-90}{$\langle 0|$}} \\
        \lstick{\rotatebox{-90}{$|0\rangle$}} & \gate[2]{\rotatebox{-90}{$\lambda_1,\nu_1$}} & & \gate[2]{\rotatebox{-90}{$\lambda_2,\nu_2$}} & & \gate[2]{\rotatebox{-90}{$\lambda_3,\nu_3$}} & \rstick{\rotatebox{-90}{$\langle 0|$}}\\
        \lstick{\rotatebox{-90}{$|1\rangle$}} & & \gate[2]{\rotatebox{-90}{$\lambda_1,\nu_2$}} & & \gate[2]{\rotatebox{-90}{$\lambda_2,\nu_3$}} & & \rstick{\rotatebox{-90}{$\langle 1|$}} \\
        \lstick{\rotatebox{-90}{$|1\rangle$}} & & & \gate[2]{\rotatebox{-90}{$\lambda_1,\nu_3$}} & & & \rstick{\rotatebox{-90}{$\langle 1|$}}\\
        \lstick{\rotatebox{-90}{$|1\rangle$}} & & & & & & \rstick{\rotatebox{-90}{$\langle 1|$}}
      \end{quantikz}
    }%
    }
  };

  \draw[->, thick] ($(circ.west)+(-0.6,-1)$) -- ($(circ.west)+(-0.6,1.0)$)
    node[midway, left] {$t$};
\end{tikzpicture}
} \quad \quad \quad 
    \subfloat{
\begin{tikzpicture}[scale = 1.2, baseline={(0,-3)}, rotate=45]
    \draw[dashed] (-1,1)--(-1,-1);
    \draw[dashed] (0,1)--(0,-1);
    \draw[dashed] (1,1)--(1,-1);
    \draw[dashed] (-1,1)--(1,1);
    \draw[dashed] (-1,0)--(1,0);
    \draw[dashed] (-1,-1)--(1,-1);
    \draw (-2, -1) -- (-1, -1) -- (-1, -2);
    \draw (2, -1) -- (1, -1) -- (1, -2);
    \draw (2, 1) -- (1, 1) -- (1, 2);
    \draw (-2, 1) -- (-1, 1) -- (-1, 2);
    \draw (-2, 0) -- (-1, 0);
    \draw (2, 0) -- (1, 0);
    \draw(0, 2) -- (0, 1);
    \draw(0, -2) -- (0, -1);
    \draw[camblue, line width=3pt, rounded corners=8pt] (2, -1) -- (1, -1) -- (1, -2);
    \draw[camblue, line width=3pt, rounded corners=8pt] (2, 1) -- (1, 1);
    \draw[camblue, line width=3pt, rounded corners=8pt] (2, 0) -- (1, 0);
    \draw[camblue, line width=3pt, rounded corners=8pt] (-1,-2) -- (-1,-1);
    \draw[camblue, line width=3pt, rounded corners=8pt] (0,-2) -- (0,-1);

    \node[anchor=west] at (-1,-2) {$\nu_1$};
    \node[anchor=west] at (0,-2) {$\nu_2$};
    \node[anchor=west] at (1,-2) {$\nu_3$};

    \node[anchor=west] at (-2,-1) {$\lambda_1$};
    \node[anchor=west] at (-2,0) {$\lambda_2$};
    \node[anchor=west] at (-2,1) {$\lambda_3$};

\end{tikzpicture}
}
\caption{\textit{Domain Wall Boundary Conditions (DWBC).} \textit{Left:} An $\mathsf{XQP}$ circuit of DWBC type, where each gate is an exchange interaction $U(\arctan (\lambda_i - \nu_j))$.  \textit{Right:} The corresponding six-vertex model with DWBC, where a vertex connecting lines with spectral parameters $\lambda_i$ and $\nu_j$ has weights $a(\lambda_i, \nu_j) = 1 + i(\lambda_i - \nu_j)$, $b(\lambda_i, \nu_j) = i (\lambda_i - \nu_j)$, and $c(\lambda_i, \nu_j) = 1$.}
\label{fig:DWBC}
\end{figure*}
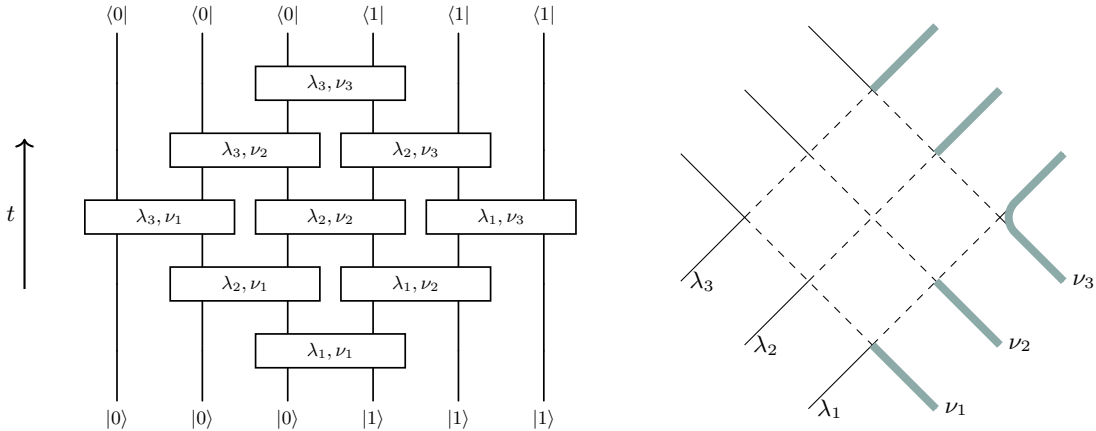

\begin{remark}
Every $\mathsf{XQP}$ circuit amplitude $\langle \mathbf{y} |U_{XQP} | \mathbf{x} \rangle$ can be represented as the partition function of a six-vertex model with complex vertex weights and boundary conditions imposed by the input and output bitstrings. 
\end{remark}

To explain the above Remark, we note that every $\mathsf{XQP}$ circuit on $n$ qubits and $N$ gates can be represented as a degree 4 graph $\mathcal{L}$ with $N$ internal vertices and $2n$ boundary edges. Each possible `path' taken from the input to output strings forms a configuration, where the arrows can be drawn to point in the direction of time for $|0\rangle$, and in the opposite direction for $|1\rangle$. The Hamming weight preservation of $U(\theta)$ gates then ensures that the ice rule is satisfied at each vertex. This construction is implicit in Figures \ref{fig:6V}, \ref{fig:DWBC}. For example, take the following $\mathsf{XQP}$ circuit:
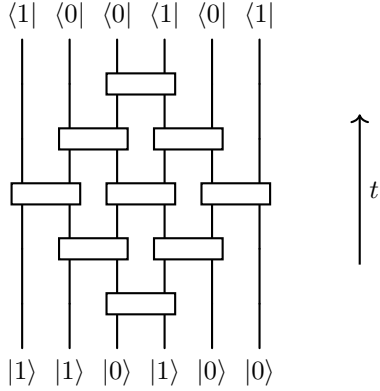
\begin{center}
\begin{tikzpicture}
\scalebox{1}{
\node (circ) {
    \rotatebox{90}{%
      \begin{quantikz}[row sep=0.35cm, column sep=0.45cm]
        \lstick{\rotatebox{-90}{$|1\rangle$}} & & & \gate[2]{} & & & \rstick{\rotatebox{-90}{$\langle 1|$}} \\
        \lstick{\rotatebox{-90}{$|1\rangle$}} & & \gate[2]{} & & \gate[2]{} & & \rstick{\rotatebox{-90}{$\langle 0|$}} \\
        \lstick{\rotatebox{-90}{$|0\rangle$}} & \gate[2]{} & & \gate[2]{} & & \gate[2]{} & \rstick{\rotatebox{-90}{$\langle 0|$}} \\
        \lstick{\rotatebox{-90}{$|1\rangle$}} & & \gate[2]{} & & \gate[2]{} & & \rstick{\rotatebox{-90}{$\langle 1|$}} \\
        \lstick{\rotatebox{-90}{$|0\rangle$}} & & & \gate[2]{} & & & \rstick{\rotatebox{-90}{$\langle 0|$}} \\
        \lstick{\rotatebox{-90}{$|0\rangle$}} & & & & & & \rstick{\rotatebox{-90}{$\langle 1|$}} \\
      \end{quantikz}
    }%
  };

  \draw[->, thick] ($(circ.east)+(0.6,-1)$) -- ($(circ.east)+(0.6,1.0)$)
    node[midway, right] {$t$};
}
\end{tikzpicture}
\end{center}

\noindent Its measurement amplitudes are given by:

\begin{align}
\langle 100101 | U_{XQP} | 110100 \rangle &= \sum_{\mathrm{arrows}} \  \begin{tikzpicture}[scale = 0.4, baseline={(0,0)}, rotate=45]
    \draw[dashed] (-1,1)--(-1,-1);
    \draw[dashed] (0,1)--(0,-1);
    \draw[dashed] (1,1)--(1,-1);
    \draw[dashed] (-1,1)--(1,1);
    \draw[dashed] (-1,0)--(1,0);
    \draw[dashed] (-1,-1)--(1,-1);
    \draw (-2, -1) -- (-1, -1) -- (-1, -2);
    \draw (2, -1) -- (1, -1) -- (1, -2);
    \draw (2, 1) -- (1, 1) -- (1, 2);
    \draw (-2, 1) -- (-1, 1) -- (-1, 2);
    \draw (-2, 0) -- (-1, 0);
    \draw (2, 0) -- (1, 0);
    \draw(0, 2) -- (0, 1);
    \draw(0, -2) -- (0, -1);
    \draw[camblue, line width=3pt, rounded corners=8pt] (-2, 1) -- (-1, 1) -- (-1, 2);
    \draw[camblue, line width=3pt, rounded corners=8pt] (-2, 0) -- (-1, 0);
    \draw[camblue, line width=3pt, rounded corners=8pt] (2, -1) -- (1, -1);
    \draw[camblue, line width=3pt, rounded corners=8pt] (2, 1) -- (1, 1);
    \draw[camblue, line width=3pt, rounded corners=8pt] (-1,-2) -- (-1,-1);
    \end{tikzpicture}  \\
& = Z_{6V},
\end{align}

\noindent where the dashed lines represent unassigned arrows, and the sum is taken over all arrow configurations consistent with the boundary conditions.

We can also use $\mathsf{XQP}$ circuits to compute the partition function of an associated complex six-vertex model with no boundary conditions, however these turn out to be trivial:

\begin{remark}
    With no boundary conditions, the partition function $Z_{6V}$ of the six-vertex model with weights $a(\theta) = e^{i \theta}$, $b(\theta) = i \sin \theta$, and $c(\theta) = \cos \theta$ at each vertex is computable as $Z_{6V} = 2^n e^{i \sum_{j}^{N} \theta_j}$, where $N$ is the number of vertices, and $n$ is the sum of the length and width of the lattice.
\end{remark}
\begin{proof}
    This follows from the fact that $Z_{6V} = 2^{n} \langle \mathbf{0} | H^{\otimes n}  U_{XQP} H^{\otimes n} | \mathbf{0} \rangle$, where $H$ is the Hadamard gate. As $U_{XQP}$ circuits commute with any $n$-fold tensor product of single qubit gates, we have $Z_{6V} = 2^{n} \langle \mathbf{0} | U_{XQP} | \mathbf{0} \rangle = 2^n e^{i \sum_{j}^{N} \theta_j}$. 
\end{proof}

Making use of the famous Izergin--Korepin determinant formula \citep{Izergin_1992}, we can consider $\mathsf{XQP}$ circuits with domain wall boundary conditions (DWBC). These are shown in Figure \ref{fig:DWBC}, and lead to non-trivial class of $\mathsf{XQP}$ circuits whose amplitudes encode efficiently computable partition functions. We first re-write each exchange interaction in an alternative way (as introduced in \citep{Aaronson_2016}), using two spectral parameters $\lambda$ and $\nu$:
\begin{align}
    U_{ij}(\lambda, \nu) &= \frac{1}{\sqrt{1 + (\lambda - \nu)^2}} \mathbf{1} + \frac{i (\lambda - \nu)}{\sqrt{1 + (\lambda - \nu)^2}} E_{ij} \\
    & = U_{ij}(\arctan(\lambda - \nu)).
\end{align}

\noindent We remark that with this parametrisation, the exchange interaction $U_{ij}(\lambda, \nu)$ satisfies the parameter-dependent Yang--Baxter equation \citep{Aaronson_2016}:

\begin{align}
    U_{12}(\lambda_1, \lambda_2) & U_{23}(\lambda_1, \lambda_3) U_{12}(\lambda_2, \lambda_3) \\
    & = U_{23}(\lambda_2, \lambda_3) U_{12}(\lambda_1, \lambda_3) U_{23}(\lambda_1, \lambda_2).
\end{align}

\begin{remark}
Consider an $\mathsf{XQP}$ circuit $U_{\mathrm{DWBC}}$ of DWBC type, i.e. one of the form given in Figure \ref{fig:DWBC}. The measurement amplitude $\langle 0^{n/2} 1^{n/2} | U_{\mathrm{DWBC}} | 0^{n/2} 1^{n/2} \rangle$ of such a circuit can be computed exactly in $\text{poly}(n)$ time using the determinant formula of Izergin, Coker and Korepin \citep{Izergin_1992}.
\end{remark}

\begin{proof}

For an $\mathsf{XQP}$ circuit of DWBC type, we can factor out $1 / \sqrt{1 + (\lambda_i - \nu_j)^2}$ from each $U(\lambda_i, \nu_j)$ gate to obtain:

\begin{equation}
    \langle 0^{\frac{n}{2}} 1^{\frac{n}{2}} | U_\mathrm{DWBC} | 0^{\frac{n}{2}} 1^{\frac{n}{2}} \rangle = \prod_{i, j} \frac{1}{\sqrt{1 + (\lambda_i - \nu_j)^2}} Z_{\mathrm{DWBC}},
\end{equation}

\noindent where the corresponding six-vertex model weights are given by $a(\lambda_i, \nu_j) = 1 + i(\lambda_i - \nu_j)$, $b(\lambda_i, \nu_j) = i (\lambda_i - \nu_j)$, and $c(\lambda_i, \nu_j) = 1$. The determinant formula is readily applicable. We only need to equate our weights to the ones given in \citep{Izergin_1992}, in order to find the correct form of the spectral parameters $\tilde{\lambda}_i$ and $\tilde{\nu}_j$. Comparing Figure \ref{fig:6V} with Figure 6 of \citep{Izergin_1992}, we have $\tilde{a}(\tilde{\lambda}_i, \tilde{\nu}_j) = b(\lambda_i, \nu_j)$, $\tilde{b}(\tilde{\lambda}_i, \tilde{\nu}_j) = a(\lambda_i, \nu_j)$, and $\tilde{c}(\tilde{\lambda}_i, \tilde{\nu}_j) = c(\lambda_i, \nu_j)$. This gives us (c.f. Equation (6.1) of \citep{Izergin_1992}):
\begin{align}
    \tilde{a}(\tilde{\lambda}_i, \tilde{\nu}_j) & = \tilde{\lambda}_i - \tilde{\nu}_j + i \kappa/2 = i(\lambda_i - \nu_j), \\
    \tilde{b}(\tilde{\lambda}_i, \tilde{\nu}_j) & = \tilde{\lambda}_i - \tilde{\nu}_j - i \kappa/2 = 1 + i(\lambda_i - \nu_j), \\
    \tilde{c}(\tilde{\lambda}_i, \tilde{\nu}_j) & = - i \kappa = 1,
\end{align}

\noindent where $\kappa$ is an arbitrary parameter, which we set to $\kappa = i$. Solving for $\tilde{\lambda}_i$ and $\tilde{\nu}_j$ in terms of $\lambda_i$ and $\nu_j$, we obtain: 
\begin{equation}
    \tilde{\lambda}_i  = i \lambda_i + 1/4, \quad 
    \tilde{\nu}_j  = i \nu_j - 1/4.
\end{equation}

\noindent As the spectral parameters $\tilde{\lambda}_i$ and $\tilde{\nu}_j$ are arbitrary this means we can substitute directly into Equation (6.3) of \citep{Izergin_1992} to obtain a closed-form expression for $Z_{\mathrm{DWBC}}$, and thus the corresponding $\mathsf{XQP}$ amplitude. Similarly, if the corresponding six-vertex model is homogeneous (i.e. $\lambda_i - \nu_j = x$ for all $i, j$), we can set $\tilde{x} = \tilde{\lambda}_i - \tilde{\nu}_j = i x + 1/2$, and apply the homogeneous limit of the determinant formula.

\end{proof}

\begin{remark} \label{rem:6vbqp}
Approximating the partition function $Z_{6V}$ of an arbitrary six-vertex model with weights $a(\theta) = e^{i \theta}$, $b(\theta) = i \sin \theta$, and $c(\theta) = \cos \theta$ is $\mathsf{BQP}$-complete.
\end{remark}
\begin{proof}
The proof proceeds in complete analogy to Result 1 of \cite{De_las_Cuevas_2011}, where the original statement was proven for 6V models with additional weights $a_1 = a_2 = 1$, $b_1 = -1/\sqrt{2}$, $b_2 = c_1 = c_2 = 1/\sqrt{2}$ available. By Corollary \ref{cor:xqpequalsbqp}, the latter (corresponding to singlet state preparation and measurement) are not required to establish $\mathsf{BQP}$-completeness, strengthening the original result. 
\end{proof}

We now study the square Potts model and its connection to $\mathsf{XQP}$ circuits. An interesting property of the six-vertex model is that, given certain modifications, it is able to encode the partition function of the Potts model as $Z_{6V'} \propto Z_{\mathrm{Potts}}$ \citep{Temperley_1971, Baxter_1976,Baxter_1982}. Given the ability to encode the partition function of a complex six-vertex model into $\mathsf{XQP}$ measurement amplitudes, this raises the question of whether $\mathsf{XQP}$ circuits can also exactly encode the partition function of the complex Potts model, which have been shown to be $\# \mathsf{P}$-hard to approximate \citep{Galanis_2021}. As other restricted models of quantum computation are able to encode quantities with $\# \mathsf{P}$ approximation hardness (for example $\mathsf{IQP}$ and the complex Ising model partition function \citep{Fujii_2017, Bremner_2016, Goldberg_2017}, or $\mathsf{BosonSampling}$ and the matrix permanent or hafnian \citep{Aaronson_2010, Aaronson_2011, Hamilton_2017}), this is a natural question to ask and a step towards proving average-case hardness of simulating $\mathsf{XQP}$ circuits to additive precision, via similar techniques to \citep{Bremner_2016, Aaronson_2010, Morimae_2017}. Surprisingly, it turns out that $\mathsf{XQP}$ circuits without $S$ gates or phase gadgets are unsuccessful in this, and in our investigation we have not been able to identify any other quantity with $\# \mathsf{P}$ approximation hardness that can be exactly encoded into $\mathsf{XQP}$ amplitudes, leaving its existence as an open question. Despite this, there is strong evidence that additive weak simulation is unlikely to be efficient, namely Theorems \ref{thm:tvxqpbqp} and \ref{thm:sqrtswap_universal}. 

We will follow the exposition of Chapter 12 of \citep{Baxter_1982}, and state the relevant results. A $q$-state Potts model on a lattice $\mathcal{L}$ of $N$ vertices (sites) consists of a spin variable $\sigma_i$ at each site, which can take values in $\sigma_i \in \{1, 2, \dots, q\}$. Any two adjacent sites (i.e. vertices connected by an edge in $\mathcal{L})$ interact with energy $-J \delta(\sigma_i, \sigma_j)$. A configuration of spins on $\mathcal{L}$, denoted $\sigma(\mathcal{L})$, has energy $E(\sigma) = - J \sum_{(i,j) \in \sigma(\mathcal{L})} \delta(\sigma_i, \sigma_j)$, where the sum is taken over all edges in $\mathcal{L}$. The partition function of the Potts model is then given by:

\begin{equation}
    Z_{\mathrm{Potts}} = \sum_{\sigma(\mathcal{L})} e^{- E(\sigma) / k_B T} = \sum_{\sigma(\mathcal{L})} e^{K \sum_{(i,j)} \delta(\sigma_i, \sigma_j)},
\end{equation}

\noindent where $K = J / k_B T$. This partition function can be then re-expressed in the Fortuin--Kasteleyn (FK) representation \citep{Kasteleyn_1969,Fortuin_1972}, by setting $v = e^{K} - 1$, so that:

\begin{equation}
    Z_{\mathrm{Potts}} = \sum_{\mathcal{G} \subseteq \mathcal{L}} q^{C(\mathcal{G})} v^{|\mathcal{G}|},
\end{equation}

\noindent where $\mathcal{G} \subseteq \mathcal{L}$ are subgraphs of $\mathcal{L}$ obtained by choosing a subset of edges, $|\mathcal{G}|$ is the number of edges in $\mathcal{G}$, and $C(\mathcal{G})$ is the number of connected components in $\mathcal{G}$ (where isolated vertices also count as connected components). 

Here, we are interested in the Potts model on a square lattice, as pictured by the white dots and dashed lines of Figure \ref{fig:potts}. We group the edge interactions in the FK representation into two classes: horizontal, which assign a weight $v_h = e^{K_h} - 1$, and vertical which assign a weight $v_v = e^{K_v} - 1$. We then have $|\mathcal{G}_h|$ and $|\mathcal{G}_v|$ as the number of horizontal and vertical edges in $\mathcal{G} \subseteq \mathcal{L}$, so that:

\begin{equation} \label{eq:potts}
    Z_{\mathrm{Potts}} = \sum_{\mathcal{G} \subseteq \mathcal{L}} q^{C(\mathcal{G})} v_h^{|\mathcal{G}_h|} v_v^{|\mathcal{G}_v|}. 
\end{equation}

\begin{figure}
    \subfloat{
\begin{tikzpicture}[scale = 0.65, baseline={(0,0)}]

\fill[fill=camblue, draw=black, opacity=0.7] (0,-1) -- (1,0) -- (0,1) -- (-1, 0) -- (0,-1);

\fill[fill=camblue, draw=black, opacity=0.7] (1, 0) -- (2,1) -- (3,0) -- (2,-1) -- (1,0);

\fill[fill=camblue, draw=black, opacity=0.7] (-1, 0) -- (-2,1) -- (-3,0) -- (-2,-1) -- (-1,0);

\fill[fill=camblue, draw=black, opacity=0.7] (0,1) -- (1,2) -- (0,3) -- (-1, 2) -- (0,1);

\fill[fill=camblue, draw=black, opacity=0.7] (1, 2) -- (2,3) -- (3,2) -- (2,1) -- (1,2);

\fill[fill=camblue, draw=black, opacity=0.7] (-1, 2) -- (-2,3) -- (-3,2) -- (-2,1) -- (-1,2);

\fill[fill=camblue, draw=black, opacity=0.7] (0,-3) -- (1,-2) -- (0,-1) -- (-1, -2) -- (0,-3);

\fill[fill=camblue, draw=black, opacity=0.7] (1, -2) -- (2,-1) -- (3,-2) -- (2,-3) -- (1,-2);

\fill[fill=camblue, draw=black, opacity=0.7] (-1, -2) -- (-2,-1) -- (-3,-2) -- (-2,-3) -- (-1,-2);

\draw[dashed] (0,-2)--(0,2);
\draw[dashed] (-2,-2)--(-2,2);
\draw[dashed] (2,-2)--(2,2);
\draw[dashed] (-2,2)--(2,2);
\draw[dashed] (-2,0)--(2,0);
\draw[dashed] (-2,-2)--(2,-2);

\fill (-3,2) circle (3pt);
\fill (-3,0) circle (3pt);
\fill (-3,-2) circle (3pt);

\fill(-2, 3) circle (3pt);
\fill(-2, 1) circle (3pt);
\fill(-2, -1) circle (3pt);
\fill(-2, -3) circle (3pt);

\fill (-1,2) circle (3pt);
\fill (-1,0) circle (3pt);
\fill (-1,-2) circle (3pt);

\fill(0, 3) circle (3pt);
\fill(0, 1) circle (3pt);
\fill(0, -1) circle (3pt);
\fill(0, -3) circle (3pt);

\fill (1,2) circle (3pt);
\fill (1,0) circle (3pt);
\fill (1,-2) circle (3pt);

\fill(2, 3) circle (3pt);
\fill(2, 1) circle (3pt);
\fill(2, -1) circle (3pt);
\fill(2, -3) circle (3pt);

\fill (3,2) circle (3pt);
\fill (3,0) circle (3pt);
\fill (3,-2) circle (3pt);

\fill[fill=white, draw=black] (-2, 2) circle (3pt);
\fill[fill=white, draw=black] (-2, 0) circle (3pt);
\fill[fill=white, draw=black] (-2, -2) circle (3pt);

\fill[fill=white, draw=black] (0, 2) circle (3pt);
\fill[fill=white, draw=black] (0, 0) circle (3pt);
\fill[fill=white, draw=black] (0, -2) circle (3pt);

\fill[fill=white, draw=black] (2, 2) circle (3pt);
\fill[fill=white, draw=black] (2, 0) circle (3pt);
\fill[fill=white, draw=black] (2, -2) circle (3pt);

\end{tikzpicture}
} \quad
\subfloat{

\begin{tikzpicture}[scale = 0.65, rotate = 90, baseline={(0, -0.65)}]

\fill[fill=camblue, draw=black, opacity=0.7] (0,-1) -- (1,0) -- (0,1) -- (-1, 0) -- (0,-1);

\fill[fill=camblue, draw=black, opacity=0.7] (-1, 0) -- (-2,1) -- (-3,0) -- (-2,-1) -- (-1,0);

\fill[fill=camblue, draw=black, opacity=0.7] (0,1) -- (1,2) -- (0,3) -- (-1, 2) -- (0,1);

\fill[fill=camblue, draw=black, opacity=0.7] (-1, 2) -- (-2,3) -- (-3,2) -- (-2,1) -- (-1,2);

\draw[dashed] (-2, 0)--(0,0);
\draw[dashed] (-2, 2)--(0,2);
\draw[dashed] (0, 3.5)--(0,-1.5);
\draw[dashed] (-2, 3.5)--(-2,-1.5);

\fill (-3,2) circle (3pt);
\fill (-3,0) circle (3pt);

\fill (-2, 3) circle (3pt);
\fill(-2, 1) circle (3pt);
\fill (-2, -1) circle (3pt);

\fill (-1,2) circle (3pt);
\fill (-1,0) circle (3pt);

\fill (0, 3) circle (3pt);
\fill(0, 1) circle (3pt);
\fill (0, -1) circle (3pt);

\fill (1,2) circle (3pt);
\fill (1,0) circle (3pt);

\fill[fill=white, draw=black] (-2, 2) circle (3pt);
\fill[fill=white, draw=black] (-2, 0) circle (3pt);

\fill[fill=white, draw=black] (0, 2) circle (3pt);
\fill[fill=white, draw=black] (0, 0) circle (3pt);

\draw[color=red, dotted, line width=0.75pt] (-3, -0.6)--(1, -0.6);

\end{tikzpicture}
}
\caption{\textit{The Square Potts Model.} White dots and dashed lines represent the vertices and edges of the square Potts model on $\mathcal{L}$. Black dots and full lines represent the vertices and edges of the equivalent modified six-vertex model on the medial lattice $\mathcal{L}'$. Black dots on dashed lines are the internal vertices, and all other dots are the external vertices. \textit{Left:} The $3 \times 3$ square Potts model. \textit{Right:} The $2 \times 2$ Potts model with cylindrical boundary conditions along the horizontal direction (external vertices on the sides are connected by a red, dotted `seam'). Figure adapted from \citep{Baxter_1982}.}
\label{fig:potts}
\end{figure}
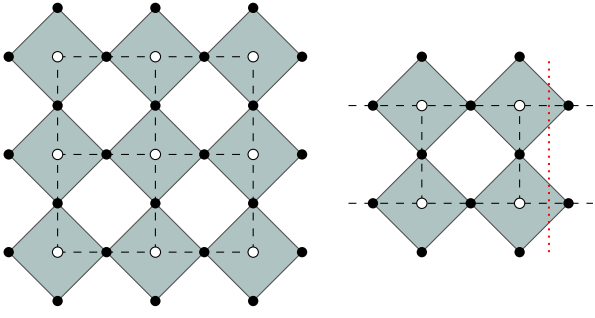

\noindent Chapter 12 of \citep{Baxter_1982} (also see \citep{Jacobsen_2013}) shows that the partition function of a $q$-state square Potts model can be re-expressed as the partition function of a modified six-vertex model on the medial lattice $\mathcal{L}'$ (i.e. another square lattice, whose vertices lie on the edges of $\mathcal{L}$, as illustrated in Figure \ref{fig:potts}). The relevant parameters of this construction are given by:
\begin{align}
    \sqrt{q} &= e^\lambda + e^{-\lambda}, & s &= e^{\lambda/4}, \\
    x_i &= v_i / \sqrt{q}, & \Delta &= - \sqrt{q} / 2,
\end{align}

\noindent where $\Delta$ is the anisotropy parameter of the six-vertex model, given by:
\begin{equation}
    \Delta = \frac{a_1 a_2 + b_1 b_2 - c_1 c_2}{2 \sqrt{a_1 a_2 b_1 b_2}}.
\end{equation}

\noindent In our case, $\Delta = 1$ which sets $\sqrt{q} = -2$, $\lambda = i \pi$, and $s = \sqrt{i}$. This means we are considering the partition function of the $q=4$ Potts model, where we note that the negativity of $\sqrt{q}$ does not affect the validity of the construction (as $q$ always appears in integer powers in $Z_{\mathrm{Potts}}$). The vertices of the modified six-vertex model are given by the black dots in Figure \ref{fig:potts}, and can be classified into two types: internal, which lie on edges of the original lattice $\mathcal{L}$, and external, which do not lie on any edge of $\mathcal{L}$. In an arrow configuration on $\mathcal{L}'$, all vertices must have the same number of incoming and outgoing arrows (i.e. thin lines corresponding to $\nwarrow$, $\nearrow$ and thick lines corresponding to $\color{camblue} \swarrow$, $\color{camblue} \searrow$). Of the internal vertices, there are two further types: those lying on the horizontal edges of $\mathcal{L}$ with weights $a_{i, h}, b_{i, h}, c_{i, h}$, and those on the vertical edges of $\mathcal{L}$ with weights $a_{i, v}, b_{i, v}, c_{i, v}$. By the construction in \citep{Baxter_1982}, the weights of the vertices are given by:

\begin{align} 
    a_{i, h} &= \begin{tikzpicture}[scale=0.5,baseline={(0,0)}, rotate=45]
  \draw (-1,0)--(1,0);
  \draw (0,-1)--(0,1);
\end{tikzpicture} = \begin{tikzpicture}[scale=0.5,baseline={(0,0)}, rotate=45]
  \draw (-1,0)--(1,0);
  \draw (0,-1)--(0,1);
  \draw[camblue, line width=3pt, rounded corners=8pt] (-1,0) -- (0,0) -- (0,1);
  \draw[camblue, line width=3pt, rounded corners=8pt] (1,0) -- (0,0) -- (0,-1);
\end{tikzpicture} = 1, \label{eq:vx_weights1} \\
b_{i, h} & = \begin{tikzpicture}[scale=0.5,baseline={(0,0)}, rotate=45]
  \draw (-1,0)--(1,0);
  \draw (0,-1)--(0,1);
  \draw[camblue, line width=3pt, rounded corners=8pt] (0,-1) -- (0,0) -- (0,1);
\end{tikzpicture} = 
\begin{tikzpicture}[scale=0.5,baseline={(0,0)}, rotate=45]
  \draw (-1,0)--(1,0);
  \draw (0,-1)--(0,1);
  \draw[camblue, line width=3pt, rounded corners=8pt] (-1,0) -- (0,0) -- (1,0);
\end{tikzpicture} = x_h, \\
c_{1,h} &= \begin{tikzpicture}[scale=0.5,baseline={(0,0)}, rotate=45]
  \draw (-1,0)--(1,0);
  \draw (0,-1)--(0,1);
  \draw[camblue, line width=3pt, rounded corners=8pt] (1,0) -- (0,0) -- (0,-1);
\end{tikzpicture} = s^2 + x_h s^{-2} = i(1-x_h), \\
c_{2, h} &= \begin{tikzpicture}[scale=0.5,baseline={(0,0)}, rotate=45]
  \draw (-1,0)--(1,0);
  \draw (0,-1)--(0,1);
  \draw[camblue, line width=3pt, rounded corners=8pt] (-1,0) -- (0,0) -- (0,1);
\end{tikzpicture} = s^{-2} + x_h s^{2} = -i(1-x_h),
\end{align}

\noindent and for the vertical vertices, 
\begin{align}
a_{i, v} &= x_v, & c_{1, v} & = s^{-2} + x_v s^{2} = -i(1-x_v), \\
b_{i, v} & = 1, & c_{2, v} & = s^{2} + x_v s^{-2} = i(1-x_v).
\end{align}

\noindent Finally, the external vertices carry the weights:

\begin{align}
\begin{tikzpicture}[scale=0.5,baseline={(0,0)}, rotate=-45]
  \draw (-1,0)--(0,0);
  \draw (0,0)--(0,1);
  \draw[camblue, line width=3pt, rounded corners=8pt] (0,0) -- (0,1);
\end{tikzpicture} &= s^{-1} = \sqrt{-i}, &
\begin{tikzpicture}[scale=0.5,baseline={(0,0)}, rotate=-45]
  \draw (-1,0)--(0,0);
  \draw (0,0)--(0,1);
  \draw[camblue, line width=3pt, rounded corners=8pt] (-1,0) -- (0,0);
\end{tikzpicture} &= s = \sqrt{i}, \label{eq:vx_weights1a}\\
\begin{tikzpicture}[scale=0.5,baseline={(0,-0.25)}, rotate=-45]
  \draw (1,0)--(0,0);
  \draw (0,0)--(0,-1);
  \draw[camblue, line width=3pt, rounded corners=8pt] (1,0) -- (0,0);
\end{tikzpicture} &= s^{-1} = \sqrt{-i}, & \begin{tikzpicture}[scale=0.5,baseline={(0,-0.25)}, rotate=-45]
  \draw (1,0)--(0,0);
  \draw (0,0)--(0,-1);
  \draw[camblue, line width=3pt, rounded corners=8pt] (0,0) -- (0,-1);
\end{tikzpicture} &= s = \sqrt{i}, \\
\begin{tikzpicture}[scale=0.5,baseline={(0,-0.1)}, rotate=-45]
  \draw (1,0)--(0,0);
  \draw (0,0)--(0,1);
\end{tikzpicture} &= s^{-1} = \sqrt{-i} , & \begin{tikzpicture}[scale=0.5,baseline={(0,-0.1)}, rotate=-45]
  \draw (1,0)--(0,0);
  \draw (0,0)--(0,1);
  \draw[camblue, line width=3pt, rounded corners=8pt] (1,0) -- (0,0) -- (0,1);
\end{tikzpicture} &= s = \sqrt{i}, \\
\begin{tikzpicture}[scale=0.5,baseline={(0,-0.1)}, rotate=-45]
  \draw (-1,0)--(0,0);
  \draw (0,0)--(0,-1);
  \draw[camblue, line width=3pt, rounded corners=8pt] (-1,0) -- (0,0) -- (0,-1);
\end{tikzpicture} &= s^{-1} = \sqrt{-i}, & \begin{tikzpicture}[scale=0.5,baseline={(0,-0.1)}, rotate=-45]
  \draw (-1,0)--(0,0);
  \draw (0,0)--(0,-1);
\end{tikzpicture} &= s = \sqrt{i}.
\end{align}

We can simplify the above by redistributing some of the factors of $s$ and $s^{-1}$ from the external vertices. For each arrow configuration on $\mathcal{L}'$, take each edge in the SW--NE direction (i.e. $\nearrow$, $\color{camblue}\swarrow$) and multiply the weights of its two adjacent vertices by $s^{-1}$ if an arrow points towards the vertex, and by $s$ if an arrow points away from the vertex. Clearly, this will multiply each term in $Z_{6V'}$ by $1$ for any arrow configuration. However, the external vertices on the sides of the lattice are modified to carry weight $1$, and the weights $c_{i, r}$ become:
\begin{align}
    c_{1, h} & \rightarrow s^{-2} c_{1, h} = 1 - x_h, &
    c_{2, h} & \rightarrow s^{2} c_{2, h} = 1 - x_h, \\
    c_{1, v} & \rightarrow s^{-2} c_{1, v} = x_v - 1, &
    c_{2, v} & \rightarrow s^{2} c_{2, v} = x_v - 1.
\end{align}

\noindent Overall, the vertex weights are modified to become:

\begin{align}
    a_h &= \begin{tikzpicture}[scale=0.5,baseline={(0,0)}, rotate=45]
  \draw (-1,0)--(1,0);
  \draw (0,-1)--(0,1);
\end{tikzpicture} = \begin{tikzpicture}[scale=0.5,baseline={(0,0)}, rotate=45]
  \draw (-1,0)--(1,0);
  \draw (0,-1)--(0,1);
  \draw[camblue, line width=3pt, rounded corners=8pt] (-1,0) -- (0,0) -- (0,1);
  \draw[camblue, line width=3pt, rounded corners=8pt] (1,0) -- (0,0) -- (0,-1);
\end{tikzpicture} = 1, \label{eq:vx_weights2}\\
b_h & = \begin{tikzpicture}[scale=0.5,baseline={(0,0)}, rotate=45]
  \draw (-1,0)--(1,0);
  \draw (0,-1)--(0,1);
  \draw[camblue, line width=3pt, rounded corners=8pt] (0,-1) -- (0,0) -- (0,1);
\end{tikzpicture} = 
\begin{tikzpicture}[scale=0.5,baseline={(0,0)}, rotate=45]
  \draw (-1,0)--(1,0);
  \draw (0,-1)--(0,1);
  \draw[camblue, line width=3pt, rounded corners=8pt] (-1,0) -- (0,0) -- (1,0);
\end{tikzpicture} = x_h, \\
c_h &= \begin{tikzpicture}[scale=0.5,baseline={(0,0)}, rotate=45]
  \draw (-1,0)--(1,0);
  \draw (0,-1)--(0,1);
  \draw[camblue, line width=3pt, rounded corners=8pt] (1,0) -- (0,0) -- (0,-1);
\end{tikzpicture} = \begin{tikzpicture}[scale=0.5,baseline={(0,0)}, rotate=45]
  \draw (-1,0)--(1,0);
  \draw (0,-1)--(0,1);
  \draw[camblue, line width=3pt, rounded corners=8pt] (-1,0) -- (0,0) -- (0,1);
\end{tikzpicture} = 1 - x_h, \\
s^{-2} &= \begin{tikzpicture}[scale=0.5,baseline={(0,0)}, rotate=-45]
  \draw (-1,0)--(0,0);
  \draw (0,0)--(0,1);
  \draw[camblue, line width=3pt, rounded corners=8pt] (0,0) -- (0,1);
\end{tikzpicture}  =  \begin{tikzpicture}[scale=0.5,baseline={(0,-0.25)}, rotate=-45]
  \draw (1,0)--(0,0);
  \draw (0,0)--(0,-1);
  \draw[camblue, line width=3pt, rounded corners=8pt] (1,0) -- (0,0);
\end{tikzpicture}  = -i, \\
s^{2} &= \begin{tikzpicture}[scale=0.5,baseline={(0,0)}, rotate=-45]
  \draw (-1,0)--(0,0);
  \draw (0,0)--(0,1);
  \draw[camblue, line width=3pt, rounded corners=8pt] (-1,0) -- (0,0);
\end{tikzpicture}  =  \begin{tikzpicture}[scale=0.5,baseline={(0,-0.25)}, rotate=-45]
  \draw (1,0)--(0,0);
  \draw (0,0)--(0,-1);
  \draw[camblue, line width=3pt, rounded corners=8pt] (0,0) -- (0,-1);
\end{tikzpicture}  = i, \\
\end{align}

\noindent and also $a_v = x_v$, $b_v = 1$, and $c_v = x_v - 1$. 

\begin{figure*}
    \subfloat{
\scalebox{0.8}{
\rotatebox{90}{
\begin{quantikz}
\lstick{\rotatebox{-90}{$| 0 \rangle$}} & \gate[2]{\rotatebox{-90}{$U_S$}} & &
  \gate[2]{\rotatebox{-90}{$\theta_v$}} & &
  \gate[2]{\rotatebox{-90}{$\theta_v$}} & &
  \gate[2]{\rotatebox{-90}{$U_S^\dagger$}} &
  \rstick{\rotatebox{-90}{$\langle 0 |$}}
\\
\lstick{\rotatebox{-90}{$| 1 \rangle$}} & &
  \gate[2]{\rotatebox{-90}{$\theta_h$}} & &
  \gate[2]{\rotatebox{-90}{$\theta_h$}} & &
  \gate[2]{\rotatebox{-90}{$\theta_h$}} & &
  \rstick{\rotatebox{-90}{$\langle 1 |$}}
\\
\lstick{\rotatebox{-90}{$| 0 \rangle$}} & \gate[2]{\rotatebox{-90}{$U_S$}} & &
  \gate[2]{\rotatebox{-90}{$\theta_v$}} & &
  \gate[2]{\rotatebox{-90}{$\theta_v$}} & &
  \gate[2]{\rotatebox{-90}{$U_S^\dagger$}} &
  \rstick{\rotatebox{-90}{$\langle 0 |$}}
\\
\lstick{\rotatebox{-90}{$| 1 \rangle$}} & &
  \gate[2]{\rotatebox{-90}{$\theta_h$}} & &
  \gate[2]{\rotatebox{-90}{$\theta_h$}} & &
  \gate[2]{\rotatebox{-90}{$\theta_h$}} & &
  \rstick{\rotatebox{-90}{$\langle 1 |$}}
\\
\lstick{\rotatebox{-90}{$| 0 \rangle$}} & \gate[2]{\rotatebox{-90}{$U_S$}} & &
  \gate[2]{\rotatebox{-90}{$\theta_v$}} & &
  \gate[2]{\rotatebox{-90}{$\theta_v$}} & &
  \gate[2]{\rotatebox{-90}{$U_S^\dagger$}} &
  \rstick{\rotatebox{-90}{$\langle 0 |$}}
\\
\lstick{\rotatebox{-90}{$| 1 \rangle$}} & & & & & & & &
  \rstick{\rotatebox{-90}{$\langle 1 |$}}
\end{quantikz}
} } } \quad \quad \quad 
    \subfloat{
\begin{tikzpicture}[baseline={(0,-3.7)}]
\node (circ) {
    \scalebox{1}{%
    \rotatebox{90}{%
      \begin{quantikz}
    \lstick{\rotatebox{-90}{$| 0 \rangle$}} &
      \gate[2]{\rotatebox{-90}{$U_S$}} &
      \gate[1,style={cup_gate}]{\rotatebox{-90}{$\theta_h$}} &
      \gate[2]{\rotatebox{-90}{$\theta_v$}} &
      \gate[1,style={cup_gate}]{\rotatebox{-90}{$\theta_h$}} &
      \gate[2]{\rotatebox{-90}{$U_S^\dagger$}} &
      \rstick{\rotatebox{-90}{$\langle 0 |$}} \\
    \lstick{\rotatebox{-90}{$| 1 \rangle$}} &
      &
      \gate[2]{\rotatebox{-90}{$\theta_h$}} &
      &
      \gate[2]{\rotatebox{-90}{$\theta_h$}} &
      &
      \rstick{\rotatebox{-90}{$\langle 1 |$}} \\
    \lstick{\rotatebox{-90}{$| 0 \rangle$}} &
      \gate[2]{\rotatebox{-90}{$U_S$}} &
      &
      \gate[2]{\rotatebox{-90}{$\theta_v$}} &
      &
      \gate[2]{\rotatebox{-90}{$U_S^\dagger$}} &
      \rstick{\rotatebox{-90}{$\langle 0 |$}} \\
    \lstick{\rotatebox{-90}{$| 1 \rangle$}} &
      &
      \gate[1,style={cap_gate}]{\rotatebox{-90}{$\theta_h$}} &
      &
      \gate[1,style={cap_gate}]{\rotatebox{-90}{$\theta_h$}} &
      &
      \rstick{\rotatebox{-90}{$\langle 1 |$}} \\
\end{quantikz}
    }%
    }
  };

  \draw[->, thick] ($(circ.east)+(0.6,-1)$) -- ($(circ.east)+(0.6,1.0)$)
    node[midway, right] {$t$};
\end{tikzpicture}
}
\caption{\textit{Quantum Circuits for the Square Potts Model.} The probability amplitudes obtained from the above circuits are proportional to the partition function of the square Potts model (Figure \ref{fig:potts}). The $U_S$ gate acts as $U_S | 01 \rangle = \frac{1}{\sqrt{2}}(|01\rangle - |10\rangle)$, allowing for singlet state preparation and measurement. When boundary conditions are periodic, additional $U(\theta_h)$ gates are implemented between the first and last qubit at each `horizontal' layer.}
\label{fig:potts_circuit}
\end{figure*}
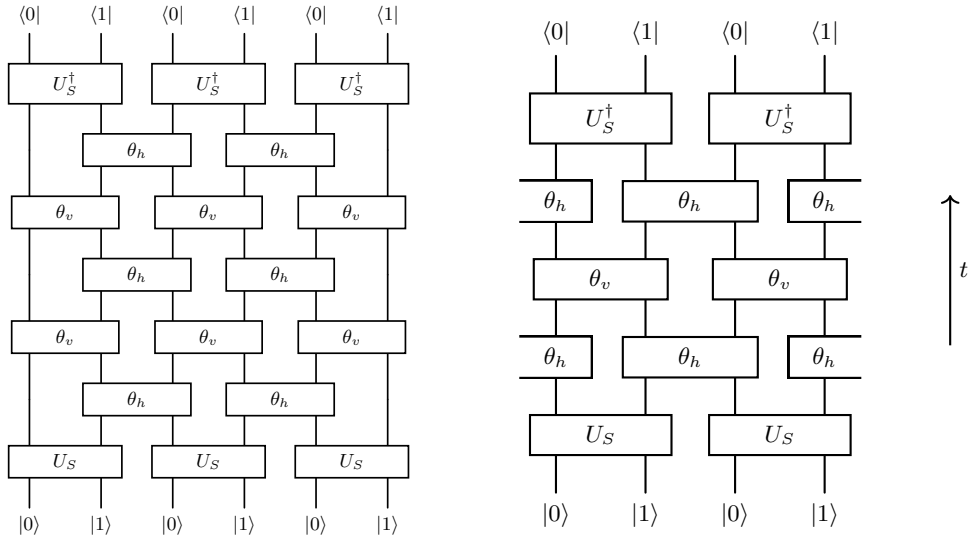

The partition function of the modified six-vertex model is given by a sum over all arrow configurations on $\mathcal{L}'$ consistent with the vertex rules. Each global arrow configuration is downstream of the orientation of the arrows around the top and bottom external vertices, each which can be either clockwise or anti-clockwise. Thus, we may evaluate the partition function by taking a sum over the $2^{2w}$ possible top/bottom orientations, where $w$ is the width of $\mathcal{L}$ (i.e. the number of vertices in each row), followed by a sum over the remaining edge arrow assignments consistent with the top/bottom orientations. We can incorporate the weights of the top and bottom external vertices into the first sum. Notice that with regards to computational basis states, the top and bottom orientations are given by:

\begin{align}
    \left( \begin{tikzpicture}[scale=0.5,baseline={(0,0.15)}, rotate=-45]
  \draw (-1,0)--(0,0);
  \draw (0,0)--(0,1);
  \draw[camblue, line width=3pt, rounded corners=8pt] (0,0) -- (0,1);
\end{tikzpicture} + \begin{tikzpicture}[scale=0.5,baseline={(0,0.15)}, rotate=-45]
  \draw (-1,0)--(0,0);
  \draw (0,0)--(0,1);
  \draw[camblue, line width=3pt, rounded corners=8pt] (-1,0) -- (0,0);
\end{tikzpicture} \right) &= -i (|01 \rangle - |10 \rangle) = -i \sqrt{2} |S \rangle, \\ \left(
\begin{tikzpicture}[scale=0.5,baseline={(0,-0.25)}, rotate=-45]
  \draw (1,0)--(0,0);
  \draw (0,0)--(0,-1);
  \draw[camblue, line width=3pt, rounded corners=8pt] (1,0) -- (0,0);
\end{tikzpicture} + \begin{tikzpicture}[scale=0.5,baseline={(0,-0.25)}, rotate=-45]
  \draw (1,0)--(0,0);
  \draw (0,0)--(0,-1);
  \draw[camblue, line width=3pt, rounded corners=8pt] (0,0) -- (0,-1);
\end{tikzpicture} \right) &= -i (\langle 01 | - \langle 10 |) = -i \sqrt{2} \langle S |.
\end{align}

\noindent Therefore, state preparation and measurement of singlets can be used to implement the sum over external vertices. We also need a sequence of gates $U(\theta_h)$, $U(\theta_v)$ to implement the internal vertices. For the horizontal vertices, we can match the vertex weights by factoring out $e^{i \theta_h}$ to obtain:
\begin{align}
    a_h & = e^{i \theta_h} e^{-i \theta_h} = 1, \\
    b_h & = i \sin \theta_h e^{-i \theta_h} = x_h, \\
    c_h & = \cos \theta_h e^{-i \theta_h} = 1 - x_h.
\end{align}

\noindent As $x_h = -v_h /2$, this lets us encode a Potts model weight of $v_h = -2 i \sin \theta_h e^{-i \theta_h} = e^{-2 i \theta_h} - 1$, i.e. $K_h = -2 i \theta_h$. Similarly, for the vertical vertices we can factor out $i \sin \theta_v$ to obtain:
\begin{align}
    a_v & = e^{i \theta_v} / i \sin \theta_v = x_v, \\
    b_v & = i\sin \theta_v / i \sin \theta_v = 1, \\
    c_v & = \cos \theta_v / i \sin \theta_v = x_v - 1.
\end{align}

\noindent Which gives us $v_v = 2 i e^{i \theta_v} / \sin \theta_v = 2i \cot \theta_v - 2$, giving $K_v = \log ( 2i \cot \theta_v - 1)$. 

It can be shown \citep{Baxter_1982} that the partition function of the modified six-vertex model is related to the original Potts partition function as follows:
\begin{equation}
    Z_{\mathrm{Potts}} = q^{N/2} Z_{6V'},
\end{equation}

\noindent where $N$ is the number of vertices in $\mathcal{L}$. Constructing a quantum circuit consisting of gates $U(\theta_h)$ and $U(\theta_v)$ as shown in Figure \ref{fig:potts_circuit}, we obtain probability amplitudes proportional to $Z_{6V'}$, and thus to $Z_{\mathrm{Potts}}$. Overall:

\begin{equation}
    Z_{\mathrm{Potts}} = \frac{(-2)^{N + w}}{(i \sin \theta_v)^{N_v} e^{i N_h \theta_h}} \langle S |^{\otimes w} U | S \rangle^{\otimes w},
\end{equation}

\noindent where $N_h$ and $N_v$ are the number of horizontal and vertical vertices in the medial lattice $\mathcal{L}'$, $U$ is the $\mathsf{XQP}$ circuit obtained by replacing each internal vertex with the corresponding gates $U(\theta_h)$ and $U(\theta_v)$ (Figure \ref{fig:potts_circuit}), and $w = 2n$ is the width of $\mathcal{L}$. We may also set $\theta_v = \theta_h = \theta$, in which case we obtain $x_h x_v = v_h v_v / 4 = 1$. This corresponds to self-dual points of the Potts model. Notably, we always require that $x_v \neq x_h$, as the solution to $x_v = x_h$ is $\theta = \pi/2$, which gives a trivial circuit with $\mathrm{SWAP}$ gates only.

The above method easily extends to the case where the original Potts model has periodic boundary conditions in one direction, connecting external edges on the sides of $\mathcal{L}'$ into internal edges, as illustrated in Figure \ref{fig:potts}. As explained in \citep{Baxter_1976,Jacobsen_2013}, this can be done by drawing a `seam', denoting the edges along which the sides of the lattice connect. For each arrow crossing through the seam to the left or right, one multiplies each term in the partition function by $s^{\pm 4} = -1$. As the number of crossing arrows is always even, this leaves each term in the partition function identical, but requires additional gates $U(\theta_h)$ acting between the first and last qubits of the circuit. These can be implemented by conjugating $U(\theta_h)$ with a network of $\mathrm{SWAP}$ gates.

We now show that the two formulations of the vertex weights (Equations \eqref{eq:vx_weights1} and \eqref{eq:vx_weights2}) correspond to the same quantum circuit (i.e. Figure \ref{fig:potts_circuit}). For simplicity, assume $\theta_h = \theta_v = \theta$ and $n=4$. Since we can only apply gates of the form $U(\theta)$ (for which $a = b + c$) the implementable vertex weights are (up to prefactors) $a = 1$, $b = x$, $c = 1-x$. In the first case (i.e. Equations \eqref{eq:vx_weights1} and \eqref{eq:vx_weights1a}) we can prepare superpositions of the top and bottom external vertex configurations with the correct phase differences using the state $U(\pi/4)|01\rangle = (|01\rangle + i|10\rangle)/\sqrt{2}$, and implement the weights of \eqref{eq:vx_weights1} with the gate $ (\mathbf{1} \otimes S) U(\theta) (S^\dagger \otimes \mathbf{1} ) $. The external vertex weights on the side are implemented (up to global phase) by multiplying the leftmost and rightmost qubits by $S, S^\dagger$. Thus we have a circuit with a top and bottom layer of $U(\pi/4) \otimes U(\pi/4)$ gates and alternating layers of $S \otimes [(\mathbf{1} \otimes S) U(\theta) (S^\dagger \otimes \mathbf{1} )] \otimes S^\dagger$, and $[(\mathbf{1} \otimes S) U(\theta) (S^\dagger \otimes \mathbf{1} )] \otimes [(\mathbf{1} \otimes S) U(\theta) (S^\dagger \otimes \mathbf{1} )]$. Multiplying them out, all intermediate $S, S^\dagger$ gates cancel, and we are left with a circuit with a bottom layer of $(\mathbf{1} \otimes S^\dag \otimes \mathbf{1} \otimes S^\dag)(U(\pi/4) \otimes U(\pi/4))$, layers of $U(\theta)$ gates, and a top layer of $(U(\pi/4) \otimes U(\pi/4))(S \otimes \mathbf{1} \otimes S \otimes \mathbf{1})$. Setting $U_S = (\mathbf{1} \otimes S^\dagger) U(\pi/4)$, up to a global phase this is the same circuit as the one depicted in Figure \ref{fig:potts_circuit}, which implements the vertex weights of \eqref{eq:vx_weights2}. The same cancellations follow with boundary conditions present. 

The discussion above should be enough to convince the reader that for our choice of parameters ($\Delta = 1$), the modified six-vertex model for encoding the Potts model translates to $\mathsf{XQP}$ circuits which must either contain $S$ gates (synthesisable via phase gadgets), or contain two-qubit gates with matrix elements of the form in Equation \eqref{eq:vx_weights1}. In other words, the required modifications of the six-vertex model to encode $Z_{\mathrm{Potts}}$ correspond to access to $S$ gates on the circuit level (i.e. to implement external vertex weights), which cannot be commuted out from the circuit. We thus expect that unlike in the case of $\mathsf{IQP}$, the $\mathsf{XQP}$ model cannot \textit{exactly} encode the partition function of the Potts model. However, with phase gadgets to implement $S$ gates to a high probability, an approximate encoding is possible.

\begin{remark}
  Without access to singlet states (via $S$ gates or otherwise), $\mathsf{XQP}$ circuits are not able to exactly encode the partition function of the Potts model.
\end{remark}

\bibliography{bibliography.bib}

@article{Setiawan_2014,
   title={Robust two-qubit gates for exchange-coupled qubits},
   volume={89},
   ISSN={1550-235X},
   url={http://dx.doi.org/10.1103/PhysRevB.89.085314},
   DOI={10.1103/physrevb.89.085314},
   number={8},
   journal={Physical Review B},
   publisher={American Physical Society (APS)},
   author={Setiawan, F. and Hui, Hoi-Yin and Kestner, J. P. and Wang, Xin and Sarma, S. Das},
   year={2014},
   month=feb }

@article{Shi_2012,
   title={Fast Hybrid Silicon Double-Quantum-Dot Qubit},
   volume={108},
   ISSN={1079-7114},
   url={http://dx.doi.org/10.1103/PhysRevLett.108.140503},
   DOI={10.1103/physrevlett.108.140503},
   number={14},
   journal={Physical Review Letters},
   publisher={American Physical Society (APS)},
   author={Shi, Zhan and Simmons, C. B. and Prance, J. R. and Gamble, John King and Koh, Teck Seng and Shim, Yun-Pil and Hu, Xuedong and Savage, D. E. and Lagally, M. G. and Eriksson, M. A. and Friesen, Mark and Coppersmith, S. N.},
   year={2012},
   month=apr }

@article{Childs2010, volume={11},
   title={Characterization of universal two-qubit Hamiltonians},
   author={Childs, Andrew M. and Leung, Debbie W. and Mančinska, Laura and Ozols, Maris},
   ISSN={1533-7146},
   url={http://dx.doi.org/10.26421/QIC11.1-2},
   DOI={10.26421/qic11.1-2},
   number={1 & 2},
   journal={Quantum Information and Computation},
   publisher={Rinton Press},
   year={2011},
   month=jan }

@misc{bouland2016,
      title={Complexity classification of two-qubit commuting hamiltonians}, 
      author={Adam Bouland and Laura Mančinska and Xue Zhang},
      year={2016},
      eprint={1602.04145},
      archivePrefix={arXiv},
      primaryClass={quant-ph},
      url={https://arxiv.org/abs/1602.04145}, 
}

@article{Bremner_2010,
   title={Classical simulation of commuting quantum computations implies collapse of the polynomial hierarchy},
   volume={467},
   ISSN={1471-2946},
   url={http://dx.doi.org/10.1098/rspa.2010.0301},
   DOI={10.1098/rspa.2010.0301},
   number={2126},
   journal={Proceedings of the Royal Society A: Mathematical, Physical and Engineering Sciences},
   publisher={The Royal Society},
   author={Bremner, Michael J. and Jozsa, Richard and Shepherd, Dan J.},
   year={2010},
   month=aug, pages={459–472} }

@article{Kempe_2001,
   title={Theory of decoherence-free fault-tolerant universal quantum computation},
   volume={63},
   ISSN={1094-1622},
   url={http://dx.doi.org/10.1103/PhysRevA.63.042307},
   DOI={10.1103/physreva.63.042307},
   number={4},
   journal={Physical Review A},
   publisher={American Physical Society (APS)},
   author={Kempe, J. and Bacon, D. and Lidar, D. A. and Whaley, K. B.},
   year={2001},
   month=mar 
}

@misc{hulse2024,
      title={A framework for semi-universality: Semi-universality of 3-qudit SU(d)-invariant gates}, 
      author={Austin Hulse and Hanqing Liu and Iman Marvian},
      year={2024},
      eprint={2407.21249},
      archivePrefix={arXiv},
      primaryClass={quant-ph},
      url={https://arxiv.org/abs/2407.21249},
}

@misc{Liu_2024,
      title={Unitary Designs from Random Symmetric Quantum Circuits}, 
      author={Hanqing Liu and Austin Hulse and Iman Marvian},
      year={2024},
      eprint={2408.14463},
      archivePrefix={arXiv},
      primaryClass={quant-ph},
      url={https://arxiv.org/abs/2408.14463}, 
}

@article{Mitsuhashi_2025,
   title={Unitary Designs of Symmetric Local Random Circuits},
   volume={134},
   ISSN={1079-7114},
   url={http://dx.doi.org/10.1103/PhysRevLett.134.180404},
   DOI={10.1103/physrevlett.134.180404},
   number={18},
   journal={Physical Review Letters},
   publisher={American Physical Society (APS)},
   author={Mitsuhashi, Yosuke and Suzuki, Ryotaro and Soejima, Tomohiro and Yoshioka, Nobuyuki},
   year={2025},
   month=may }

@article{Havlicek_2018,
   title={Quantum Schur Sampling Circuits can be Strongly Simulated},
   volume={121},
   ISSN={1079-7114},
   url={http://dx.doi.org/10.1103/PhysRevLett.121.060505},
   DOI={10.1103/physrevlett.121.060505},
   number={6},
   journal={Physical Review Letters},
   publisher={American Physical Society (APS)},
   author={Havlíček, Vojtěch and Strelchuk, Sergii},
   year={2018},
   month=aug }

@article{Marvian_2022,
   title={Restrictions on realizable unitary operations imposed by symmetry and locality},
   volume={18},
   ISSN={1745-2481},
   url={http://dx.doi.org/10.1038/s41567-021-01464-0},
   DOI={10.1038/s41567-021-01464-0},
   number={3},
   journal={Nature Physics},
   publisher={Springer Science and Business Media LLC},
   author={Marvian, Iman},
   year={2022},
   month=jan, pages={283–289} }

@article{Van_den_Nest_2009,
   title={Quantum algorithms for spin models and simulable gate sets for quantum computation},
   volume={80},
   ISSN={1094-1622},
   url={http://dx.doi.org/10.1103/PhysRevA.80.052334},
   DOI={10.1103/physreva.80.052334},
   number={5},
   journal={Physical Review A},
   publisher={American Physical Society (APS)},
   author={Van den Nest, M. and Dür, W. and Raussendorf, R. and Briegel, H. J.},
   year={2009},
   month=nov }

@article{De_las_Cuevas_2011,
   title={Quantum algorithms for classical lattice models},
   volume={13},
   ISSN={1367-2630},
   url={http://dx.doi.org/10.1088/1367-2630/13/9/093021},
   DOI={10.1088/1367-2630/13/9/093021},
   number={9},
   journal={New Journal of Physics},
   publisher={IOP Publishing},
   author={De las Cuevas, G and Dür, W and Van den Nest, M and Martin-Delgado, M A},
   year={2011},
   month=sep, pages={093021} }

@article{Izergin_1992,
doi = {10.1088/0305-4470/25/16/010},
url = {https://doi.org/10.1088/0305-4470/25/16/010},
year = {1992},
month = {aug},
publisher = {},
volume = {25},
number = {16},
pages = {4315},
author = {A G Izergin and D A Coker and V E Korepin},
title = {Determinant formula for the six-vertex model},
journal = {Journal of Physics A: Mathematical and General},
}

@misc{Aaronson_2016,
      title={The Computational Complexity of Ball Permutations}, 
      author={Scott Aaronson and Adam Bouland and Greg Kuperberg and Saeed Mehraban},
      year={2016},
      eprint={1610.06646},
      archivePrefix={arXiv},
      primaryClass={quant-ph},
      url={https://arxiv.org/abs/1610.06646}, 
}

@article{ZanardiZalkaFaoro2000,
  title = {Entangling power of quantum evolutions},
  author = {Zanardi, Paolo and Zalka, Christof and Faoro, Lara},
  journal = {Phys. Rev. A},
  volume = {62},
  issue = {3},
  pages = {030301},
  numpages = {4},
  year = {2000},
  month = {Aug},
  publisher = {American Physical Society},
  doi = {10.1103/PhysRevA.62.030301},
  url = {https://link.aps.org/doi/10.1103/PhysRevA.62.030301}
}

@book{Baxter_1982,
  title={Exactly Solved Models in Statistical Mechanics},
  author={Baxter, Rodney J.},
  url={https://physics.anu.edu.au/research/ftp/mpg/baxter_book.php},
  year={1982},
  publisher={Academic Press}
}

@inproceedings{Kasteleyn_1969,
  title={Phase transitions in lattice systems with random local properties},
  author={P. W. Kasteleyn and C. M. Fortuin},
  year={1969},
}

@article{Fortuin_1972,
title = {On the random-cluster model: I. Introduction and relation to other models},
journal = {Physica},
volume = {57},
number = {4},
pages = {536-564},
year = {1972},
issn = {0031-8914},
doi = {https://doi.org/10.1016/0031-8914(72)90045-6},
url = {https://www.sciencedirect.com/science/article/pii/0031891472900456},
author = {C.M. Fortuin and P.W. Kasteleyn},
}

@article{Baxter_1976,
doi = {10.1088/0305-4470/9/3/009},
url = {https://doi.org/10.1088/0305-4470/9/3/009},
year = {1976},
month = {mar},
publisher = {},
volume = {9},
number = {3},
pages = {397},
author = {R J Baxter and S B Kelland and F Y Wu},
title = {Equivalence of the Potts model or Whitney polynomial with an ice-type model},
journal = {Journal of Physics A: Mathematical and General},
}

@misc{Galanis_2021,
      title={The complexity of approximating the complex-valued Potts model}, 
      author={Andreas Galanis and Leslie Ann Goldberg and Andrés Herrera-Poyatos},
      year={2021},
      eprint={2005.01076},
      archivePrefix={arXiv},
      primaryClass={cs.CC},
      url={https://arxiv.org/abs/2005.01076}, 
}

@article{Fujii_2017,
   title={Commuting quantum circuits and complexity of Ising partition functions},
   volume={19},
   ISSN={1367-2630},
   url={http://dx.doi.org/10.1088/1367-2630/aa5fdb},
   DOI={10.1088/1367-2630/aa5fdb},
   number={3},
   journal={New Journal of Physics},
   publisher={IOP Publishing},
   author={Fujii, Keisuke and Morimae, Tomoyuki},
   year={2017},
   month=mar, pages={033003} }

@article{Bremner_2016,
   title={Average-Case Complexity Versus Approximate Simulation of Commuting Quantum Computations},
   volume={117},
   ISSN={1079-7114},
   url={http://dx.doi.org/10.1103/PhysRevLett.117.080501},
   DOI={10.1103/physrevlett.117.080501},
   number={8},
   journal={Physical Review Letters},
   publisher={American Physical Society (APS)},
   author={Bremner, Michael J. and Montanaro, Ashley and Shepherd, Dan J.},
   year={2016},
   month=aug }

@misc{Goldberg_2017,
      title={The complexity of approximating complex-valued Ising and Tutte partition functions}, 
      author={Leslie Ann Goldberg and Heng Guo},
      year={2017},
      eprint={1409.5627},
      archivePrefix={arXiv},
      primaryClass={cs.CC},
      url={https://arxiv.org/abs/1409.5627}, 
}

@misc{Aaronson_2010,
      title={The Computational Complexity of Linear Optics}, 
      author={Scott Aaronson and Alex Arkhipov},
      year={2010},
      eprint={1011.3245},
      archivePrefix={arXiv},
      primaryClass={quant-ph},
      url={https://arxiv.org/abs/1011.3245}, 
}

@article{Aaronson_2011,
   title={A linear-optical proof that the permanent is \#P-hard},
   volume={467},
   ISSN={1471-2946},
   url={http://dx.doi.org/10.1098/rspa.2011.0232},
   DOI={10.1098/rspa.2011.0232},
   number={2136},
   journal={Proceedings of the Royal Society A: Mathematical, Physical and Engineering Sciences},
   publisher={The Royal Society},
   author={Aaronson, Scott},
   year={2011},
   month=jul, pages={3393-3405} }

@article{Morimae_2017,
   title={Hardness of classically sampling the one-clean-qubit model with constant total variation distance error},
   volume={96},
   ISSN={2469-9934},
   url={http://dx.doi.org/10.1103/PhysRevA.96.040302},
   DOI={10.1103/physreva.96.040302},
   number={4},
   journal={Physical Review A},
   publisher={American Physical Society (APS)},
   author={Morimae, Tomoyuki},
   year={2017},
   month=oct }

@misc{Cai_2017,
      title={Complexity Classification Of The Six-Vertex Model}, 
      author={Jin-Yi Cai and Zhiguo Fu and Mingji Xia},
      year={2017},
      eprint={1702.02863},
      archivePrefix={arXiv},
      primaryClass={cs.CC},
      url={https://arxiv.org/abs/1702.02863}, 
}

@misc{Cai_2017b,
      title={Approximability of the Six-vertex Model}, 
      author={Jin-Yi Cai and Tianyu Liu and Pinyan Lu},
      year={2017},
      eprint={1712.05880},
      archivePrefix={arXiv},
      primaryClass={cs.CC},
      url={https://arxiv.org/abs/1712.05880}, 
}

@article{Wu_2003,
   title={Universal quantum computation using exchange interactions and measurements of single- and two-spin observables},
   volume={67},
   ISSN={1094-1622},
   url={http://dx.doi.org/10.1103/PhysRevA.67.050303},
   DOI={10.1103/physreva.67.050303},
   number={5},
   journal={Physical Review A},
   publisher={American Physical Society (APS)},
   author={Wu, Lian-Ao and Lidar, Daniel A.},
   year={2003},
   month=may }

@article{Temperley_1971,
    author = {Temperley, H. N. V. and Lieb, Elliott H},
    title = {Relations between the ‘percolation’ and ‘colouring’ problem and other graph-theoretical problems associated with regular planar lattices: some exact results for the ‘percolation’ problem},
    journal = {Proceedings of the Royal Society of London. A. Mathematical and Physical Sciences},
    volume = {322},
    number = {1549},
    pages = {251-280},
    year = {1971},
    month = {04},
    issn = {0080-4630},
    doi = {10.1098/rspa.1971.0067},
    url = {https://doi.org/10.1098/rspa.1971.0067},
    eprint = {https://royalsocietypublishing.org/rspa/article-pdf/322/1549/251/59283/rspa.1971.0067.pdf},
}

@misc{Jacobsen_2013,
 title={Algèbres, Intégrabilité et Modèles Exactement Solubles}, 
 author={Jesper Lykke Jacobsen and Yacine Ikhlef},
 year={2013},
 url={https://www.phys.ens.psl.eu/~jacobsen/AIMES/AIMES-complete-lecture-notes.pdf},
}

@misc{Jacobs_2024,
      title={The Space Just Above One Clean Qubit}, 
      author={Dale Jacobs and Saeed Mehraban},
      year={2024},
      eprint={2410.08051},
      archivePrefix={arXiv},
      primaryClass={quant-ph},
      url={https://arxiv.org/abs/2410.08051}, 
}

@article{Hangleiter_2023,
   title={Computational advantage of quantum random sampling},
   volume={95},
   ISSN={1539-0756},
   url={http://dx.doi.org/10.1103/RevModPhys.95.035001},
   DOI={10.1103/revmodphys.95.035001},
   number={3},
   journal={Reviews of Modern Physics},
   publisher={American Physical Society (APS)},
   author={Hangleiter, Dominik and Eisert, Jens},
   year={2023},
   month=jul }

@misc{Fefferman_2015,
      title={The Power of Quantum Fourier Sampling}, 
      author={Bill Fefferman and Chris Umans},
      year={2015},
      eprint={1507.05592},
      archivePrefix={arXiv},
      primaryClass={cs.CC},
      url={https://arxiv.org/abs/1507.05592}, 
}

@misc{Vyalyi_2003,
      title={Hardness of approximating the weight enumerator of a binary linear code}, 
      author={M. N. Vyalyi},
      year={2003},
      eprint={cs/0304044},
      archivePrefix={arXiv},
      primaryClass={cs.CC},
      url={https://arxiv.org/abs/cs/0304044}, 
}

@article{Fujii_2018,
   title={Impossibility of Classically Simulating One-Clean-Qubit Model with Multiplicative Error},
   volume={120},
   ISSN={1079-7114},
   url={http://dx.doi.org/10.1103/PhysRevLett.120.200502},
   DOI={10.1103/physrevlett.120.200502},
   number={20},
   journal={Physical Review Letters},
   publisher={American Physical Society (APS)},
   author={Fujii, Keisuke and Kobayashi, Hirotada and Morimae, Tomoyuki and Nishimura, Harumichi and Tamate, Shuhei and Tani, Seiichiro},
   year={2018},
   month=may }

@article{DiVincenzo_2000,
   title={Universal quantum computation with the exchange interaction},
   volume={408},
   ISSN={1476-4687},
   url={http://dx.doi.org/10.1038/35042541},
   DOI={10.1038/35042541},
   number={6810},
   journal={Nature},
   publisher={Springer Science and Business Media LLC},
   author={DiVincenzo, D. P. and Bacon, D. and Kempe, J. and Burkard, G. and Whaley, K. B.},
   year={2000},
   month=nov, pages={339–342} }

@misc{Van_Meter_2021,
      title={Universality of swap for qudits: a representation theory approach}, 
      author={James R. van Meter},
      year={2021},
      eprint={2103.12303},
      archivePrefix={arXiv},
      primaryClass={quant-ph},
      url={https://arxiv.org/abs/2103.12303}, 
}

@article{van_Meter_2019,
   title={Approximate exchange-only entangling gates for the three-spin decoherence-free subsystem},
   volume={99},
   ISSN={2469-9934},
   url={http://dx.doi.org/10.1103/PhysRevA.99.042331},
   DOI={10.1103/physreva.99.042331},
   number={4},
   journal={Physical Review A},
   publisher={American Physical Society (APS)},
   author={van Meter, James R. and Knill, Emanuel},
   year={2019},
   month=apr }

@misc{Kempe_2001b,
      title={Encoded Universality from a Single Physical Interaction}, 
      author={J. Kempe and D. Bacon and D. P. DiVincenzo and K. B. Whaley},
      year={2001},
      eprint={quant-ph/0112013},
      archivePrefix={arXiv},
      primaryClass={quant-ph},
      url={https://arxiv.org/abs/quant-ph/0112013}, 
}

@article{Woodworth_2006,
   title={Few-body spin couplings and their implications for universal quantum computation},
   volume={18},
   ISSN={1361-648X},
   url={http://dx.doi.org/10.1088/0953-8984/18/21/S02},
   DOI={10.1088/0953-8984/18/21/s02},
   number={21},
   journal={Journal of Physics: Condensed Matter},
   publisher={IOP Publishing},
   author={Woodworth, Ryan and Mizel, Ari and Lidar, Daniel A},
   year={2006},
   month=may, pages={S721–S744} }

@misc{Bacon_2003,
      title={Decoherence, Control, and Symmetry in Quantum Computers}, 
      author={D. Bacon},
      year={2003},
      eprint={quant-ph/0305025},
      archivePrefix={arXiv},
      primaryClass={quant-ph},
      url={https://arxiv.org/abs/quant-ph/0305025}, 
}

@misc{Fong_2011,
      title={Universal Quantum Computation and Leakage Reduction in the 3-Qubit Decoherence Free Subsystem}, 
      author={Bryan H. Fong and Stephen M. Wandzura},
      year={2011},
      eprint={1102.2909},
      archivePrefix={arXiv},
      primaryClass={quant-ph},
      url={https://arxiv.org/abs/1102.2909}, 
}

@article{Zeuch_2014,
   title={Analytic pulse-sequence construction for exchange-only quantum computation},
   volume={90},
   ISSN={1550-235X},
   url={http://dx.doi.org/10.1103/PhysRevB.90.045306},
   DOI={10.1103/physrevb.90.045306},
   number={4},
   journal={Physical Review B},
   publisher={American Physical Society (APS)},
   author={Zeuch, Daniel and Cipri, R. and Bonesteel, N. E.},
   year={2014},
   month=jul }

@article{Lidar_1998,
   title={Decoherence-Free Subspaces for Quantum Computation},
   volume={81},
   ISSN={1079-7114},
   url={http://dx.doi.org/10.1103/PhysRevLett.81.2594},
   DOI={10.1103/physrevlett.81.2594},
   number={12},
   journal={Physical Review Letters},
   publisher={American Physical Society (APS)},
   author={Lidar, D. A. and Chuang, I. L. and Whaley, K. B.},
   year={1998},
   month=sep, pages={2594–2597} }

@article{Bacon_2000,
   title={Universal Fault-Tolerant Quantum Computation on Decoherence-Free Subspaces},
   volume={85},
   ISSN={1079-7114},
   url={http://dx.doi.org/10.1103/PhysRevLett.85.1758},
   DOI={10.1103/physrevlett.85.1758},
   number={8},
   journal={Physical Review Letters},
   publisher={American Physical Society (APS)},
   author={Bacon, D. and Kempe, J. and Lidar, D. A. and Whaley, K. B.},
   year={2000},
   month=aug, pages={1758–1761} }

@misc{Hsieh_2003,
      title={An Explicit Universal Gate-set for Exchange-Only Quantum Computation}, 
      author={M. Hsieh and J. Kempe and S. Myrgren and K. B. Whaley},
      year={2003},
      eprint={quant-ph/0309002},
      archivePrefix={arXiv},
      primaryClass={quant-ph},
      url={https://arxiv.org/abs/quant-ph/0309002}, 
}

@article{West_2010,
   title={High Fidelity Quantum Gates via Dynamical Decoupling},
   volume={105},
   ISSN={1079-7114},
   url={http://dx.doi.org/10.1103/PhysRevLett.105.230503},
   DOI={10.1103/physrevlett.105.230503},
   number={23},
   journal={Physical Review Letters},
   publisher={American Physical Society (APS)},
   author={West, Jacob R. and Lidar, Daniel A. and Fong, Bryan H. and Gyure, Mark F.},
   year={2010},
   month=dec }

@article{Bacon_1999,
   title={Robustness of decoherence-free subspaces for quantum computation},
   volume={60},
   ISSN={1094-1622},
   url={http://dx.doi.org/10.1103/PhysRevA.60.1944},
   DOI={10.1103/physreva.60.1944},
   number={3},
   journal={Physical Review A},
   publisher={American Physical Society (APS)},
   author={Bacon, D. and Lidar, D. A. and Whaley, K. B.},
   year={1999},
   month=sep, pages={1944–1955} }

@inbook{Lidar_2003,
   title={Decoherence-Free Subspaces and Subsystems},
   ISBN={9783540448747},
   ISSN={0075-8450},
   url={http://dx.doi.org/10.1007/3-540-44874-8_5},
   DOI={10.1007/3-540-44874-8_5},
   booktitle={Irreversible Quantum Dynamics},
   publisher={Springer Berlin Heidelberg},
   author={Lidar, Daniel A. and Birgitta Whaley, K.},
   year={2003},
   pages={83–120} }

@book{Lidar_2014,
   title={Quantum Information and Computation for Chemistry},
   author ={Daniel A. Lidar},
   ISBN={9781118742631},
   ISSN={1934-4791},
   url={http://dx.doi.org/10.1002/9781118742631},
   DOI={10.1002/9781118742631},
   journal={Advances in Chemical Physics},
   publisher={Wiley},
   year={2014},
   month=feb }

@misc{Piddock_2015,
      title={The complexity of antiferromagnetic interactions and 2D lattices}, 
      author={Stephen Piddock and Ashley Montanaro},
      year={2015},
      eprint={1506.04014},
      archivePrefix={arXiv},
      primaryClass={quant-ph},
      url={https://arxiv.org/abs/1506.04014}, 
}

@misc{Cubitt_2016,
      title={Complexity classification of local Hamiltonian problems}, 
      author={Toby Cubitt and Ashley Montanaro},
      year={2016},
      eprint={1311.3161},
      archivePrefix={arXiv},
      primaryClass={quant-ph},
      url={https://arxiv.org/abs/1311.3161}, 
}

@misc{Piddock_2018,
      title={Universal qudit Hamiltonians}, 
      author={Stephen Piddock and Ashley Montanaro},
      year={2018},
      eprint={1802.07130},
      archivePrefix={arXiv},
      primaryClass={quant-ph},
      url={https://arxiv.org/abs/1802.07130}, 
}

@article{Cubitt_2018,
author = {Toby S. Cubitt  and Ashley Montanaro  and Stephen Piddock },
title = {Universal quantum Hamiltonians},
journal = {Proceedings of the National Academy of Sciences},
volume = {115},
number = {38},
pages = {9497-9502},
year = {2018},
doi = {10.1073/pnas.1804949115},
URL = {https://www.pnas.org/doi/abs/10.1073/pnas.1804949115},
eprint = {https://www.pnas.org/doi/pdf/10.1073/pnas.1804949115},
}

@article{Lidar_2001a,
   title={Decoherence-free subspaces for multiple-qubit errors. I. Characterization},
   volume={63},
   ISSN={1094-1622},
   url={http://dx.doi.org/10.1103/PhysRevA.63.022306},
   DOI={10.1103/physreva.63.022306},
   number={2},
   journal={Physical Review A},
   publisher={American Physical Society (APS)},
   author={Lidar, Daniel A. and Bacon, Dave and Kempe, Julia and Whaley, K. B.},
   year={2001},
   month=jan }

@article{Lidar_2001b,
   title={Decoherence-free subspaces for multiple-qubit errors. II. Universal, fault-tolerant quantum computation},
   volume={63},
   ISSN={1094-1622},
   url={http://dx.doi.org/10.1103/PhysRevA.63.022307},
   DOI={10.1103/physreva.63.022307},
   number={2},
   journal={Physical Review A},
   publisher={American Physical Society (APS)},
   author={Lidar, Daniel A. and Bacon, Dave and Kempe, Julia and Whaley, K. B.},
   year={2001},
   month=jan }

@article{Lidar_1999,
   title={Concatenating Decoherence-Free Subspaces with Quantum Error Correcting Codes},
   volume={82},
   ISSN={1079-7114},
   url={http://dx.doi.org/10.1103/PhysRevLett.82.4556},
   DOI={10.1103/physrevlett.82.4556},
   number={22},
   journal={Physical Review Letters},
   publisher={American Physical Society (APS)},
   author={Lidar, D. A. and Bacon, D. and Whaley, K. B.},
   year={1999},
   month=may, pages={4556–4559} }

@misc{Dasu_2025,
      title={Order-of-magnitude extension of qubit lifetimes with a decoherence-free subspace quantum error correction code}, 
      author={Shival Dasu and Ben Criger and Cameron Foltz and Justin A. Gerber and Christopher N. Gilbreth and Kevin Gilmore and Craig A. Holliman and Nathan K. Lysne and Alistair. R. Milne and Daichi Okuno and Grahame Vittorini and David Hayes},
      year={2025},
      eprint={2503.22107},
      archivePrefix={arXiv},
      primaryClass={quant-ph},
      url={https://arxiv.org/abs/2503.22107}, 
}

@misc{Kasatkin_2026,
      title={Quantum Error Correction and Dynamical Decoupling: Better Together or Apart?}, 
      author={Victor Kasatkin and Mario Morford-Oberst and Arian Vezvaee and Daniel A. Lidar},
      year={2026},
      eprint={2602.19042},
      archivePrefix={arXiv},
      primaryClass={quant-ph},
      url={https://arxiv.org/abs/2602.19042}, 
}

@article{Dash_2024,
   title={Concatenating quantum error-correcting codes with decoherence-free subspaces and vice versa},
   volume={109},
   ISSN={2469-9934},
   url={http://dx.doi.org/10.1103/PhysRevA.109.062411},
   DOI={10.1103/physreva.109.062411},
   number={6},
   journal={Physical Review A},
   publisher={American Physical Society (APS)},
   author={Dash, Nihar Ranjan and Dutta, Sanjoy and Srikanth, R. and Banerjee, Subhashish},
   year={2024},
   month=jun }

@article{Wu_2002,
   title={Creating Decoherence-Free Subspaces Using Strong and Fast Pulses},
   volume={88},
   ISSN={1079-7114},
   url={http://dx.doi.org/10.1103/PhysRevLett.88.207902},
   DOI={10.1103/physrevlett.88.207902},
   number={20},
   journal={Physical Review Letters},
   publisher={American Physical Society (APS)},
   author={Wu, L.-A. and Lidar, D. A.},
   year={2002},
   month=may }

@article{Weinstein_2023,
   title={Universal logic with encoded spin qubits in silicon},
   volume={615},
   ISSN={1476-4687},
   url={http://dx.doi.org/10.1038/s41586-023-05777-3},
   DOI={10.1038/s41586-023-05777-3},
   number={7954},
   journal={Nature},
   publisher={Springer Science and Business Media LLC},
   author={Weinstein, Aaron J. and Reed, Matthew D. and Jones, Aaron M. and Andrews, Reed W. and Barnes, David and Blumoff, Jacob Z. and Euliss, Larken E. and Eng, Kevin and Fong, Bryan H. and Ha, Sieu D. and Hulbert, Daniel R. and Jackson, Clayton A. C. and Jura, Michael and Keating, Tyler E. and Kerckhoff, Joseph and Kiselev, Andrey A. and Matten, Justine and Sabbir, Golam and Smith, Aaron and Wright, Jeffrey and Rakher, Matthew T. and Ladd, Thaddeus D. and Borselli, Matthew G.},
   year={2023},
   month=feb, pages={817–822} }

@article{Petta_2005,
author = {J. R. Petta  and A. C. Johnson  and J. M. Taylor  and E. A. Laird  and A. Yacoby  and M. D. Lukin  and C. M. Marcus  and M. P. Hanson  and A. C. Gossard },
title = {Coherent Manipulation of Coupled Electron Spins in Semiconductor Quantum Dots},
journal = {Science},
volume = {309},
number = {5744},
pages = {2180-2184},
year = {2005},
doi = {10.1126/science.1116955},
URL = {https://www.science.org/doi/abs/10.1126/science.1116955},
eprint = {https://www.science.org/doi/pdf/10.1126/science.1116955},
   }

@article{Laird_2010,
   title={Coherent spin manipulation in an exchange-only qubit},
   volume={82},
   ISSN={1550-235X},
   url={http://dx.doi.org/10.1103/PhysRevB.82.075403},
   DOI={10.1103/physrevb.82.075403},
   number={7},
   journal={Physical Review B},
   publisher={American Physical Society (APS)},
   author={Laird, E. A. and Taylor, J. M. and DiVincenzo, D. P. and Marcus, C. M. and Hanson, M. P. and Gossard, A. C.},
   year={2010},
   month=aug }

@article{Medford_2013,
   title={Self-consistent measurement and state tomography of an exchange-only spin qubit},
   volume={8},
   ISSN={1748-3395},
   url={http://dx.doi.org/10.1038/nnano.2013.168},
   DOI={10.1038/nnano.2013.168},
   number={9},
   journal={Nature Nanotechnology},
   publisher={Springer Science and Business Media LLC},
   author={Medford, J. and Beil, J. and Taylor, J. M. and Bartlett, S. D. and Doherty, A. C. and Rashba, E. I. and DiVincenzo, D. P. and Lu, H. and Gossard, A. C. and Marcus, C. M.},
   year={2013},
   month=sep, pages={654–659} }

@article{Eng_2015,
author = {Kevin Eng  and Thaddeus D. Ladd  and Aaron Smith  and Matthew G. Borselli  and Andrey A. Kiselev  and Bryan H. Fong  and Kevin S. Holabird  and Thomas M. Hazard  and Biqin Huang  and Peter W. Deelman  and Ivan Milosavljevic  and Adele E. Schmitz  and Richard S. Ross  and Mark F. Gyure  and Andrew T. Hunter },
title = {Isotopically enhanced triple-quantum-dot qubit},
journal = {Science Advances},
volume = {1},
number = {4},
pages = {e1500214},
year = {2015},
doi = {10.1126/sciadv.1500214},
URL = {https://www.science.org/doi/abs/10.1126/sciadv.1500214},
eprint = {https://www.science.org/doi/pdf/10.1126/sciadv.1500214}, }

@article{Madzik_2025,
   title={Operating two exchange-only qubits in parallel},
   volume={647},
   ISSN={1476-4687},
   url={http://dx.doi.org/10.1038/s41586-025-09767-5},
   DOI={10.1038/s41586-025-09767-5},
   number={8091},
   journal={Nature},
   publisher={Springer Science and Business Media LLC},
   author={Mądzik, Mateusz T. and Luthi, Florian and Guerreschi, Gian Giacomo and Mohiyaddin, Fahd A. and Borjans, Felix and Chadwick, Jason D. and Curry, Matthew J. and Ziegler, Joshua and Atanasov, Sarah and Bavdaz, Peter L. and Connors, Elliot J. and Corrigan, J. and Ercan, H. Ekmel and Flory, Robert and George, Hubert C. and Harpt, Benjamin and Henry, Eric and Islam, Mohammad M. and Khammassi, Nader and Keith, Daniel and Lampert, Lester F. and Mladenov, Todor M. and Morris, Randy W. and Nethwewala, Aditi and Neyens, Samuel and Otten, René and Osuna Ibarra, Linda P. and Patra, Bishnu and Pillarisetty, Ravi and Premaratne, Shavindra and Ramsey, Mick and Risinger, Andrew and Rooney, John D. and Savytskyy, Rostyslav and Watson, Thomas F. and Zietz, Otto K. and Matsuura, Anne Y. and Pellerano, Stefano and Bishop, Nathaniel C. and Roberts, Jeanette and Clarke, James S.},
   year={2025},
   month=nov, pages={870–875} }

@article{van_Diepen_2021,
   title={Quantum Simulation of Antiferromagnetic Heisenberg Chain with Gate-Defined Quantum Dots},
   volume={11},
   ISSN={2160-3308},
   url={http://dx.doi.org/10.1103/PhysRevX.11.041025},
   DOI={10.1103/physrevx.11.041025},
   number={4},
   journal={Physical Review X},
   publisher={American Physical Society (APS)},
   author={van Diepen, C. J. and Hsiao, T.-K. and Mukhopadhyay, U. and Reichl, C. and Wegscheider, W. and Vandersypen, L. M. K.},
   year={2021},
   month=nov }

@article{Havlicek_2019,
   title={Classical algorithm for quantum SU(2) Schur sampling},
   volume={99},
   ISSN={2469-9934},
   url={http://dx.doi.org/10.1103/PhysRevA.99.062336},
   DOI={10.1103/physreva.99.062336},
   number={6},
   journal={Physical Review A},
   publisher={American Physical Society (APS)},
   author={Havlíček, Vojtěch and Strelchuk, Sergii and Temme, Kristan},
   year={2019},
   month=jun }

@misc{Glaudell_2021,
      title={Optimal Two-Qubit Circuits for Universal Fault-Tolerant Quantum Computation}, 
      author={Andrew N. Glaudell and Neil J. Ross and Jacob M. Taylor},
      year={2021},
      eprint={2001.05997},
      archivePrefix={arXiv},
      primaryClass={quant-ph},
      url={https://arxiv.org/abs/2001.05997}, 
}

@article{Jozsa_2008,
   title={Matchgates and classical simulation of quantum circuits},
   volume={464},
   ISSN={1471-2946},
   url={http://dx.doi.org/10.1098/rspa.2008.0189},
   DOI={10.1098/rspa.2008.0189},
   number={2100},
   journal={Proceedings of the Royal Society A: Mathematical, Physical and Engineering Sciences},
   publisher={The Royal Society},
   author={Jozsa, Richard and Miyake, Akimasa},
   year={2008},
   month=jul, pages={3089–3106} }

@misc{Gottesman_1998,
      title={The Heisenberg Representation of Quantum Computers}, 
      author={Daniel Gottesman},
      year={1998},
      eprint={quant-ph/9807006},
      archivePrefix={arXiv},
      primaryClass={quant-ph},
      url={https://arxiv.org/abs/quant-ph/9807006}, 
}

@article{Bremner_2017,
   title={Achieving quantum supremacy with sparse and noisy commuting quantum computations},
   volume={1},
   ISSN={2521-327X},
   url={http://dx.doi.org/10.22331/q-2017-04-25-8},
   DOI={10.22331/q-2017-04-25-8},
   journal={Quantum},
   publisher={Verein zur Forderung des Open Access Publizierens in den Quantenwissenschaften},
   author={Bremner, Michael J. and Montanaro, Ashley and Shepherd, Dan J.},
   year={2017},
   month=apr, pages={8} }

@article{Tillmann_2013,
   title={Experimental boson sampling},
   volume={7},
   ISSN={1749-4893},
   url={http://dx.doi.org/10.1038/nphoton.2013.102},
   DOI={10.1038/nphoton.2013.102},
   number={7},
   journal={Nature Photonics},
   publisher={Springer Science and Business Media LLC},
   author={Tillmann, Max and Dakić, Borivoje and Heilmann, René and Nolte, Stefan and Szameit, Alexander and Walther, Philip},
   year={2013},
   month=may, pages={540–544} }

@article{Deutsch_1995,
   title={Universality in quantum computation},
   volume={449},
   ISSN={2053-9177},
   url={http://dx.doi.org/10.1098/rspa.1995.0065},
   DOI={10.1098/rspa.1995.0065},
   number={1937},
   journal={Proceedings of the Royal Society of London. Series A: Mathematical and Physical Sciences},
   publisher={The Royal Society},
   author={Deutsch, David Elieser and Barenco, Adriano and Ekert, Artur},
   year={1995},
   month=jun, pages={669–677} }

@article{LLoyd_1995,
  title = {Almost Any Quantum Logic Gate is Universal},
  author = {Lloyd, Seth},
  journal = {Phys. Rev. Lett.},
  volume = {75},
  issue = {2},
  pages = {346--349},
  numpages = {0},
  year = {1995},
  month = {Jul},
  publisher = {American Physical Society},
  doi = {10.1103/PhysRevLett.75.346},
  url = {https://link.aps.org/doi/10.1103/PhysRevLett.75.346}
}

@misc{Grinberg_2025,
      title={An introduction to the symmetric group algebra}, 
      author={Darij Grinberg},
      year={2025},
      eprint={2507.20706},
      archivePrefix={arXiv},
      primaryClass={math.CO},
      url={https://arxiv.org/abs/2507.20706}, 
}

@misc{Vershik_2005,
      title={A New Approach to the Representation Thoery of the Symmetric Groups. 2}, 
      author={A. M. Vershik and A. Yu. Okounkov},
      year={2005},
      eprint={math/0503040},
      archivePrefix={arXiv},
      primaryClass={math.RT},
      url={https://arxiv.org/abs/math/0503040}, 
}

@book{Polyanskiy_2025, 
place={Cambridge}, 
title={Information Theory: From Coding to Learning}, 
publisher={Cambridge University Press}, 
author={Polyanskiy, Yury and Wu, Yihong}, 
year={2025}}

@article{Hamilton_2017,
   title={Gaussian Boson Sampling},
   volume={119},
   ISSN={1079-7114},
   url={http://dx.doi.org/10.1103/PhysRevLett.119.170501},
   DOI={10.1103/physrevlett.119.170501},
   number={17},
   journal={Physical Review Letters},
   publisher={American Physical Society (APS)},
   author={Hamilton, Craig S. and Kruse, Regina and Sansoni, Linda and Barkhofen, Sonja and Silberhorn, Christine and Jex, Igor},
   year={2017},
   month=Oct }

@misc{antoneli_2014,
      title={Maximal Subgroups of Compact Lie Groups}, 
      author={Fernando Antoneli and Michael Forger and Paola Gaviria},
      year={2014},
      eprint={math/0605784},
      archivePrefix={arXiv},
      primaryClass={math.RA},
      url={https://arxiv.org/abs/math/0605784}, 
}

@misc{dawson_2005,
      title={The Solovay-Kitaev algorithm}, 
      author={Christopher M. Dawson and Michael A. Nielsen},
      year={2005},
      eprint={quant-ph/0505030},
      archivePrefix={arXiv},
      primaryClass={quant-ph},
      url={https://arxiv.org/abs/quant-ph/0505030}, 
}

@article{Marvian_2024,
   title={Rotationally Invariant Circuits: Universality with the Exchange Interaction and Two Ancilla Qubits},
   volume={132},
   ISSN={1079-7114},
   url={http://dx.doi.org/10.1103/PhysRevLett.132.130201},
   DOI={10.1103/physrevlett.132.130201},
   number={13},
   journal={Physical Review Letters},
   publisher={American Physical Society (APS)},
   author={Marvian, Iman and Liu, Hanqing and Hulse, Austin},
   year={2024},
   month=Mar }

@misc{yordanov_2020,
      title={Invariant subspaces of two-qubit quantum gates and their application in the verification of quantum computers}, 
      author={Yordan S. Yordanov and Jacob Chevalier-Drori and Thierry Ferrus and Matthew Applegate and Crispin H. W. Barnes},
      year={2020},
      eprint={2009.03642},
      archivePrefix={arXiv},
      primaryClass={quant-ph},
      url={https://arxiv.org/abs/2009.03642}, 
}

@article{Rosenberg_2024,
   title={Dynamics of magnetization at infinite temperature in a Heisenberg spin chain},
   volume={384},
   ISSN={1095-9203},
   url={http://dx.doi.org/10.1126/science.adi7877},
   DOI={10.1126/science.adi7877},
   number={6691},
   journal={Science},
   publisher={American Association for the Advancement of Science (AAAS)},
   author={Rosenberg, E. and Andersen, T. I. and Samajdar, R. and Petukhov, A. and Hoke, J. C. and Abanin, D. and Bengtsson, A. and Drozdov, I. K. and Erickson, C. and Klimov, P. V. and Mi, X. and Morvan, A. and Neeley, M. and Neill, C. and Acharya, R. and Allen, R. and Anderson, K. and Ansmann, M. and Arute, F. and Arya, K. and Asfaw, A. and Atalaya, J. and Bardin, J. C. and Bilmes, A. and Bortoli, G. and Bourassa, A. and Bovaird, J. and Brill, L. and Broughton, M. and Buckley, B. B. and Buell, D. A. and Burger, T. and Burkett, B. and Bushnell, N. and Campero, J. and Chang, H.-S. and Chen, Z. and Chiaro, B. and Chik, D. and Cogan, J. and Collins, R. and Conner, P. and Courtney, W. and Crook, A. L. and Curtin, B. and Debroy, D. M. and Barba, A. Del Toro and Demura, S. and Di Paolo, A. and Dunsworth, A. and Earle, C. and Faoro, L. and Farhi, E. and Fatemi, R. and Ferreira, V. S. and Burgos, L. Flores and Forati, E. and Fowler, A. G. and Foxen, B. and Garcia, G. and Genois, E. and Giang, W. and Gidney, C. and Gilboa, D. and Giustina, M. and Gosula, R. and Dau, A. Grajales and Gross, J. A. and Habegger, S. and Hamilton, M. C. and Hansen, M. and Harrigan, M. P. and Harrington, S. D. and Heu, P. and Hill, G. and Hoffmann, M. R. and Hong, S. and Huang, T. and Huff, A. and Huggins, W. J. and Ioffe, L. B. and Isakov, S. V. and Iveland, J. and Jeffrey, E. and Jiang, Z. and Jones, C. and Juhas, P. and Kafri, D. and Khattar, T. and Khezri, M. and Kieferová, M. and Kim, S. and Kitaev, A. and Klots, A. R. and Korotkov, A. N. and Kostritsa, F. and Kreikebaum, J. M. and Landhuis, D. and Laptev, P. and Lau, K.-M. and Laws, L. and Lee, J. and Lee, K. W. and Lensky, Y. D. and Lester, B. J. and Lill, A. T. and Liu, W. and Locharla, A. and Mandrà, S. and Martin, O. and Martin, S. and McClean, J. R. and McEwen, M. and Meeks, S. and Miao, K. C. and Mieszala, A. and Montazeri, S. and Movassagh, R. and Mruczkiewicz, W. and Nersisyan, A. and Newman, M. and Ng, J. H. and Nguyen, A. and Nguyen, M. and Niu, M. Y. and O'Brien, T. E. and Omonije, S. and Opremcak, A. and Potter, R. and Pryadko, L. P. and Quintana, C. and Rhodes, D. M. and Rocque, C. and Rubin, N. C. and Saei, N. and Sank, D. and Sankaragomathi, K. and Satzinger, K. J. and Schurkus, H. F. and Schuster, C. and Shearn, M. J. and Shorter, A. and Shutty, N. and Shvarts, V. and Sivak, V. and Skruzny, J. and Smith, W. Clarke and Somma, R. D. and Sterling, G. and Strain, D. and Szalay, M. and Thor, D. and Torres, A. and Vidal, G. and Villalonga, B. and Heidweiller, C. Vollgraff and White, T. and Woo, B. W. K. and Xing, C. and Yao, Z. Jamie and Yeh, P. and Yoo, J. and Young, G. and Zalcman, A. and Zhang, Y. and Zhu, N. and Zobrist, N. and Neven, H. and Babbush, R. and Bacon, D. and Boixo, S. and Hilton, J. and Lucero, E. and Megrant, A. and Kelly, J. and Chen, Y. and Smelyanskiy, V. and Khemani, V. and Gopalakrishnan, S. and Prosen, T. and Roushan, P.},
   year={2024},
   month=Apr, pages={48–53} }

@article{Fogarty_2018,
   title={Integrated silicon qubit platform with single-spin addressability, exchange control and single-shot singlet-triplet readout},
   volume={9},
   ISSN={2041-1723},
   url={http://dx.doi.org/10.1038/s41467-018-06039-x},
   DOI={10.1038/s41467-018-06039-x},
   number={1},
   journal={Nature Communications},
   publisher={Springer Science and Business Media LLC},
   author={Fogarty, M. A. and Chan, K. W. and Hensen, B. and Huang, W. and Tanttu, T. and Yang, C. H. and Laucht, A. and Veldhorst, M. and Hudson, F. E. and Itoh, K. M. and Culcer, D. and Ladd, T. D. and Morello, A. and Dzurak, A. S.},
   year={2018},
   month=Oct }

@article{Nowack_2011,
author = {K. C. Nowack  and M. Shafiei  and M. Laforest  and G. E. D. K. Prawiroatmodjo  and L. R. Schreiber  and C. Reichl  and W. Wegscheider  and L. M. K. Vandersypen },
title = {Single-Shot Correlations and Two-Qubit Gate of Solid-State Spins},
journal = {Science},
volume = {333},
number = {6047},
pages = {1269-1272},
year = {2011},
doi = {10.1126/science.1209524},
URL = {https://www.science.org/doi/abs/10.1126/science.1209524},
eprint = {https://www.science.org/doi/pdf/10.1126/science.1209524}}

@article{Abraham_2026,
	author = {{Members of the HRL Quantum Team and Collaborators}},
	date = {2026/07/01},
	doi = {10.1038/s41586-026-10754-7},
	id = {Abraham2026},
	isbn = {1476-4687},
	journal = {Nature},
	number = {8125},
	pages = {1154--1159},
	title = {A digitally controlled silicon quantum processing unit},
	url = {https://doi.org/10.1038/s41586-026-10754-7},
	volume = {655},
	year = {2026}}

@misc{Koehler_2026,
      title={Anticoncentration of the Permanent in Ginibre Ensembles}, 
      author={Frederic Koehler and Pui Kuen Leung},
      year={2026},
      eprint={2607.20329},
      archivePrefix={arXiv},
      primaryClass={math.PR},
      url={https://arxiv.org/abs/2607.20329}, 
}

@article{He_2019,
	author = {He, Y. and Gorman, S. K. and Keith, D. and Kranz, L. and Keizer, J. G. and Simmons, M. Y.},
	date = {2019/07/01},
	doi = {10.1038/s41586-019-1381-2},
	id = {He2019},
	isbn = {1476-4687},
	journal = {Nature},
	number = {7765},
	pages = {371--375},
	title = {A two-qubit gate between phosphorus donor electrons in silicon},
	url = {https://doi.org/10.1038/s41586-019-1381-2},
	volume = {571},
	year = {2019}}

@article{Zhang_2023,
   title={Scalable Multipartite Entanglement Created by Spin Exchange in an Optical Lattice},
   volume={131},
   ISSN={1079-7114},
   url={http://dx.doi.org/10.1103/PhysRevLett.131.073401},
   DOI={10.1103/physrevlett.131.073401},
   number={7},
   journal={Physical Review Letters},
   publisher={American Physical Society (APS)},
   author={Zhang, Wei-Yong and He, Ming-Gen and Sun, Hui and Zheng, Yong-Guang and Liu, Ying and Luo, An and Wang, Han-Yi and Zhu, Zi-Hang and Qiu, Pei-Yue and Shen, Ying-Chao and Wang, Xuan-Kai and Lin, Wan and Yu, Song-Tao and Li, Bin-Chen and Xiao, Bo and Li, Meng-Da and Yang, Yu-Meng and Jiang, Xiao and Dai, Han-Ning and Zhou, You and Ma, Xiongfeng and Yuan, Zhen-Sheng and Pan, Jian-Wei},
   year={2023},
   month=Aug }

@article{Bojovic_2026,
	author = {Bojovi{\'c}, Petar and Hilker, Timon and Wang, Si and Obermeyer, Johannes and Barendregt, Marnix and Tell, Dorothee and Chalopin, Thomas and Preiss, Philipp M. and Bloch, Immanuel and Franz, Titus},
	date = {2026/04/01},
	doi = {10.1038/s41586-026-10356-3},
	id = {Bojovi{\'c}2026},
	isbn = {1476-4687},
	journal = {Nature},
	number = {8110},
	pages = {602--608},
	title = {High-fidelity collisional quantum gates with fermionic atoms},
	url = {https://doi.org/10.1038/s41586-026-10356-3},
	volume = {652},
	year = {2026}}

@misc{Aaronson_2004,
      title={Quantum Computing, Postselection, and Probabilistic Polynomial-Time}, 
      author={Scott Aaronson},
      year={2004},
      eprint={quant-ph/0412187},
      archivePrefix={arXiv},
      primaryClass={quant-ph},
      url={https://arxiv.org/abs/quant-ph/0412187}, 
}

@article{Toda_1991,
author = {Toda, Seinosuke},
title = {PP is as Hard as the Polynomial-Time Hierarchy},
journal = {SIAM Journal on Computing},
volume = {20},
number = {5},
pages = {865-877},
year = {1991},
doi = {10.1137/0220053},
URL = {https://doi.org/10.1137/0220053},
eprint = {https://doi.org/10.1137/0220053}}

\end{document}